\documentclass{article}

\usepackage{arxiv}

\usepackage[utf8]{inputenc} 
\usepackage[T1]{fontenc}    
\usepackage{hyperref}       
\usepackage{url}            
\usepackage{booktabs}       
\usepackage{threeparttable, tablefootnote}
\usepackage{multirow}
\usepackage{amsfonts}       
\usepackage{nicefrac}       
\usepackage{microtype}      
\usepackage{lipsum}         
\usepackage{multicol}       
\usepackage{caption}        
\usepackage{titlecaps}      
\usepackage{syntax}         
\usepackage{graphicx}       
\usepackage{tcolorbox}      
\usepackage{indentfirst}    
\usepackage{fancyvrb}       
\usepackage{tikz}           
\usetikzlibrary{arrows, shapes.gates.logic.US, calc, decorations.markings}
\usepackage{tabularx}       
\usepackage{bussproofs}     
\usepackage{amsthm}         
\usepackage{enumitem}       

\newtheorem{thm}{Theorem}

\newtheorem{lem}[thm]{Lemma}

\newcommand{\titlesc}[1]{\titlecap{\textsc{#1}}}
\newcommand{\GEMINI}{\titlesc{Gemini}}
\newcommand{\SWKIND}{\texttt{*}$_{S}$}
\newcommand{\HWKIND}{\texttt{*}$_{H}$}
\newcommand{\MDKIND}{\texttt{*}$_{M}$}
\newcommand{\SSCONS}{$_{S}\Rightarrow_{S}$}
\newcommand{\HSCONS}{$_{H}\Rightarrow_{S}$}
\newcommand{\HHCONS}{$_{H}\Rightarrow_{H}$}
\newcommand{\SWREC}{\texttt{\{l$_1$: \SWKIND, ..., l$_n$: \SWKIND\}}}
\newcommand{\HWREC}{\texttt{\#\{l$_1$: \HWKIND, ..., l$_n$: \HWKIND\}}}
\newcommand{\nonterm}[1]{$\langle$\textit{#1}$\rangle$}
\newcommand{\ANDGATE}[2]{
    \begin{tikzpicture}
        \node (v0) at (0, 0.9) {#1};
        \node (vdots) at (0, 0.55) {$\vdots$};
        \node (vn) at (0, 0) {#2};
        \node[and gate US, draw, logic gate inputs=nnnnn, scale=1.1] at ($(vdots) + (1.5, -0.1)$) (xand) {};
        \draw (v0) -| (xand.input 1);
        \draw (vn) -| (xand.input 5);
        \draw (xand.output) -- ($(xand) + (1, 0)$);
        \draw ($(xand.input 2) + (-0.5, 0)$) -- (xand.input 2); \draw ($(xand.input 3) + (-0.5, 0)$) -- (xand.input 3);
        \draw ($(xand.input 4) + (-0.5, 0)$) -- (xand.input 4);
    \end{tikzpicture}
}

\newcommand{\ORGATE}[2]{
    \begin{tikzpicture}
        \node (v0) at (0, 0.9) {#1};
        \node (vdots) at (0, 0.55) {$\vdots$};
        \node (vn) at (0, 0) {#2};
        \node[or gate US, draw, logic gate inputs=nnnnn, scale=1.1] at ($(vdots) + (1.5, -0.1)$) (xand) {};
        \draw (v0) -| (xand.input 1);
        \draw (vn) -| (xand.input 5);
        \draw (xand.output) -- ($(xand) + (1, 0)$);
        \draw ($(xand.input 2) + (-0.5, 0)$) -- (xand.input 2); \draw ($(xand.input 3) + (-0.5, 0)$) -- (xand.input 3);
        \draw ($(xand.input 4) + (-0.5, 0)$) -- (xand.input 4);
    \end{tikzpicture}
}

\newcommand{\XORGATE}[2]{
    \begin{tikzpicture}
        \node (v0) at (0, 0.9) {#1};
        \node (vdots) at (0, 0.55) {$\vdots$};
        \node (vn) at (0, 0) {#2};
        \node[xor gate US, draw, scale=2] at ($(vdots) + (1.5, -0.1)$) (xand) {};
        \draw ($(xand.input 1) + (-0.5, 0)$) -- (xand.input 1);
        \draw ($(xand.input 2) + (-0.5, 0)$) -- (xand.input 2);
        \draw ($(xand.west) + (-0.5, 0)$) -- (xand.west);
        \draw ($(xand.north west) + (-0.2, -0.1)$) -- ($(xand.north west) + (0.05, -0.1)$);
        \draw ($(xand.south west) + (-0.2, 0.1)$) -- ($(xand.south west) + (0.05, 0.1)$);
        \draw (xand.output) -- ($(xand) + (1, 0)$);
    \end{tikzpicture}
}

\newcommand{\NOTGATE}[1]{
    \begin{tikzpicture}
        \node (v) at (0, 0) {#1};
        \node[not gate US, draw, scale=1.1] at ($(v) + (1.5, 0)$) (xnot) {};
        \draw (v) -- (xnot.input);
        \draw (xnot.output) -- ($(xnot) + (1, 0)$);
    \end{tikzpicture}
}

\newcommand{\tyrule}[2]{%
    \begin{center}
    \begin{tabularx}{1.15\textwidth}{C R}
    #1 & (\titlesc{#2}) \\
    \end{tabularx}
    \end{center}
}

\fvset{xleftmargin=3mm,numbers=left,fontsize=\small,numbersep=10px}

\newcolumntype{C}{>{\centering\arraybackslash\hspace{0pt}}X}
\newcolumntype{R}{>{\footnotesize\hfill\arraybackslash\hspace{0pt}\hsize=.55\hsize}X}

\newenvironment{Figure}
  {\par\medskip\noindent\minipage{\linewidth}}
  {\endminipage\par\medskip}

\title{Gemini: A Functional Programming Language for Hardware Description}

\author{
  Aditya Srinivasan\\
  Pratt School of Engineering\\
  Duke University '18\\
  Durham, NC 27708 \\
  \texttt{aditya.srinivasan@alumni.duke.edu} \\
   \And
 Andrew D. Hilton \\
  Pratt School of Engineering\\
  Duke University\\
  Durham, NC 27708 \\
  \texttt{adhilton@ee.duke.edu} \\
}

\begin{document}
\maketitle

\begin{abstract}
This paper presents \GEMINI, a functional programming language for hardware description that provides features such as parametric polymorphism, recursive datatypes, higher-order functions, and type inference for higher expressivity compared to modern hardware description languages. \GEMINI\space demonstrates the theory and implementation of novel type-theoretical concepts through its unique type system consisting of multiple atomic kinds and dependent types, which allows the language to model both software and hardware constructs safely and perform type inference through multi-staged compilation. The primary technical results of this paper include formalizations of the \GEMINI\space grammar, typing rules, and evaluation rules, a proof of safety of \GEMINI's type system, and a prototype implementation of the compiler's semantic analysis phase.
\end{abstract}


\begin{multicols}{2}

\section{Introduction}
Over the latter half of the 20th century, methods in electronic circuit design evolved significantly to cope with growing complexity and scale due in part to Moore's Law. Traditional methods of manual design at the transistor level were rendered ineffectual due to limitations at scale, and gate-level descriptions (also known as \textit{netlists}), were used as a higher-level abstraction. A netlist description is an enumeration of the electronic components in a circuit and their connectivity to other nodes in the circuit, and is translated to a physical transistor-level implementation by place-and-route algorithms. Eventually, direct specification of netlists became unfeasible due to differences in implementation across different target hardwares. Hardware description languages, which specify designs at a register-transfer level (RTL) abstraction, were therefore created \cite{ciletti2010,barbacci1973}.
\par
As very-large-scale-integration (VLSI) became increasingly popular, the need for HDLs became more important, and the first modern HDL, Verilog, was introduced between 1983 and 1984 \cite{eetimes2005}. At approximately the same time, the Department of Defense began developing a new standard named VHDL \cite{dod2017,barbacci1984}. These and variant descendent languages have enabled electronic circuit designers to specify highly complex structure and behavior at a high level of abstraction, and leverage synthesis tools to translate the RTL specification to an optimized netlist for a specific target hardware, much in the same way a compiler generates optimized low-level code for a specific target machine. There also exist tools to verify and simulate digital logic circuits written in HDLs.
\par
As HDLs continued to evolve, successive iterations have introduced increasingly powerful features, such as datatypes and strong type systems, borrowed from popular software programming languages. However, development of software languages has outpaced that of HDLs, and they have been continuously improving the programmer's ability to concisely express complex programs.
\par
A logical next step is to merge the advanced ideas from software programming languages with existing hardware description languages, to provide another level of abstraction that equips electronic circuit designers to represent increasingly complex designs. The \GEMINI\space language does so by incorporating features such as parametric polymorphism, type inference, higher-order functions, recursive types, and recursion.
\par
The idea of using software programming languages to generate Verilog or some other HDL is not novel. There exist projects that do so based on popular programming languages: Clash and Lava based on Haskell \cite{baaij2010,lava}, Chisel based on Scala \cite{bachrach2012}, and HML based on SML \cite{hml}. While there exist similarities between these languages and \GEMINI, the \GEMINI\space language models a wider range of software constructs. \GEMINI\space also offers certain abstractions such as the hardware record type, which translates to a bit vector when compiled, but allows the programmer to name and access fields in \GEMINI\space as if the vector were a software record.
\par
Most notably, the design and implementation of \GEMINI\space demonstrates novel concepts in type theory by virtue of its unique type system. This includes a kinding system with multiple atomic kinds, and dependent types. Yet, the compiler is still able to perform type inference by virtue of multi-staged compilation.
\par
In Section \ref{sec:langspec}, we provide a full specification of the \GEMINI\space language. We formalize the type system, software and hardware grammars, typing rules, and evaluation rules. In Section \ref{sec:metatheory}, we provide a proof of a desirable property of our type system, safety, which is bolstered by proofs of progress and preservation. In Section \ref{sec:compiler}, we transition from discussing the language design to the compiler implementation. We provide a description of all phases of the \GEMINI\space compiler, and offer an overview of all phases as well as a prototype implementation of the semantic analysis phase. In Section \ref{sec:example}, we provide examples of some \GEMINI\space programs and the compiled Verilog output. Finally, in Section \ref{sec:future} we discuss the scope for future extensions and improvements towards the \GEMINI\space project.

\section{The \GEMINI\space Language}
\label{sec:langspec}

We will first present a specification of the \GEMINI\space language. \GEMINI\space is inspired conceptually by languages in the ML family, and its syntax in particular is derived mainly from SML. We begin with an informal overview of the key components of the language, namely the type system, values, expressions, declarations, and the core library. We will then formalize these through the presentation of grammars, typing rules, and evaluation rules. Such formalizations are necessary in order to support our metatheory proofs in Section \ref{sec:metatheory}.

\subsection{Types and kinds in \GEMINI}
\label{sec:typesandkinds}
In conventional type systems, there are two primitives from which all kinds are produced: a single atomic kind \texttt{*}\footnote{Pronounced \textit{type}}, and the constructor $\Rightarrow$\footnote{Pronounced \textit{to}}. Proper types, such as \texttt{int}, \texttt{real}, and \texttt{string} belong to kind \texttt{*}. Type constructors, such as \texttt{list} and \texttt{ref}, belong to kind \texttt{*} $\Rightarrow$ \texttt{*}. The latter are not types in their own right, but instead act as functions that accept a type as an argument in order to produces a type, such as an \texttt{int list} or a \texttt{string list ref} \cite{tapl}.
\par
In comparison, \GEMINI\space possesses a type system consisting of three atomic kinds:
\begin{enumerate}
    \item \SWKIND: software type
    \item \HWKIND: hardware type
    \item \MDKIND: module type
\end{enumerate}
The separation of kinds is an important and unique aspect of \GEMINI. The motivation is to enforce a trichotomy between software, hardware, and module types in the type system. Figure \ref{fig:typesys} illustrates the three atomic kinds and their constituent types.
\begin{Figure}
    \textbf{Types in \SWKIND}\vspace{5pt}
    \begin{tcolorbox}
        \texttt{int}\\
        \texttt{real}\\
        \texttt{string}\\
        \SWKIND\space \texttt{list}\\
        \SWREC\\
        \texttt{\SWKIND\space ref}\\
        \texttt{\HWKIND\space sw}\\
        \texttt{C$_i$ \SWKIND}\\
        \texttt{\SWKIND\space $\rightarrow$ \SWKIND}
    \end{tcolorbox}
    \vspace{10pt}
    \textbf{Types in \HWKIND}\vspace{5pt}
    \begin{tcolorbox}
        \texttt{bit}\\
        \texttt{\HWKIND\space [n]}\\
        \texttt{\HWKIND\space @ n}\\
        \HWREC
    \end{tcolorbox}
    \vspace{10pt}
    \textbf{Types in \MDKIND}\vspace{5pt}
    \begin{tcolorbox}
        \texttt{\HWKIND\space $\leadsto$ \HWKIND}
    \end{tcolorbox}
    \captionof{figure}{\GEMINI\space kinds}
    \label{fig:typesys}
\end{Figure}
\par
Further, the constructor $\Rightarrow$ is separated into three distinct variations:
\begin{enumerate}
    \item \SSCONS: sofware-to-software
    \item \HSCONS: hardware-to-software
    \item \HHCONS: hardware-to-hardware
\end{enumerate}
\par
The subscript to the left of the constructor denotes the kind that is accepted as an argument, and the subscript to the right of the constructor denotes the kind that is produced. For example, the type constructors \texttt{list} and \texttt{ref} are constructors of kind $_{S}\Rightarrow_{S}$ as they construct a software type from another software type. The constructor \texttt{sw} on the other hand is a constructor of kind \HSCONS\space as it constructs a software type from a provided hardware type.
\par
There are several interesting aspects of the design of this type system that warrant further elaboration.
\par
\textbf{\textit{Type descriptions.}} First, we will explicitly describe each of the types listed above. In \SWKIND, the primitive types are \texttt{int}, \texttt{real}, and \texttt{string} which represent integers, real numbers, and sequences of characters, respectively. The type constructor \texttt{list} of kind \SSCONS\space accepts a software type and produces a software type representing an ordered collection of elements of that type. The software record is a type constructor of kind \SSCONS\space that accepts an arbitrary number of software type parameters and produces a record with named fields that can be used to access each value. It is worth noting that tuples and records are syntactically distinct, as will be seen in the software grammar in Section \ref{sec:swgrammar}, but are represented identically in the type system; beyond the lexical analysis phase, a tuple is treated as a record with numerical indexes for field labels. The type constructor \texttt{ref} of kind \SSCONS\space accepts a software type and produces a reference container whose inner value can be mutated as in typical imperative programming languages. The type constructor \texttt{sw} is the only type constructor of kind \HSCONS, as it accepts a hardware type and produces a software type. The resulting value acts as a wrapper around a hardware value, which enables the programmer to use it in software-typed constructs such as lists or functions. The wrapper is opaque, in that the underlying hardware value can never be directly accessed or read in software constructs since it is not known what exact value the hardware will take. However, it may be unwrapped to expose the hardware value which may be used in hardware-typed constructs. The type constructor \texttt{C$_i$} of kind \SSCONS\space is a variant in some discriminated union type \texttt{C}. Lastly, the type constructor $\rightarrow$, pronounced "function", of kind \SSCONS\space accepts a software type as the argument and produces another software type. The recursive definition in terms of \SWKIND\space allows for the existence of higher-order functions.
\par
In \HWKIND, the only primitive type is \texttt{bit}, which represents a single bit value. There are three type constructors of kind \HHCONS\space that form the remaining types in \HWKIND. The first is the array type constructor, which accepts a hardware type and produces another hardware type representing a fixed-length ordered array of elements of that type. It is important to note that the array type is only partially defined by the element type; the type is completely defined by the combination of the element type \textit{and} the array size, denoted as \texttt{n} in Figure \ref{fig:typesys}. This point will be elaborated upon further later in this section. The second is the temporal type constructor, which accepts a hardware type and produces another hardware type representing a time-delayed value. As with the array type, the temporal type is completely defined by the combination of the input type \textit{and} the delay amount, denoted as \texttt{n} in \ref{fig:typesys}. Lastly, the hardware record is a type constructor that accepts an arbitrary number of hardware type parameters and produces a record with named fields that can be used to access each value. The hardware-typed record is similar to the software-typed record, with the sole distinction being the kind to which the inner values must belong. Further, the tuple-record duality holds for hardware-typed records as well.
\par
The kind \MDKIND is simple as it contains a single type constructor $\leadsto$, pronounced "module", of kind \HHCONS that accepts a hardware type as the argument and produces another hardware type. The separation of the module type constructor into a separate kind is intentional, and the motivation will be elaborated upon further later in this section.
\par
\textbf{\textit{Asymmetric constructors.}} It is important to note that among the four possible permutations of constructors, only three are realized; the last, $_{S}\Rightarrow_{H}$, does not exist. This is because it is not possible to convert all software-typed values into hardware, so such a constructor would have to be restricted to operating on certain kinds of software types as arguments. A simpler alternative approach is to define functions that convert certain software types into hardware types, or more accurately into software-wrapped hardware values. An example of such a utility function could convert an integer to its 32-bit signed representation as a bit array.
\par
\textbf{\textit{Restrained higher-order types.}} As alluded to earlier, the separation of the module type constructor $\leadsto$ into a third kind is intentional and necessary in order to enforce the expected semantics of hardware. If the module type constructor were to have been defined in \HWKIND, then this would give rise to higher-order modules in which a module could be the input to another module. These semantics do not have any meaning in hardware and so the division of kinds was made to limit the expression of modules to the first-order; in fact, the module type cannot be passed as an argument to any other type constructor.
\par
\textbf{\textit{Dependent types.}} As mentioned previously, an instance of the array and temporal types is defined by the combination of a hardware type \textit{and} an integer value. The hardware type is a type argument that parameterizes these types, and the integer value can similarly be considered as a \textit{value argument}. As one would expect, the type of an array of 8 bits (\texttt{bit[8]}) is distinct from the type of an array of 8 two-entry bit-tuples (\texttt{(bit * bit)[8]}) since recursive comparison of the types reveals that the element types differ. Following the same principle, the type of an array of 8 bits (\texttt{bit[8]}) is distinct from the type of an array of 16 bits (\texttt{bit[16]}) since the array types are dependently typed by the length; recursive comparison of the types reveals that the lengths differ.
\par
An important detail to explain is how semantic analysis can be performed on a type system with dependent types. In order to carry out type-checking, it is necessary to know the types of all values and expressions. Usually, this is not a problem since types require no evaluation, however dependent types must be known in order to determine the complete instantiation of a type constructor. Consider the following snippet of \GEMINI.
\begin{tcolorbox}
\begin{Verbatim}
let
    val size = (* ... *)
    module my_mod (a: bit[8]) = a[:7:]
    val arg = #[size; gen i => 'b:0]
    val h = my_mod arg
in
    (* ... *)
end
\end{Verbatim}
\end{tcolorbox}
We will not dwell on the details of syntax here; however, as an overview, this snippet makes a few declarations. First, on line 2, \texttt{size} is assigned a value by some expression. On line 3, a module \texttt{my_mod} is declared that accepts an 8-bit array and returns the last bit. One line 4, \texttt{arg} is declared using an array generator with \texttt{size} as its size and the zero-bit as the initializing value of elements. Lastly, \texttt{h} is declared as the result of passing\texttt{arg} to the module \texttt{my_mod}. In determining whether this program is well-typed, it is necessary to know the value of \texttt{size}; the program is well-typed if and only if \texttt{size} is equal to \texttt{8}.
\par
The \GEMINI\space compiler makes this kind of analysis possible through multi-staged compilation, an approach taken by compilers of some other languages such as MetaML\cite{metaml}. Typically, the compiler first performs semantic analysis and then, if the program is well-typed, evaluates values. In \GEMINI, the pair of phases is performed twice: first for software-typed values and expressions, and then for hardware-typed values. After the first pass, assuming the program has well-typed software values and expressions, the resulting program consists solely of hardware values whose types are fully known, since all software values have already been evaluated. This allows for the subsequent semantic analysis on hardware values, and ensures that the module produced by a \GEMINI\space program is well-typed in terms of the hardware as well.
\par
Another interesting result of having dependent types is that functions can be written to be parametrically polymorphic on values instead of types. For example, consider the following snippet of \GEMINI.
\begin{tcolorbox}
\begin{Verbatim}
fun negate b = sw !(unsw b)
module map_module a = HW.map negate a
\end{Verbatim}
\end{tcolorbox}
\par
In Section \ref{sec:expr} we will elaborate more on the function of \texttt{sw} and \texttt{unsw} as well as explain how \texttt{HW.map} is actually written, though it is sufficient to know that \texttt{negate} is a function that unwraps a bit, negates it, and rewraps it, and that \texttt{map_module} is a module that takes a bit array \texttt{a} and applies the negating function to each bit. Here, \texttt{map_module} is parametrically polymorphic in the size of the array; its type signature would be \texttt{bit[n] $\leadsto$ bit[n]}. Thus, a bit array of any size may be supplied as an argument. It is important to note that while a \GEMINI\space program must return a module, it cannot be parametrically polymorphic since explicit sizes must be known in order to produce the appropriate Verilog. Thus, the only way \texttt{map_module} can be used is by applying it at some point in the program.
\par
\textbf{\textit{Hardware abstractions.}} The hardware record is an example of how \GEMINI\space makes it easier for the programmer to work with hardware values by introducing a higher level of abstraction. In Verilog, records do not exist since all data is expressed in terms of bits and arrays. As such, at the final phase of compilation hardware records are encoded as arrays by concatenating the values of fields. Similarly, when accessing a field in a record, the correct indices are computed based on the sizes of preceding fields in order to retrieve the appropriate value. As a result, the output Verilog is correctly expressed in terms of native types, and the \GEMINI\space compiler handles the translation to and from this higher level of abstraction.

\subsection{\GEMINI\space expressions}
\label{sec:expr}
Having introduced the \GEMINI\space type system, we now turn our attention to the various expressions one may use in a program. Expressions are different from values, which are atomic units of data that cannot be evaluated any further, in that an expression can be reduced by rules of evaluation to other subexpressions or terminal values. In Section \ref{sec:swgrammar} we will formalize the expression grammar, and in Section \ref{sec:evalrules} we will formalize the rules of evaluation; now we limit ourselves to an overview of the expressions available in the language.

\par
\textbf{\textit{Function application.}} As described in Section \ref{sec:typesandkinds}, a function accepts a software-typed argument and produces a software-typed result. Function application is the invocation of a given function with a given argument in order to produce a result. For parametrically polymorphic functions, the type of the result may depend on the type of the argument. For example, consider the following \GEMINI\space function that constructs a singleton list from an element.
\begin{tcolorbox}
\begin{Verbatim}
fun singleton x = [x]
\end{Verbatim}
\end{tcolorbox}
\par
The type signature of this function is \texttt{`a -> `a list}, where \texttt{`a} is a \textit{type variable} representing any type. As a result, the type of the result will depend on what \texttt{singleton} is supplied as an argument.
\begin{tcolorbox}
\begin{Verbatim}
singleton 1   (* [1]  : int list    *)
singleton "a" (* ["a"]: string list *)
\end{Verbatim}
\end{tcolorbox}
\par
Since \GEMINI\space also supports higher-order functions, partial application (or \textit{currying}) will result in another function.
\begin{tcolorbox}
\begin{Verbatim}
fun add a b = a + b
val add3 = add 3
\end{Verbatim}
\end{tcolorbox}
\par
In the snippet above, \texttt{add} has type signature \texttt{int $\rightarrow$ int $\rightarrow$ int}, whereas \texttt{add3} has type signature \texttt{int $\rightarrow$ int} and will return the result of adding \texttt{3} to the integer to which it is applied. This also demonstrates the first-class status of functions as values, allowing them to be passed as arguments to other functions. A common example of this is the \texttt{map} function which has the type signature \texttt{(`a -> `b) -> `a list -> `b list}.

\par
\textbf{\textit{Operators.}} Operators are analogous to functions, except their application is infix and the arguments are the operands. In \GEMINI, operator overloading is very rare as is the case with many strongly typed languages. As an example, the addition operator for integers is different from that for reals\footnote{For integers the operator is \texttt{+} whereas for reals it is \texttt{+.}}. This has two primary benefits. First, it enables the exact determination of types when performing inference without the need for explicit type decoration. Second, it eliminates implicit type-casting which reduces the potential for certain classes of program bugs and improves performance by avoiding redundant conversions which can be expensive. For an exhaustive list of the operators and their semantic results, consult \ref{sec:hwops} through \ref{sec:lstops} of the Appendix.

\par
\textbf{\textit{Accesses.}} There are several types that contain inner values or expressions that may need to be accessed. In particular, the types are software records (and by extension software tuples), hardware records (and by extension hardware tuples), arrays, and references. Records contain values named by labels, or fields, and inner values can be accessed with the syntax \texttt{\#f r}, where \texttt{f} is the field name and \texttt{r} is the record to access. As previously mentioned, a tuple is a special case of a record with consecutively numbered fields beginning at \texttt{1}, and so tuple accesses are similarly made with the syntax \texttt{\#n t}, where \texttt{n} is an integer literal corresponding to the index and \texttt{t} is the tuple to access. Elements in an array can be accessed with the syntax \texttt{a[:i:]} where \texttt{a} is the array and \texttt{i} is the index, which may be an expression or a value. Since the size of the array and the index are both evaluated fully at the software evaluation stage, it is possible to detect out-of-range accesses at compile-time. Lastly, references can be accessed, or \textit{dereferenced}, with the syntax \texttt{\$r} where \texttt{r} is the reference. The result is the current value inside the reference container.

\par
\textbf{\textit{Conditionals.}} In \GEMINI, there are two forms of conditional expressions to enable control flow: if-then-else and if-then. The former is written with the syntax \texttt{if e1 then e2 else e3}, and the result is \texttt{e2} if the guard \texttt{e1}, which must have an integer type, is nonzero, else it is \texttt{e3}. The latter is a special case of the former, where the \texttt{else} clause implicitly returns the empty tuple, also known as \textit{unit}. In both forms, the types of the expressions in the \texttt{then} and \texttt{else} clauses must match. Since the latter form always has the unit type in the \texttt{else} clause, the expression of the \texttt{then} clause must always have the unit type as well. The return type of conditional expressions must also always be in \SWKIND.

\par
\textbf{\textit{Assignments.}} While dereferencing is used to access the inner value of a reference, assignments are used to mutate the inner value itself. The syntax is \texttt{r := e}, where \texttt{r} is the reference and \texttt{e} is the expression to assign to the inner value. The type of an assignment statement is unit, as is the convention for side-effecting expressions.

\par
\textbf{\textit{Sequences.}} Sequence expressions allow multiple expressions to be evaluated for side-effect, with the final one serving as the return value. The syntax for a sequence is a semi-colon separated list of expressions of arbitrary length surrounded with parentheses.

\par
\textbf{\textit{Pattern-match.}} The pattern-match is ubiquitous in strongly typed languages and serves as another, more general, form of control flow. A pattern-match expression consists of a test expression \texttt{e} and an ordered set of match-result pairs \texttt{(m$_1$, r$_1$), ..., (m$_n$, r$_n$)}. The value of \texttt{e} is compared to each match and for the first \texttt{m$_i$} that matches, the corresponding \texttt{r$_i$} is returned. Pattern-matching can be more expressive than simple conditionals in certain cases as they can be used to inspect the structure of the test expression, such as whether it is a particular variant within a discriminated union type, or how many elements it possesses if it is a list. The match cases must also be exhaustive, or else the compiler will issue a warning; this protects the programmer from certain classes of run-time errors by necessitating the explicit handling of each case.

\par
\textbf{\textit{Let-bindings.}} Let-bindings are used to bind identifiers with values or types and to then evaluate an expression in the augmented lexical scope. If an identifier bound outside of a let-binding is bound again inside, the previous binding is overridden within its scope and the identifier is bound to the most recent value. In general, an identifier is bound to the value from its most recent declaration.

\par
\textbf{\textit{Software-wrapper.}} As mentioned in Section \ref{sec:typesandkinds}, the type constructor \texttt{sw} is the sole member of \HSCONS. The syntax to wrap a hardware value is \texttt{sw h}, where \texttt{h} is a hardware value. The resulting value can be used in software constructs such as functions or lists and eventually unwrapped with the function \texttt{unsw} which has type signature \texttt{`a sw $\rightarrow$ `a}. The code snippet below demonstrates a particularly useful method of using \texttt{sw} and \texttt{unsw}.
\begin{tcolorbox}
\begin{Verbatim}
unsw Array.fromList(
    List.map hwMapFn (
        Array.toList(sw array)
    )
)
\end{Verbatim}
\end{tcolorbox}
\par
The \texttt{HW.map} function introduced in Section \ref{sec:typesandkinds} is actually constructed as shown above in expanded form. We assume here that \texttt{hwMapFn} is a function with type signature \texttt{`a sw $\rightarrow$ `b sw} and \texttt{array} is an array with type \texttt{`a[n]} (note that the element type must match the argument type of \texttt{hwMapFn}, and that the array size is parametrically polymorphic). The use of \texttt{sw} allows us to wrap and pass the array to other functions, which can unwrap it as well as its elements in order to map to new values. At no point is the programmer able to read or use the specific hardware values in any way, though they are able to apply logical operations which are later translated into Verilog.
\par
\textbf{\textit{Array generation.}} A special expression called \textit{array generation} allows for the initialization of hardware arrays. The expression consists of two clauses: the array size and an inline generator function. The array size is specified as an integer-typed expression or value; recall that due to the multi-staged compilation of \GEMINI\space programs, expressions are permitted since they will be evaluated fully prior to hardware type-checking, and so the value argument to the array type constructor will be known. The inline generator function is used to specify the value of each element within the array. Within the body of the generator function, an index variable can be used to determine the element value. For example, the following \GEMINI\space snippet initializes a bit array of size \texttt{8} of alternating \texttt{0}s and \texttt{1}s.
\begin{tcolorbox}
\begin{Verbatim}
#[8; gen i =>
    if (i % 2) = 0
    then 'b:0
    else 'b:1
]
\end{Verbatim}
\end{tcolorbox}

\subsection{\GEMINI\space declarations}
As described previously, let-bindings allow for the binding of values or types to identifiers that can be referenced in other expressions within the lexical scope. There are five types of declarations that can be made. The specific syntax for each will be formalized in Section \ref{sec:swgrammar}, though we will provide an overview first.

\par
\textbf{\textit{Values.}} Declarations of values bind an identifier to an expression or value. Since the \GEMINI\space compiler performs type inference, value declarations can either be made implicitly, wherein the value type is omitted, or explicitly. The latter may be useful in the case that type inference would infer a parametrically polymorphic type whereas the programmer wants to restrict the type.

\par
\textbf{\textit{Functions.}} Declarations of functions bind an identifier to a function by specifying one or more arguments and the expression to evaluate when the function is invoked. The argument identifiers are lexically scoped in the body of the function; further, if an argument uses the same name as an identifier bound elsewhere, it will take precedence. The function's return value and any arguments may also be optionally typed explicitly.

\par
\textbf{\textit{Types.}} Declarations of types bind an identifier to a type. Parametrically polymorphic type constructors may also be declared by preceding the identifier with a type variable that is referenced by the type.

\par
\textbf{\textit{Datatypes.}} Declarations of datatypes bind an identifier to a discriminated union type, as well as specifying each of the variants and their type arguments. Datatype declarations may be parametrically polymorphic as is the case with type declarations. Further, software datatype declarations are syntactically distinct from hardware datatype declarations in order to enforce that the type arguments of all variants belong to the appropriate kind.

\par
\textbf{\textit{Modules.}} Declarations of modules bind an identifier to a module by specifying a single argument and the body of the module. The module's return value and argument may also be optionally typed explicitly. Modules can also be dependently typed in the declaration by specifying an integer-typed identifier that can be referenced throughout the module body. When the module is invoked with a given integer value, it is substituted for the identifier throughout the module. This is the mechanism by which polymorphic modules can be created; however, as noted in our discussion of dependent types in Section \ref{sec:typesandkinds}, a polymorphic module cannot be the result of a \GEMINI\space program.

\subsection{\GEMINI\space core library}
\label{sec:library}
\par
The \GEMINI\space core library provides useful built-in functions and modules for operating on software- and hardware-typed terms, respectively. A complete list of these can be found in \ref{sec:lib} of the Appendix.

\subsection{Derived terms}
As discussed in some previous sections, certain \GEMINI\space expressions are merely special cases of other more general expressions. For example, the expression
\begin{tcolorbox}
\begin{Verbatim}
if e1 then e2
\end{Verbatim}
\end{tcolorbox}
\noindent
is a special case of the more general if-then-else term and can be rewritten as the following equivalent expression:
\begin{tcolorbox}
\begin{Verbatim}
if e1 then e2 else ()
\end{Verbatim}
\end{tcolorbox}
\par
In fact, since the tuple is a special case of a record with consecutively numbered fields beginning at \texttt{1}, the expression above is equivalently expressed as the following:
\begin{tcolorbox}
\begin{Verbatim}
if e1 then e2 else {}
\end{Verbatim}
\end{tcolorbox}
\par
These special cases are called \textit{derived terms} as they are derived from more general expressions. Beyond the lexical analysis phase of compilation, they are treated identically to the expressions from which they are derived, which simplifies the implementation of later phases. The complete list of derived terms can be found in \ref{sec:derivedterms} of the Appendix.

\subsection{Software grammar}
\label{sec:swgrammar}
Having provided an informal overview of \GEMINI, we will begin to formalize the language beginning with the grammars. In this and Section \ref{sec:hwgrammar}, we will provide context-free grammars (CFG) for software and hardware terms as a set of production rules, each of which specifies how some non-terminal can be substituted by a sequence of non-terminals and terminals.

\subsubsection{Software value grammar}
The first grammar represents software values; these are non-reducible (or \textit{terminal}) software-typed terms. The grammar here does not represent the syntax, but instead represents each value abstractly in a mathematical sense. Figure \ref{fig:swvgrammar} depicts an excerpt of the software value grammar in Backus-Naur form (BNF).
\begin{Figure}
    \begin{tcolorbox}
    \begin{grammar}
        <swv> ::= <integer>
        \alt <real>
        \alt <string>
        \alt <list>
        \alt <software record>
        \alt <sw>
        \alt <ref>
        \alt <variant>
        \alt <function>
        
        <integer> ::= $i \in \mathbb{Z} \cap [-2^{31}, 2^{31} - 1]$
        
        <function> ::= $\lambda x : T_S.e$
    \end{grammar}
    \end{tcolorbox}
    \captionof{figure}{\GEMINI\space software value grammar in BNF (excerpt)}
    \label{fig:swvgrammar}
\end{Figure}
\par
In the excerpt shown, there are three non-terminals. The first represents a \textit{software value}, which can be substituted for any of the non-terminals on the right. Two of these are included in the excerpt. The first represents a 32-bit integer value, which is some $i$ in the set of integer numbers $\mathbb{Z}$ between $-2^{31}$ and $2^{31}$, inclusive. The second represents a function expressed as an \textit{abstraction} in the notation of the lambda calculus \cite{tapl}. In this notation, $x$ represents the argument which has type $T_S$ which denotes some type of kind \SWKIND\space and some body expression $e$. The complete software value grammar can be found in \ref{sec:swvgrammar} of the Appendix.

\subsubsection{Software term grammar}
We now shift our attention to the grammar for software terms, which includes both expressions and values. Unlike the previous grammar, this one reflects the syntax to be used in writing \GEMINI\space programs. Figure \ref{fig:swtgrammar} depicts an excerpt of the software term grammar in BNF.
\begin{Figure}
    \begin{tcolorbox}
    \begin{grammar}
        <exp> ::= <literal>
        \alt <access>
        \alt <let binding>
        \alt <conditional>
        \alt <operation>
        \alt <assignment>
        \alt <pattern match>
        \alt <sequence>
        \alt <application>
        
        <literal> ::= <identifier>
        \alt <integer literal>
        \alt <real literal>
        \alt <string literal>
        \alt <list literal>
        \alt <software record literal>
        \alt <ref literal>
        \alt <sw literal>
        
        <list literal> ::= "[" <list-body> "]"
        \alt "nil"
        
        <list-body> ::= <exp> <exp comma tail>
        \alt $\epsilon$
        
        <exp comma tail> ::= "," <exp> <exp comma tail>
        \alt $\epsilon$
        
        <ref literal> ::= "ref" <exp>
    
        <sw literal> ::= "sw" <hwv>
    \end{grammar}
    \end{tcolorbox}
    \captionof{figure}{\GEMINI\space software term grammar in BNF (excerpt)}
    \label{fig:swtgrammar}
\end{Figure}
\par
The non-terminal \nonterm{exp} represents all possible expressions of software terms. The first such choice is the literal declaration of a value, denoted by the non-terminal \nonterm{literal}. Each choice is another non-terminal corresponding to some software type; the exception here is the first choice, \nonterm{identifier}, which refers to an expression or value by the name to which it is bound.
\par
Figure \ref{fig:swtgrammar} shows the rules for three literals. The first is the literal declaration of a list, which is either a comma-separated list of software terms (that may be empty) surrounded by square brackets, or the terminal \texttt{nil} which represents the empty list. Here our definition of the non-terminal \nonterm{list-body} is recursive in reference to \nonterm{exp}. Further, the definition of non-terminal \nonterm{exp comma tail} is recursive in reference to itself. The second literal declaration is that of the reference, which is the characters \texttt{ref} followed by another \nonterm{exp}. The third literal declaration is that of the software wrapper, which is the characters \texttt{sw} followed by a hardware value represented by the non-terminal \nonterm{hwv}. We will see the full hardware value grammar containing this definition in Section \ref{sec:hwgrammar}, though here we note that the definitions of the software term grammar and hardware value grammar must coexist such that references from one to the other can be made; the separation as presented in this paper is merely to compartmentalize the syntax based on the kind to which term types belong. The complete software term grammar can be found in \ref{sec:swtgrammar} of the Appendix.

\subsection{Hardware grammar}
\label{sec:hwgrammar}
Another reason for the dichotomy of software and hardware in the design of \GEMINI\space is that while there are both software-typed expressions, which are reducible, and software-typed values, which are not, there only exist hardware-typed values. This is because the semantics of hardware-typed \textit{expressions} are ill-defined, since no hardware circuit can be reduced any further than the structure it takes\footnote{The exception is hardware subcircuits defined in terms of literal bit values, which can be evaluated at compile-time. This is considered an optimization by the compiler.}.

\subsubsection{Hardware value grammar}
We first present the hardware value grammar, which represents hardware values abstractly instead of syntactically. The full hardware value grammar is shown in Figure \ref{fig:hwvgrammar}.
\par
We note a few interesting aspects of this value grammar. First, the sole terminal values are \texttt{0} and \texttt{1} representing the binary values that a bit may take. This supports the idea that the atomic units of hardware circuits are bits, and are built using a variety of constructors which are represented by the remaining alternations of the grammar.
\par
The next four alternations are logic gates: \textit{AND}, \textit{OR}, \textit{XOR}, and \textit{NOT}. We note that compound logic gates such as \textit{NAND} and \textit{NOR} and not represented as first-class entities in this grammar and are instead constructed by sequencing the appropriate gates\footnote{In the case of \textit{NAND}, this is an \textit{AND} gate followed by a \textit{NOT} gate}. The inputs to these logic gates are recursively defined in terms of the non-terminal \nonterm{hwv} and may thus be themselves the outputs of other logic gates.
\par
The penultimate alternation is the array, which is represented as a collection of zero or more hardware values\footnote{The element values must all have the same type; this is not represented in the grammar but is enforced in the typing rules.}.
\par
We also note the deliberate omission of hardware records from our grammar definition. This is because hardware records are translated into bit arrays at compilation and therefore never exist as hardware values in a true sense; they are merely a convenient abstraction for programmers.
\par
The final alternation is the delayed hardware value, denoted by the application of the delay function $\delta$. This value will be temporally typed and the specific value parameterizing the type will depend on the number of clock cycles the value is delayed.

\begin{Figure}
    \begin{tcolorbox}
    \begin{grammar}
        <hwv> ::= "0"
        \alt "1"
        \alt \ANDGATE{\nonterm{hwv}$_1$}{\nonterm{hwv}$_n$}
        \alt \ORGATE{\nonterm{hwv}$_1$}{\nonterm{hwv}$_n$}
        \alt \XORGATE{\nonterm{hwv}$_1$}{\nonterm{hwv}$_n$}
        \alt \NOTGATE{\nonterm{hwv}}
        \alt "#[" <hwv>$_i^{i \in 0..n-1}$ "]"
        \alt $\delta$"(" <hwv> ")"
    \end{grammar}
    \end{tcolorbox}
    \captionof{figure}{\GEMINI\space hardware value grammar in BNF}
    \label{fig:hwvgrammar}
\end{Figure}

\subsubsection{Hardware syntax grammar}
We now shift our attention to the grammar representing the syntax for declaring hardware values in \GEMINI\space programs. Figure \ref{fig:hwsgrammar} presents an excerpt of the hardware syntax grammar in BNF.
\begin{Figure}
    \begin{tcolorbox}
    \begin{grammar}
        <exp> ::= <literal>
        \alt <access>
        \alt <let binding>
        \alt <operation>
        \alt <parameterization>
    
        <literal> ::= <bit literal>
        \alt <array literal>
        \alt <hardware record literal>
        
        <bit literal> ::= "'b:" <binary-digit>
    \end{grammar}
    \end{tcolorbox}
    \captionof{figure}{\GEMINI\space hardware syntax grammar in BNF (excerpt)}
    \label{fig:hwsgrammar}
\end{Figure}
\par
As mentioned previously, the grammar here is intended to augment the grammar of software terms shown in Figure \ref{fig:swtgrammar}. Certain non-terminals, such as \nonterm{exp}, are repeated and any new rules appearing here should be appended to those from earlier. The complete hardware syntax grammar can be found in \ref{sec:hwsgrammar} of the Appendix.
\par
We also reiterate that the grammar itself is not responsible for enforcing typing rules. This allows us to augment our existing software grammar with the hardware productions. The typing rules will be formalized in the next section.

\subsection{Typing rules}
\label{sec:tyrules}
It is now possible to verify whether a given program is grammatically valid \GEMINI\space with these grammars, and they additionally allow us to may construct an abstract syntax tree from said program. However, given \GEMINI\space is strongly and statically typed, grammatical validity alone is not sufficient. We must also verify at compile-time that the program is well-typed. We defer discussing the implementation details until Section \ref{sec:semanalysis}. At this point, we will define the formal tools for encoding what it means for a program to be well-typed, which are the basis not only for implementation but also our proofs of metatheory in Section \ref{sec:metatheory}.

\par
A \textit{typing rule} is a theorem, consisting of a set of propositions or hypothesis, and a conclusion. They are illustrated diagrammatically as a horizontal line above which is written the antecedent clause is written and below which is written the consequent clause. An excerpt of the typing rules is shown in Figure \ref{fig:tyrules}.

\par
The antecedent clause is true iff each of the propositions is true. Further, if the antecedent clause is true, then the consequent clause is true. As an example, the rule \titlesc{T-IfElse} is pronounced "if \texttt{t$_1$} has type \texttt{int} and \texttt{t$_2$} has type \texttt{T} and \texttt{t$_3$} has type \texttt{T}, then \texttt{if t$_1$ then t$_2$ else t$_3$} has type \texttt{T}".

\par
The typing rules are shown in full in \ref{sec:fulltyrules} of the Appendix. While these rules enable us to verify that expressions are well-typed, they rely on knowing the types of subexpressions. Since \GEMINI\space allows implicit typing of expressions and identifiers, the compiler must perform type inference in order to determine the types of these subexpressions. The typing rules in their current form will not enable us to do so; for this, we must rely on the inversion of the typing relation which we will see in Section \ref{sec:metatheory}.

\begin{Figure}
    \begin{tcolorbox}
    \tyrule{%
        \AxiomC{$\mathtt{x : T \in \Gamma}$}%
        \UnaryInfC{$\mathtt{\Gamma \vdash x : T}$}
        \DisplayProof
    }{T-Var}
    
    \tyrule{%
        \AxiomC{$\mathtt{\Gamma, x : T_1 \vdash t_2 : T_2}$}%
        \UnaryInfC{$\mathtt{\Gamma \vdash \lambda x : T_1 . t_2 : T_1 \rightarrow T_2}$}
        \DisplayProof
    }{T-Abs}
    
    \tyrule{%
        \AxiomC{$\mathtt{\Gamma \vdash t_1 : T_1 \rightarrow T_2}$}%
        \AxiomC{$\mathtt{\Gamma \vdash t_2 : T_1}$}
        \BinaryInfC{$\mathtt{\Gamma \vdash t_1 t_2 : T_2}$}
        \DisplayProof
    }{T-App}
    
    \tyrule{%
        \AxiomC{$\mathtt{t_1 : int}$}%
        \AxiomC{$\mathtt{t_2 : int}$}
        \BinaryInfC{$\mathtt{t_1}$ \texttt{+} $\mathtt{t_2 : int}$}
        \DisplayProof
    }{T-Int-Add}
    
    \tyrule{%
        \AxiomC{$\mathtt{t_1 : T_H}$}
        \AxiomC{$\mathtt{t_2 : T_H}$}
        \BinaryInfC{$\mathtt{t_1}$ \texttt{\&} $\mathtt{t_2 : T_H}$}
        \DisplayProof
    }{T-And}
    
    \tyrule{%
        \AxiomC{$\mathtt{t_1 : int}$}%
        \AxiomC{$\mathtt{t_2 : T}$}
        \AxiomC{$\mathtt{t_1 : T}$}
        \TrinaryInfC{\texttt{if} $\mathtt{t_1}$ \texttt{then} $\mathtt{t_2}$ \texttt{else} $\mathtt{t_3 : T}$}
        \DisplayProof
    }{T-IfElse}
    \end{tcolorbox}
    \captionof{figure}{\GEMINI\space typing rules (excerpt)}
    \label{fig:tyrules}
\end{Figure}

\subsection{Evaluation rules}
\label{sec:evalrules}
Now with a rigorous formulation of the syntax and typing rules of our language from Sections \ref{sec:swgrammar}, \ref{sec:hwgrammar}, and \ref{sec:tyrules}, we must define precisely how expressions are evaluated. In defining the semantics of our language, we must choose between three approaches: \textit{operational semantics}, \textit{denotational semantics}, and \textit{axiomatic semantics} \cite{tapl}. In this paper, we elect to define the semantics in terms of operational semantics for its simplicity and flexibility.

\par
Operational semantics define an abstract state machine, where each state is an expression. The machine’s behavior is defined by a transition function that either yields the next state by performing a step of computation, or declares that the machine has halted by reaching some terminal value. Operational semantics can be further partitioned into small-step and big-step semantics. Small-step semantics, or \textit{structural operational semantics}, consider how evaluation takes place one step at a time. Big-step semantics, or \textit{natural semantics}, instead describe the final value to which some expression evaluates \cite{tapl}.

\par
In general, big-step semantics are less verbose since intermediate expression states need not be encoded in the machine behavior. However, small-step semantics are more precise and readily translatable for implementation, and since our goal is to develop a compiler we elect to define our evaluation rules in terms of small-step semantics.

\begin{Figure}
    \begin{tcolorbox}
        \tyrule{%
            \AxiomC{$\mathtt{t_1} \longrightarrow \mathtt{t_1'}$}
            \UnaryInfC{$\mathtt{t_1 t_2} \longrightarrow \mathtt{t_1' t_2}$}
            \DisplayProof
        }{E-App1}
        
        \tyrule{%
            \AxiomC{$\mathtt{t_2} \longrightarrow \mathtt{t_2'}$}
            \UnaryInfC{$\mathtt{v_1 t_2} \longrightarrow \mathtt{v_1 t_2'}$}
            \DisplayProof
        }{E-App2}
        
        \tyrule{%
            \AxiomC{$\mathtt{(\lambda x.t_1)v_1} \longrightarrow \mathtt{[x \mapsto v_1]t_1}$}
            \DisplayProof
        }{E-AppAbs}
        
        \tyrule{%
            \AxiomC{$\mathtt{t_1 \longrightarrow t_1'}$}
            \UnaryInfC{\texttt{if} $\mathtt{t_1}$ \texttt{then} $\mathtt{t_2}$ \texttt{else} $\mathtt{t_3}$}
            \noLine
            \UnaryInfC{$\longrightarrow$ \texttt{if} $\mathtt{t_1'}$ \texttt{then} $\mathtt{t_2}$ \texttt{else} $\mathtt{t_3}$}
            \DisplayProof
        }{E-IfElse}
        
        \tyrule{%
            \AxiomC{$\mathtt{v_1 : int}$}
            \AxiomC{$\mathtt{v_1 \neq 0}$}
            \BinaryInfC{\texttt{if} $\mathtt{v_1}$ \texttt{then} $\mathtt{t_2}$ \texttt{else} $\mathtt{t_3} \longrightarrow \mathtt{t_2}$}
            \DisplayProof
        }{E-IfElse-T}
        
        \tyrule{%
            \AxiomC{\texttt{if 0 then} $\mathtt{t_2}$ \texttt{else} $\mathtt{t_3} \longrightarrow \mathtt{t_3}$}
            \DisplayProof
        }{E-IfElse-F}
    \end{tcolorbox}
    \captionof{figure}{\GEMINI\space evaluation rules (excerpt)}
    \label{fig:evrules}
\end{Figure}

\par
Similarly to typing rules, each evaluation rule is a theorem. An excerpt of the set of evaluation rules is shown in Figure \ref{fig:evrules}. In these rules, the character \texttt{t} denotes a term that may be further reduced by some evaluation rule. The character \texttt{v} denotes a terminal value that cannot be evaluated any further.

\par
As an example, the rule \titlesc{E-IfElse-T} is pronounced "if \texttt{v$_1$} has type \texttt{int} and \texttt{v$_1$} is not equal to \texttt{0}, then the expression \texttt{if v$_1$ then t$_2$ else t$_3$} evaluates to \texttt{t$_2$}". Next, we would evaluate \texttt{t$_2$} by the appropriate rule until the result is a terminal value. The complete set of evaluation rules is shown in \ref{sec:fullevrules} of the Appendix.

\par
The evaluation rules together define a precise and unambiguous evaluation strategy. For example, consider the expression \texttt{if x then (if 1 then "a" else "b") else "c"}. Under the evaluation rules of our language, it is not possible for this to evaluate to \texttt{if x then "a" else "c"}, despite this being a state that would evaluate to an equivalent value. We must first evaluate the guard of the outer-conditional by rule \titlesc{E-IfElse}. Once it is a terminal value, then we pick one of the then- and else-clauses based on rules \titlesc{E-IfElse-T} and \titlesc{E-IfElse-F}. A useful property is the determinacy of one-step evaluation, stating that if \texttt{t $\longrightarrow$ t'} and \texttt{t $\longrightarrow$ t''}, then \texttt{t' = t''}. This ensures that evaluation is a deterministic process.

\section{\GEMINI\space Metatheory}
\label{sec:metatheory}
In this section, we utilize the formalizations of grammar, semantics, and evaluation made in the previous section in order to prove desirable properties of \GEMINI.

\par
First, we must reiterate the definition of some terms. A term \texttt{t} is in \textit{normal form} if no evaluation rule can be applied to it. A term \texttt{t} is in a \textit{stuck state} if it is in normal form but it is not a value. We aim to prove a basic property of \GEMINI's type system: \textit{safety}. We will have achieved this if we prove that a well-typed term can never reach a stuck state during evaluation.

\par
It is important to prove that a type system possesses safety in order to ensure that well-typed programs can be compiled and executed without entering a stuck state. It is especially critical in the case of the \GEMINI\space language since software evaluation occurs as an intermediary step of compilation, and failure to be type-safe could lead to a non-terminating process at compile-time.

\par
In order to build to our proof of safety, we must first establish supporting proofs of two other properties: \textit{progress} and \textit{preservation}.

\subsection{Proof of progress}
Recall that a type system has the property of \textit{progress} if a well-typed term is never in a stuck state; either it is a value or it or it can take a step according to some evaluation rule.

\par
In order to prove that our type system has this property, we must first prove two supporting lemmas.

\begin{lem}[Inversion of Typing Relation]
The following are true, and constitute the inversion of the typing relation:
\label{lem:inversion-of-typing}
\end{lem}

\begin{enumerate}
    \item If $\mathtt{\Gamma \vdash x : R}$, then $\mathtt{x : R \in \Gamma}$.
    \item If $\mathtt{\Gamma \vdash \lambda x : T_1 . t_2 : R}$, then $\mathtt{R = T_1 \rightarrow R_2}$ for some $\mathtt{R_2}$ with $\mathtt{\Gamma, x : T_1 \vdash t_2 : R_2}$.
    \item If $\mathtt{\Gamma \vdash t_1}$ $\mathtt{t_2 : R}$ then there is some type $\mathtt{T_{11}}$ such that $\mathtt{\Gamma \vdash t_1 : T_{11} \rightarrow R}$ and that $\mathtt{\Gamma \vdash t_2 : T_{11}}$.
    \item If $\langle integer \rangle \mathtt{: R}$, then \texttt{R = int}.
\end{enumerate}

The remainder of the cases are omitted here and are shown in full in \ref{sec:inversion-of-typing} of the Appendix.

\begin{proof}
    Immediate from the definition of the typing rules.
\end{proof}

\begin{lem}[Canonical Forms]
The following are true, and constitute the canonical forms:
\label{lem:canonical-forms}
\end{lem}

\begin{enumerate}
    \item If \texttt{v} is a value of type \texttt{int}, then \texttt{v} is an integer value according to the software value grammar.
    \item If \texttt{v} is a value of type \texttt{real}, then \texttt{v} is a real value according to the software value grammar.
    \item If \texttt{v} is a value of type \texttt{string}, then \texttt{v} is a string value according to the software value grammar.
    \item If \texttt{v} is a value of type \texttt{bit}, then \texttt{v} is either \texttt{0} or \texttt{1}.
    \item If \texttt{v} is a value of type \texttt{T$_1$ $\rightarrow$ T$_2$}, then \texttt{v = $\lambda x$:T$_1$.t$_2$}.
    \item If \texttt{v} is a value of type \texttt{T$_s$ ref}, then \texttt{v} is a location in store $\mu$.
    \item If \texttt{v} is a value of type \texttt{\{l$_i$: T$_i^{i\in1..n}$\}}, then \texttt{v} is a value with the form \texttt{\{l$_i$ = v$_i^{i\in1..n}$\}}.
    \item If \texttt{C} is a constructor of datatype \texttt{D} accepting type \texttt{T$_1$} and \texttt{v} is a value of type \texttt{T$_1$}, then \texttt{C v} is a value of type \texttt{D} with form \texttt{$\langle$C=v$\rangle$}.
    \item If \texttt{v} is a value of type \texttt{T$_H$[n]}, then \texttt{v} is an array value according to the hardware value grammar.
    \item If \texttt{v} is a value of type \texttt{\#\{l$_i$: T$_i^{i\in1..n}$\}}, then \texttt{v} is a value with the form \texttt{\#\{l$_i$ = v$_i^{i\in1..n}$\}}.
    \item If \texttt{v} is a value of type \texttt{T$_H$ sw}, then \texttt{v} is a value with the form \texttt{$\omega$(v$_H$)} for some \texttt{v$_H$} of type \texttt{T$_H$}.
\end{enumerate}
\begin{proof}
    We proceed through each case of the canonical forms and refer to the inversion from Lemma \ref{lem:inversion-of-typing}.
    \begin{description}
        \item \textit{Case} 1: Values in this language can take several forms. The case of an integer gives us our desired result immediately. All other forms cannot occur since we assumed that \texttt{v} has type \texttt{int} and among the cases in consideration from Lemma \ref{lem:inversion-of-typing}, only case 4 tells us that the value has type \texttt{int}.
    \end{description}
    
    The remaining cases are similar.
\end{proof}

We are now equipped to prove the theorem of progress.

\begin{thm}[\textbf{\titlesc{Progress}}]
Suppose \texttt{t} is a closed, well-typed term ($\vdash$\texttt{t:T} for some type \texttt{T}). Then either \texttt{t} is a value or else there is some \texttt{t'} such that \texttt{t $\longrightarrow$ t'}.
\label{thm:progress}
\end{thm}

\begin{proof}
    By structural induction on a derivation of \texttt{t:T}.
    
    \begin{description}
        \item \textit{Case} \titlesc{T-Int}, \titlesc{T-Real}, \titlesc{T-String}, \titlesc{T-Bit}, \titlesc{T-Nil}:\\
        Immediate since \texttt{t} is a value.
        
        \item \textit{Case} \titlesc{T-App}:\\
        \texttt{t = t$_1$ t$_2$}\\
        \texttt{$\vdash$ t$_1$ : T$_{11}$ $\rightarrow$ T$_{12}$}\\
        \texttt{$\vdash$ t$_2$ : T$_{12}$}\\
        \\
        By the induction hypothesis, either \texttt{t$_1$} is a value or else there is some other \texttt{t$_1$'} for which \texttt{t$_1$ $\longrightarrow$ t$_1$'}, and likewise for \texttt{t$_2$}. If \texttt{t$_1$ $\longrightarrow$ t$_1$'} then by \titlesc{E-App1}, \texttt{t $\longrightarrow$ t$_1$' t$_2$}. On the other hand, if \texttt{t$_1$} is a value and \texttt{t$_2$ $\longrightarrow$ t$_2$'}, then by \titlesc{E-App2}, \texttt{t $\longrightarrow$ t$_1$ t$_2$'}. Finally, if both \texttt{t$_1$} and \texttt{t$_2$} are values, then case 5 of the canonical forms lemma tells us that \texttt{t$_1$} has the form \texttt{$\lambda x$:T$_{11}$.t$_{12}$} and so by \titlesc{E-AppAbs}, \texttt{t $\longrightarrow$ [x$\mapsto$t$_2$]t$_{12}$} which is a value.
    \end{description}
    The remaining cases are shown in full in \ref{sec:progress} of the Appendix.
\end{proof}

\subsection{Proof of preservation}
We are also equipped to prove the theorem of preservation, which states that performing one step of evaluation preserves the type of the original term.

\begin{thm}[\textbf{\titlesc{Preservation}}]
If \texttt{t:T} and \texttt{t $\longrightarrow$ t'}, then \texttt{t':T}.
\label{thm:preservation}
\end{thm}

\begin{proof}
    By structural induction on a derivation of \texttt{t:T}. At each step of the induction, we assume that the desired property holds for all subderivations\footnote{This means that if \texttt{s:S} and \texttt{s $\longrightarrow$ s'}, then \texttt{s':S} whenever \texttt{s:S} is proved by a subderivation of the present one} and proceed by case analysis on the final rule in the derivation.
    
    \begin{description}
        \item \textit{Case} \titlesc{T-Var}:\\
        \texttt{t = x}\\
        \texttt{x:T $\in \Gamma$}\\\\
        If the last rule in the derivation is \titlesc{T-Var}, then we know from the form of this rule that \texttt{t} must be a variable of type \texttt{T}. Thus \texttt{t} is a value, so it cannot be the case that \texttt{t $\longrightarrow$ t'} for any \texttt{t'}, and the requirements of the theorem are vacuously satisfied.

        \item \textit{Case} \titlesc{T-App}:\\
        \texttt{t = t$_1$ t$_2$}\\
        \texttt{$\Gamma \vdash$ t$_1$:T$_{11} \rightarrow$ T$_{12}$}\\
        \texttt{$\Gamma \vdash$ t$_2$:T$_{11}$}\\
        \texttt{T = T$_{12}$}\\
        
        Looking at the evaluation rules with application on the left-hand side, we find that there are three rules by which \texttt{t $\longrightarrow$ t'} can be derived: \titlesc{E-App1}, \titlesc{E-App2}, and \titlesc{E-AppAbs}. We consider each case separately.
        
        \begin{description}
            \item \textit{Subcase} \titlesc{E-App1}:\\
            \texttt{t$_1$ $\longrightarrow$ t$_1$'}\\
            \texttt{t' = t$_1$' t$_2$}\\
            
            From the assumptions of the \titlesc{T-App} case, we have a subderivation of the original typing derivation whose conclusion is \texttt{$\Gamma \vdash$ t$_1$:T$_{11} \rightarrow$ T$_{12}$}. We can apply the induction hypothesis to this subderivation obtaining \texttt{$\Gamma \vdash$ t$_1$':T$_{11} \rightarrow$ T$_{12}$}. Combining this with the fact that \texttt{$\Gamma \vdash$ t$_2$:T$_{11}$}, we can apply rule \titlesc{T-App} to conclude that \texttt{$\Gamma \vdash$ t':T}.
            
            \item \textit{Subcase} \titlesc{E-App2}:\\
            Similar to \titlesc{E-App1}.
            
            \item \textit{Subcase} \titlesc{E-AppAbs}:\\
            \texttt{t$_1$ = $\lambda x$:T$_{11}$.t$_{12}$}\\
            \texttt{t$_2$ = v$_2$}\\
            \texttt{t' = [x$\mapsto$v$_2$]t$_{12}$}\\
            
            Using Lemma \ref{lem:inversion-of-typing}, we can deconstruct the typing derivation for \texttt{$\lambda x$:T$_{11}$.t$_{12}$} yielding \texttt{$\Gamma, x$:T$_{11} \vdash$ t$_{12}$ : T$_{12}$}. From this we obtain \texttt{$\Gamma \vdash$ t':T$_{12}$}.
        \end{description}
    \end{description}
    The remaining cases are shown in full in \ref{sec:preservation} of the Appendix.
\end{proof}

\subsection{Proof of safety}
Having proved the theorems of progress and preservation, we are now ready to prove that the \GEMINI\space type system has the desired property of safety.

\begin{thm}[\textbf{\titlesc{Safety}}]
A well-typed term can never reach a stuck state in evaluation.
\label{thm:safety}
\end{thm}

\begin{proof}
    Theorem \ref{thm:progress} demonstrates that a well-typed term is not stuck, and Theorem \ref{thm:preservation} demonstrates that if a well-typed term takes a step of evaluation, then the resulting term is also well-typed. In combination and by induction, these guarantee safety.
\end{proof}

\section{The \GEMINI\space Compiler}
\label{sec:compiler}
At this stage, we have formalized the \GEMINI\space language grammars, typing rules, and evaluation rules, and proved the desirable property of safety for our type system. This positions us to implement our compiler.

\par
The compiler discussed in this paper accepts a \GEMINI\space program as an input and produces Verilog as an output. We have written the compiler in the SML/NJ language. In order to explain the implementation of the compiler, we will decompose the implementation into five sequential phases.

\subsection{Lexer}
The first phase of compilation is \textit{lexical analysis}, performed by the \textit{lexer} module. In this phase, the program is scanned to produce syntactic units called \textit{lexemes}, which are then classified into a particular token class. The lexer is specified by an ordered set of patterns which match against certain sequences of characters within the language's alphabet. The lexer scans linearly until it encounters a sequence of characters that match a defined pattern. We denote a sequence of characters from position $i$ to $j$ as $C_{i,j}$.

\par
The lexer behaves deterministically by following two priority rules:
\begin{enumerate}
    \item \textbf{Rule of longest match}: if the lexer has encountered a valid sequence $C_{i,j}$, it will first check whether $C_{i,j+1}$ is also valid; if it is, then it will disregard $C_{i,j}$ in favor of $C_{i,j+1}$, else it tokenizes $C_{i,j}$.
    \item \textbf{Rule of earliest pattern}: if two lexer rules match the same sequence $C_{i,j}$, then it will tokenize based on the pattern that appears earliest in the list.
\end{enumerate}
\par
We provide illustrative examples of both rules. The tokens \texttt{>}, \texttt{=}, and \texttt{>=} all exist in the \GEMINI\space language as operators. By the rule of longest match, when the lexer encounters the character sequence \texttt{>=} it will tokenize it as the single token \texttt{>=} instead of the token \texttt{>} followed by the token \texttt{=}.

\par
Further, \GEMINI\space possesses the keyword \texttt{if} and defines identifiers as alphanumeric sequences of characters\footnote{Identifiers are actually defined slightly more restrictively, though this definition is fine for our example; a more precise definition can be found in the software grammar of \ref{sec:swtgrammar} of the Appendix.}. If the character sequence \texttt{if} is encountered, the tokenization will depend on the relative ordering of the patterns. Since we desire the sequence to be tokenized as the keyword \texttt{if} as opposed to the identifier \texttt{if}, we must specify the pattern for the keyword before the pattern for the identifiers such that we yield a keyword when applying the rule of earliest pattern. If the order were reversed, the lexer would never toknize the keyword\footnote{For this reason, the pattern for identifiers appears after all the pattern for all keywords.}.

\par
In some cases, the lexer needs to perform some basic computation to attach values to tokens. For example, \GEMINI\space allows integers to be declared in various bases. When the lexer encounters a hexadecimal representation of an integer, such as \texttt{\#'h:beef}, it computes the integer value in base-10 and tokenizes it as an integer value, the same way it would have treated the equivalent base-10 integer \texttt{48879}. This allows various representations to be treated identically past lexical analysis, thereby simplifying the implementation of later phases by reducing the number of cases to consider.

\par
In addition to performing tokenization, the \GEMINI\space lexer is responsible for ensuring that comments are balanced and that string quotes are closed before the end of the program. It does so by maintaining global state of nested comment depth and whether a quote has been left open, and ensuring that when the end-of-file is encountered that both are the appropriate values (\texttt{0} and \texttt{false}, respectively).

\par
The lexer was written using the ML-Lex tool developed by Andrew Appel. Each pattern is specified by a regular expression to match against and an action to execute; the action may be to report an error, generate a token, perform some side-effect, or any combination of the three. These patterns are written in a \texttt{.lex} file which is compiled to generate the appropriate SML code that perform lexing \cite{appel1994}.

\subsection{Parser}
The responsibility of the lexer is to enforce syntactical correctness of the input program. However, syntactically correct programs may not necessarily be grammatically correct. For example, consider the following grammatically incorrect program:
\begin{tcolorbox}
\begin{Verbatim}
if if if
\end{Verbatim}
\end{tcolorbox}
\par
The lexer would process the program successfully, though this is clearly an ill-formed program. We need to additionally enforce correctness of the structure of programs. This is the responsibility of the parser.

\par
Once the program passes through the lexer, we obtain a linear stream of tokens. The parser provides structure to the tokens by constructing an \textit{abstract syntax tree} (AST). The separation of the lexer and parser modules is a useful software design as it allows us to modify the concrete syntax of our language in the lexer – for example, by replacing the equality operator \texttt{=} with \texttt{==} – without requiring any changes to be made to the parser which only sees the tokenized output from the lexer as something such as \textit{EqualityOperator}\footnote{This is the reason the AST is considered \textit{abstract}, since references to the concrete syntax of the language have been shed.}.

\par
The tool ML-Yacc was utilized to specify the CFG for \GEMINI\space programs. We define a set of production rules in terms of the tokens produced by the lexer as well as a set of non-terminals. Further, each production rule is accompanied by a semantic action to specify some return value. As the program is processed by a look-ahead LR parser, the AST is constructed from the expressions returned by the semantic actions \cite{tarditi}. The production rules can easily be translated from the definitions of our syntax grammars from Sections \ref{sec:swgrammar} and \ref{sec:hwgrammar}.

\par
ML-Yacc further allows for the specification of \textit{precedence rules} which dictate the order of precedence for terminals. This affects the order in which the AST is constructed. These are important to specify correctly in order to enforce the expected semantics of certain expressions such as arithmetic ones that must follow the correct order of operations. The order of precedence for operators in \GEMINI\space is shown in \ref{sec:order-of-precedence} of the Appendix.

\subsection{Semantic analysis}
\label{sec:semanalysis}
The responsibility of the parser is to enforce grammatical correctness of a syntactically correct program. We have made progress, since our previously grammatically incorrect program would be caught at the parsing phase. However, the parser does not enforce semantic correctness of a program. Consider the following syntactically and grammatically correct, yet semantically incorrect program:
\begin{tcolorbox}
\begin{Verbatim}
42 * "a"
\end{Verbatim}
\end{tcolorbox}
\par
According to the definition of our typing rules in Section \ref{sec:tyrules}, this program is not well-typed; the multiplication operator \texttt{*} can only operate on operands of type \texttt{int}, yet the second operand has type \texttt{string}. We need to enforce semantic correctness of programs. This is the responsibility of the semantic analysis phase.

\par
In this phase, we recurse over the AST produced by the parser and verify that the semantics of the program are valid. However, a prohibitive issue is that not all types are known yet, since declarations of values, functions, and modules may be implicitly typed. We must first infer the actual types of expressions in a program. This is possible to do given that \GEMINI's type system can be classified as a Hindley-Milner type system \cite{tapl}. Thus, we further divide semantic analysis into three subphases: \textit{decoration}, \textit{inference}, and \textit{type-checking}.

\subsubsection{Type decoration}
In the first subphase, the \GEMINI\space program is transformed into an intermediate language we will refer to as \titlesc{ExplicitGemini}. In this language, all implicitly-typed identifiers are decorated with explicit types, as demonstrated in Figure \ref{fig:explicit-gemini}.

\begin{Figure}
    \begin{tcolorbox}
    \begin{Verbatim}
fun foo(x, y, s: string) =
    (print(s); x * y)
    \end{Verbatim}
    \end{tcolorbox}
    \begin{center}
        (a)
    \end{center}
    
    \begin{tcolorbox}
    \begin{Verbatim}
fun foo(x: `a, y: `b, s: string): `c =
    (print(s); x * y)
    \end{Verbatim}
    \end{tcolorbox}
    \begin{center}
        (b)
    \end{center}
    
    \captionof{figure}{A function written in (a) \GEMINI\space and (b) \titlesc{ExplicitGemini}}
    \label{fig:explicit-gemini}
\end{Figure}
\par
We begin by decorating each implicitly-typed identifier with a type variable, or \textit{metavariable}. Metavariables in \GEMINI\space are different from those in conventional programming languages in that there is a need to differentiate between metavariables of different kinds. A software metavariable may be later substituted by some software type, but not a hardware type, and vice versa.

\par
Each node in the AST is associated with information about its type, which is some variant of the discriminated union type \texttt{Absyn.ty} shown in full in \ref{sec:absyn-ty} of the Appendix. During the parsing phase, explicitly typed identifiers are associated with the given type, such as \texttt{Absyn.IntTy} for integer-typed identifiers, while implicitly typed identifiers are associated with a placeholder type \texttt{Absyn.PlaceholderTy}.

\par
In the decoration phase, the AST is reconstructed and each identifier is newly associated the variant \texttt{Absyn.ExplicitTy} which is constructed using a variant of the discriminated union type \texttt{Types.ty}, shown in full in \ref{sec:types-ty} of the Appendix. Thus, an identifier that was explicitly typed and associated with \texttt{Absyn.IntTy} during parsing would now have the type \texttt{Absyn.ExplicitTy(Types.S_TY(T.INT))}. Further, an identifier that had the type \texttt{Absyn.PlaceholderTy} would be given a new fresh metavariable using either the \texttt{Types.S_META} or \texttt{Types.H_META} constructor based on its kind. In some cases, the kind to which a metavariable belongs cannot be known yet, in which case it is temporarily given a new fresh metavariable of type \texttt{Types.META} and the kind is inferred later.

\subsubsection{Type inference}
After all identifiers are decorated with explicit types, it is time to perform type inference. Also called type reconstruction, this is the most complex phase of the \GEMINI\space compiler. There are two primary algorithms underlying the inference phase: \textit{unification} and \textit{substitution}.

\par
The goal of the unification algorithm is to compute the smallest possible substitution mapping $\sigma$ from metavariables to types. The unification algorithm is summarized in Figure \ref{fig:unification}.

\begin{figure*}
    \begin{tcolorbox}[title=\texttt{IsHWType}($t$),
    fonttitle=\bfseries]
    \begin{Verbatim}
case t of
    H_TY(_) => true
  | _ => false
    \end{Verbatim}

    \tcbsubtitle[before skip=\baselineskip]%
    {\texttt{IsSWType}($t$)}
    \begin{Verbatim}
case t of
    S_TY(_) => true
  | _ => false
    \end{Verbatim}
    
    \tcbsubtitle[before skip=\baselineskip]%
    {\texttt{Unify}($t_1$, $t_2$)}
    \begin{Verbatim}[
        commandchars=\\\{\},
        codes={%
            \catcode`$=3\relax
            \catcode`^=7\relax
            \catcode`_=8\relax
        }]
case $t_1$ of
    META(m) => $\{m \mapsto t_2\}$
  | H\string_META(hm) => if IsHWType($t_2$) then $\{hm \mapsto t_2\}$ else raise KindError
  | S\string_META(sm) => if IsSWType($t_2$) then $\{sm \mapsto t_2\}$ else raise KindError
  | \string_ => case $t_2$ of
            META(m) => $\{m \mapsto t_1\}$
          | H\string_META(hm) => if IsHWType($t_1$) then $\{hm \mapsto t_1\}$ else raise KindError
          | S\string_META(sm) => if IsSWType($t_1$) then $\{sm \mapsto t_1\}$ else raise KindError
          | \string_ => if IsHWType($t_1$) and IsHWType($t_2$)
                 then UnifyHWType($t_1$, $t_2$)
                 else if IsSWType($t_1$) and IsSWType($t_2$)
                      then UnifySWType($t_1$, $t_2$)
                      else raise KindError
    \end{Verbatim}

    \end{tcolorbox}
  \caption{Unification algorithm and subroutines (pseudocode)}
  \label{fig:unification}
\end{figure*}

\par
The subroutines \texttt{UnifyHWType} and \texttt{UnifySWType} are omitted from Figure \ref{fig:unification} for the sake of brevity. Both of these subroutines operate on the basis of structural recursion over the variants of each discriminated union. If the two types share the same outermost type, then the appropriate subroutine is recursively called on the inner types. The recursion terminates in three cases: (1) either the terminal types are known and match, (2) the terminal types are known and don't match, or (3) a metavariable is being unified with some other type. In the first case, nothing happens and unification continues. In the second case, there is a type mismatch and an error is raised. In the third case, a mapping is made from the metavariable to the other type.

\par
The result of the unification algorithm is a mapping $\sigma$ which is in turn used to augment a global substitution environment $\Sigma$. Since each metavariable is created freshly, it is safe to maintain a global environment since no two metavariables will correspond to the same element in the domain, and each metavariable can only map to a single element in the range; that is, the mapping function is injective.

\par
In addition to $\Sigma$, there are two more environments maintained although their scope is only local to their lexical closure. These are the type environment $\tau$ and the variable environment $\Gamma$. Within a let-binding, declarations bind symbols to their types in these environments. Values, functions, and modules are bound in $\Gamma$ whereas types and datatypes are bound in $\tau$. When processing a function or module, the parameters are added to the environment $\Gamma$ and only exist within the scope of the body. Since SML is a functional programming language, the implementation lends itself to discarding the augmented environments once the body has been processed, which correctly emulates the behavior of lexical scoping. As the AST is traversed, the mappings in $\Sigma$ are applied to both $\tau$ and $\Gamma$ in order to persist the results of unification. This constitutes the substitution algorithm, the main idea of which is summarized in Figure \ref{fig:substitution}.

\begin{figure*}
    \begin{tcolorbox}[title=\texttt{Substitute($\Sigma$, $env$)},
    fonttitle=\bfseries]
    \begin{Verbatim}[
        commandchars=\\\{\},
        codes={%
            \catcode`$=3\relax
            \catcode`^=7\relax
            \catcode`_=8\relax
        }]
$hasChanged \leftarrow$ false
$env^{\prime} \leftarrow env$
while $hasChanged$ do
    for ($name$, $type$) in $env$ do
        $env^{\prime} \leftarrow env^{\prime} \cup \{name \mapsto$ SubstituteType($type$, $\Sigma$, $hasChanged$)$\}$
return $env^{\prime}$
    \end{Verbatim}

    \tcbsubtitle[before skip=\baselineskip]%
    {\texttt{SubstituteType}($type$, $\Sigma$, $hasChanged$) (excerpt)}
    \begin{Verbatim}[
        commandchars=\\\{\},
        codes={%
            \catcode`$=3\relax
            \catcode`^=7\relax
            \catcode`_=8\relax
        }]
case $type$ of
    S\string_TY($stype$) => SubSW($type$, $\Sigma$, $hasChanged$, $\emptyset$)
  | H\string_TY($htype$) => SubHW($type$, $\Sigma$, $hasChanged$, $\emptyset$)
  | M\string_TY($mtype$) => SubMod($type$, $\Sigma$, $hasChanged$, $\emptyset$)
.
.
.
    \end{Verbatim}
    
    \tcbsubtitle[before skip=\baselineskip]%
    {\texttt{SubSW}($type$, $\Sigma$, $hasChanged$, $BV$) (excerpt)}
    \begin{Verbatim}[
        commandchars=\\\{\},
        codes={%
            \catcode`$=3\relax
            \catcode`^=7\relax
            \catcode`_=8\relax
        }]
case $type$ of
    S\string_META($sm$) => if $sm \in BV$
                    then $type$
                    else if $sm \in dom(\Sigma)$
                         then case $\Sigma(sm)$ of
                                    S\string_TY($type^{\prime}$) => (case $type^{\prime}$ of
                                                   S\string_META($sm^{\prime}$) => if $sm \neq sm^{\prime}$
                                                                  then $hasChanged \leftarrow$ true
                                                 | \string_ => $hasChanged \leftarrow$ true;
                                                  $type^{\prime}$)
                                  | \string_ => $type^{\prime}$
                         else $type$
  | INT => INT
  | ARROW($stype_1$, $stype_2$) => 
        ARROW(SubSW($stype_1$, $\Sigma$, $hasChanged$, $BV$), SubSW($stype_2$, $\Sigma$, $hasChanged$, $BV$))
  | S\string_POLY($PolyVars$, $ty^{\prime}$) =>
        S\string_POLY($PolyVars$, SubSW($ty^{\prime}$, $\Sigma$, $hasChanged$, $BV \cup PolyVars$))
  | S\string_MU($MuVars$, $ty^{\prime}$) => S\string_MU($MuVars$, SubSW($ty^{\prime}$, $\Sigma$, $hasChanged$, $BV \cup MuVars$))
.
.
.
    \end{Verbatim}

    \end{tcolorbox}
  \caption{Substitution algorithm and subroutines (pseudocode)}
  \label{fig:substitution}
\end{figure*}

\par
The definition of subroutines \texttt{SubHW} and \texttt{SubMod} are omitted, but are similar to that of \texttt{SubSW}. Some cases from \texttt{SubstituteType} and \texttt{SubSW} are also omitted, but the representative ones are shown. In substituting a metavariable, we first determine if it is a bound variable. If it is, then we must not substitute. If it is not, we look up the mapping in $\Sigma$ and return the mapped type if it exists. We must also make sure to let the iteration algorithm know if any substitution has occurred in order for it to continue iterating until it reaches a fixed point. Types like \texttt{INT} cannot be substituted any further and are therefore returned as is. For type constructors that possess inner types, such as \texttt{ARROW}, the \texttt{SubSW} routine is called recursively on the type arguments. The two most interesting cases are \texttt{S_POLY} and \texttt{S_MU}, which we will discuss further.

\par
The \texttt{S_POLY} type is inferred any time a function is parametrically polymorphic in its arguments. It represents the mathematical concept of the universal quantifier $\forall$. The set $PolyVars$ denotes the metavariables that are bound by the quantifier. As such, when performing substitution, these must not be substituted. Only upon function application does the \texttt{S_POLY} type become instantiated at which point the universal quantifier is shed, and each metavariable in $PolyVars$ is substituted uniformly with whichever type is supplied.

\par
Before discussing \texttt{S_MU}, we must first momentarily bring light to a special consideration made during the decoration phase. Since datatypes may be recursive, it is necessary to decorate their type uniquely when they are declared. The reason for this is twofold. First, while processing the body of the datatype there must exist some reference to the datatype itself since the constructor may be self-referential. In decorating datatype $d$, a temporary fresh metavariable $m$ is generated and the type environment $\tau$ is augmented with the mapping $\{d \mapsto m\}$. Then, the datatype constructors are decorated with any recursive reference to $d$ being replaced with the metavariable $m$. Once the entire datatype has been processed, the true discriminated union type $T_d$ can be determined and $\Sigma$ is augmented with the mapping $\{m \mapsto T_d\}$. Second, we wish to prevent infinite substitution from occuring in the inference phase if the datatype is recursive, and as such we wrap the type with the operator $\mu$, commonly known as $\mu$-recursion \cite{tarditi}. In substitution, whenever we encounter the $\mu$ operator, we refrain from substituting any variables it binds. Only when constructors are instantiated do we expand the recursive definition once, preventing infinite expansion.

\par
In inferring recursive functions, an approach similar to the handling of recursive datatypes is taken. Namely, the environment $\Gamma$ is augmented with a mapping from the function name to a function type with metavariables taking the place of parameter and return types. When processing the body, any application of the recursive function can be unified since the preliminary definition was polymorphic.

\par
Type inference enables parametric polymorphism by constraining types as loosely as possible. This allows a single part of a program to be used with different types. As an example, consider the following snippet of some language \titlesc{StrictGemini} which has neither type inference nor parametric polymorphism.

\begin{tcolorbox}
\begin{Verbatim}
fun concatInt
(x: int) (y: int list): int list = x::y

fun concatString
(x: string) (y: string list): string list
= x::y
\end{Verbatim}
\end{tcolorbox}

\par
The bodies of both functions are identical, yet we must declare them separately in order to be able to concatenate integers and strings. In the lambda calculus, the type of \texttt{concatInt} is $\lambda x:\mathtt{int}.\lambda y:\mathtt{int}$ $\mathtt{list}.x::y$, and the type of \texttt{concatString} is $\lambda x:\mathtt{string}.\lambda y:\mathtt{string}$ $\mathtt{list}.x::y$. For each further type, a separate concatenation function would need to be written. In \GEMINI, type construction enables us to instead define just the following.

\begin{tcolorbox}
\begin{Verbatim}
fun concat x y = x::y
\end{Verbatim}
\end{tcolorbox}

\par
Not only is the \GEMINI\space code less verbose, but it can be used to concatenate an element of any type to a list of the same type. By case 44 of the inversion of the typing relation, we know that \texttt{x} has some type \texttt{T$_S$} and \texttt{y} has some type \texttt{T$_S$} where \texttt{T$_S$} is the same in both cases. Since there are no further restrictions on the type \texttt{T$_S$}, the inference algorithm would find the loosest possible type which would be some metavariable that we will call \texttt{`a}. Therefore, our parametrically polymorphic function \texttt{concat} has the type signature \texttt{`a $\rightarrow$ `a list $\rightarrow$ `a list}. In the notation of the lambda calculus, the function is represented as $\forall \mathtt{`a}.\lambda x:\mathtt{`a}.\lambda y:\mathtt{`a}$ $\mathtt{list}.x::y$. When the function is applied, the quantifier is removed and the metavariable is substituted uniformly for the concrete type of the arguments.

\par
An optimization of the compiler is to gracefully handle type mismatches at this stage in order to allow the rest of the program to be type-checked without propagating the error forward. The way this is done is by using the \texttt{TOP} and \texttt{BOTTOM} types, and the software- and hardware-typed equivalents, found in \ref{sec:types-ty}. These act as the top and bottom types of the type system and are assigned to identifiers that fail to be unified during type inference; this prevents a single type mismatch from causing many errors if the identifier is used in many other expressions.

\subsubsection{Type checking}
With all types decorated and subsequently inferred, we are able to perform type-checking. The typing rules from Section \ref{sec:tyrules} are directly utilized in semantic verification. This is done in a recursive manner, since propositions of typing rules make statements about the types of subexpressions in order to verify the semantics of the entire expression. The recursion terminates once we reach a typing rule with no propositions, which indicates it is an axiom of our system.

\par
In our implementation, type-checking and type inference are actually performed concurrently. This is an optimization to avoid an unnecessary additional traversal of the AST, since the typing rules can be enforced once the types of subexpressions have been inferred. This is done by augmenting the unification algorithm from Figure \ref{fig:unification} in order to determine whether any two types can be unified, even if neither are metavariables. The way that this is done is by comparing the structure of the two types. If the types are both type constructors, then the unification algorithm recurses on the type arguments. The termination case of the recursive algorithm is when both types are proper types. If the types differ at any point, then there is a mismatch and an error is reported.

\subsection{Software evaluation}
\label{sec:sw-evaluation}
Once semantic analysis has been performed, we can safely begin the software evaluation phase knowing we have a well-typed program with regards to software expressions\footnote{Recall we can't always claim with certainty that the hardware terms are well-typed at this point, since their types may rely on the evaluation of software terms.}. Typically, compilers do not evaluate the program they are compiling and instead produce code in some target language that is functionally equivalent to the source program. However, in the case of the \GEMINI\space compiler there are two reasons we need to evaluate software terms. First, we wish to produce Verilog code which does not support certain software primitives. Second, we are required to evaluate software terms in order to perform hardware type-checking. As a result, this phase of compilation evaluates all software-typed expressions in order to generate an intermediate representation (IR) tree that consists solely of hardware-typed values, which we will later use to generate Verilog code.

\par
In our implementation, we evaluate according to the rules from Section \ref{sec:evalrules}. Having intentionally elected to define our rules by small-step semantics, the translation to implementation is much easier. For example, Figure \ref{fig:eval-conditionals} demonstrates how conditional expressions are evaluated in this phase\footnote{The code here is modified slightly for ease of readability}.

\begin{figure*}
    \begin{tcolorbox}
    \begin{Verbatim}
evalExp(Absyn.IfExp{guard, then', else', pos}) =
    let
        val guardVal = evalExp(guard)
    in
        if (getInt(guardVal)) <> 0
        then evalExp(then')
        else evalExp(else')
    end
    \end{Verbatim}
    \end{tcolorbox}
  \caption{Evaluating conditional expressions}
  \label{fig:eval-conditionals}
\end{figure*}

\par
As can be seen on line 3 of Figure \ref{fig:eval-conditionals} the guard is first evaluated recursively until a value is achieved. On line 5 we retrieve the integer constant for that value and compare it to \texttt{0}: if it is non-zero, then we evaluate the then-clause, else we evaluate the else-clause. Compare this strategy of evaluation to the rules \titlesc{E-IfThen}, \titlesc{E-IfThen-T}, and \titlesc{E-IfThen-F} and notice how they are aligned.

\par
While performing evaluation, a value store is maintained that maps symbols to the values they possess, specifically utilizing the \texttt{Value.value} datatype shown in full in \ref{sec:value-value} of the Appendix. Any time an identifier is referenced, its assigned value is looked up in the value store. This naturally enables lexical scoping since the value store within an inner scope is discarded once that scope is exited.

\par
Evaluating function and module declarations warrants special discussion. When a function is declared, its name is bound in the value store to a \texttt{Value.value $\rightarrow$ Value.value} function. When this function is called, the argument \texttt{Value.value} is bound to the function parameter names and the body is evaluated with the augmented value store. The reason we bind to a function as opposed to simply binding the function name to the function body is because the values of the arguments are only known upon function application. SML's closure rules enable the value function to remember the state of the value store upon declaration and process the body correctly after augmenting it with the supplied parameters.

\par
Module declarations are complicated by one more consideration. If modules are instantiated in the program itself, a similar approach can be taken in order to expand the module body inline. However, the top-level module that is returned is not instantiated but needs to be expanded in order to generate the final Verilog code. This is done by capturing the argument names when the module is declared and storing it as part of the module value; this is exhibited by the \texttt{Value.ModuleVal} variant of the datatype in \ref{sec:value-value}. At the top-level, the argument names are applied to the module function to expand the module body with all appearances of the argument variables replaced with \texttt{Value.NamedVal}, which represents some input pin the output Verilog module.

\subsection{Code generation}
At the final phase of compilation, we are left with a tree of hardware values representing a circuit constituted only of bits, gates, arrays, records, and pins.

\par
First, an additional pass of hardware type-checking must be performed. The manner in which this is done is similar to the process for software type-checking. As an optimization, this is done concurrently with the production of Verilog instead of performing an extra traversal.

\par
The manner in which Verilog is produced is also similar to previous phases. We recurse over the AST and build the output from the results of subtrees. The base case of recursion is when an input pin or constant bit value is encountered, in which case the appropriate symbol or constant is returned. At each node of the hardware tree, the subexpressions are evaluated to determine the wire that holds their value, and then the node itself is evaluated. A fresh variable name is generated and returned for superexpressions to use in computing themselves. For each generated instruction, the type of the wire is noted and an appropriate declaration is made at the top of the module. The result is a series of declarations and instructions which are finally emitted in order to produce the output.

\par
Hardware records are treated specially during this phase\footnote{Tuples are treated in the same way but since they are a derivation of records, they are not mentioned here.}. Verilog does not support record types since all data must be expressed in terms of bits and arrays. As a result, records are encoded into arrays by contanating the values of fields. Similarly, when accessing a record, the appropriate range of bits is determined and the array-transformed record is decoded in order to retrieve the appropriate field.

\section{Examples}
\label{sec:example}
We will now present some examples of \GEMINI\space programs to reinforce some of the concepts described throughout this paper. Each of the examples gets increasingly more practical in demonstrating how \GEMINI\space can be effectively used to generate Verilog code and we will see the output of a \GEMINI\space program in the final example.

\subsection{Canonical Functions}
The first example we observe exemplifies that the expressive power of \GEMINI\space is comparable to that of modern programming languages. In Figure \ref{fig:canonical}, we see some \GEMINI\space code for the canonical list functions of the functional programming paradigm. Each of these are defined using pattern-matching with structural decomposition and recursion. We also see how discriminated union types can be used to create the parametrically polymorphic \texttt{option} datatype.

\begin{figure*}
    \begin{tcolorbox}
    \begin{Verbatim}
fun map f x = case x of
                   [] => []
                |: a::rest => (f a)::(map f rest)

fun filter f x = case x of
                      [] => []
                   |: a::rest => if (f a)
                                 then a::(filter f rest)
                                 else (filter f rest)

fun foldl f acc init = case init of
                            [] => acc
                         |: (x::rest) => (foldl f (f(x, acc)) rest)
                         
fun foldr f acc init = case init of
                            [] => acc
                         |: (x::rest) => f(x, (foldr f acc rest))
                         
sdatatype 'a option = SOME of 'a
                   |: NONE

fun mapPartial f x = case x of
                          [] => []
                       |: (a::rest) => (case (f a) of
                                             NONE => mapPartial f rest
                                          |: SOME v => v::(mapPartial f rest))
    \end{Verbatim}
    \end{tcolorbox}
  \caption{Canonical list functions in \GEMINI}
  \label{fig:canonical}
\end{figure*}

\subsection{Explicit Logic}
While software functions like the ones shown in Figure \ref{fig:canonical} are useful, they are not the main point of \GEMINI. Software constructs are a means to an end used to make it easier to describe hardware, but ultimately a \GEMINI\space program must return some circuit that can be represented in Verilog.

\par
In this example, we showcase a \GEMINI\space program that makes use of helpful features such as pattern-matching, recursion, and discriminated union types in order to build a system of making explicit logic declarations. Figure \ref{fig:explicit-logic} shows this program.

\begin{figure*}
    \begin{tcolorbox}
    \begin{Verbatim}
let
    sdatatype explicitLogic = AND of explicitLogic list
                           |: OR of explicitLogic list
                           |: XOR of explicitLogic list
                           |: NOT of explicitLogic
                           |: INP of bit sw

    fun NAND lst = NOT(AND lst)
    fun NOR lst = NOT(OR lst)
    
    fun toHW expl =
        let
            fun listToArray elist =
                sw #[List.length elist; gen i => (unsw (toHW (List.nth(elist, i))))]
        in
            case expl of
                 INP(b) => b
              |: AND(lst) => sw (&-> (unsw (listToArray lst)))
              |: OR(lst) => sw (|-> (unsw (listToArray lst)))
              |: XOR(lst) => sw (^-> (unsw (listToArray lst)))
              |: NOT(el) => sw (! (unsw (toHW el)))
        end

    module mycircuit #(a, b, c) =
        let
            (* equivalent to if we had written !(c ^ (a & b)) *)
            val temp = NOT(XOR[INP(sw c), AND [INP(sw a), INP(sw b)]])
        in
            unsw (toHW temp)
        end
in
    mycircuit
end
    \end{Verbatim}
    \end{tcolorbox}
  \caption{Explicit Logic in \GEMINI}
  \label{fig:explicit-logic}
\end{figure*}

\par
It is worth paying attention to the interplay between software- and hardware-typed values in this example, and how \texttt{sw} and \texttt{unsw} are strategically used to convert from one kind to the other. This program exemplifies the way in which hardware values can be passed around in software constructs, such as functions and datatypes; their values are never read or accessed, but operations may be applied to them.

\par
One may imagine lines 2 through 22 being defined in a library, with \texttt{toHW} exposed for use by other programs. The function \texttt{toHW} takes an expression of type \texttt{explicitLogic} and returns a value of type \texttt{bit}. It does so by matching the expression against the variant types, each time recursively applying itself to the inner values except in the case of the base variant \texttt{INP}. We can see on line 21, for example, that if the expression has the variant type \texttt{NOT} then \texttt{toHW} first calls itself recursively on the inner value which it then unwraps, applies with the bitwise-not operator, and rewraps. The unwrapping and rewrapping are necessary to adhere to the kinding system since the bitwise-not operator can only be applied to values with the type \texttt{bit} which is in \HWKIND. In the production of Verilog, we collapse the wrap operators and apply the operators as specified.

\par
The module defined on line 24 takes three arguments, each of type \texttt{bit}\footnote{This is not explicitly specified, but it can be inferred since they are each wrapped with \texttt{sw} and passed as an argument to the \texttt{INP} constructor. Since the \texttt{INP} constructor accepts a \texttt{bit sw}, it must follow from our unification algorithm that each has the type \texttt{bit}.}. A logic expression is then explicitly written using the variants from \texttt{explicitLogic} and the result is converted to the hardware circuit equivalent using \texttt{toHW} and unwrapped in order to reveal the hardware value. As the comment on line 26 explains, this module definition is equivalent to if we had written \texttt{!(c \textasciicircum (a \& b))}. However, the point of this example is to demonstrate how we can introduce greater flexibility in a few ways. First, we allow an arbitrary number of inputs to be passed to the operators as a list. This is useful if we wish to construct a hardware circuit from some external source, such as a text file or as command line arguments. We can imagine writing interpreters of certain specifications to read in files and generate circuits by leveraging our explicit logic generator. Further, we are able to define composite operators from our primitive ones. We have done this for \texttt{NAND} and \texttt{NOR} on lines 8 and 9, though we may define more complex operators as well.

\subsection{N-bit RCA}
The final example will demonstrate the implementation of an n-bit ripple carry adder (RCA) circuit as a \GEMINI\space program. A difficulty of Verilog is that modules cannot be parametrically polymorphic in the size of inputs. As a result a circuit that requires a certain module, such as an RCA, to be instantiated on inputs of different sizes must duplicate the definition of these modules for each desired input size. From the perspective of software design, there are several problems with this approach. First, the same circuit definition must be made repeatedly which increases the overall size of the module definition for which readability and maintainability suffer. Second, any changes that need to be made to the general circuit algorithm must be made to each definition of the module. Third, as the input sizes increase, the complexity of the circuit may too increase exponentially making it unfeasible to program without the possibility of introducing subtle errors.

\par
The n-bit RCA shown in the program in Figure \ref{fig:rca} demonstrates how \GEMINI\space can be used to define an RCA module that is parametrically polymorphic in input size.

\begin{figure*}
    \begin{tcolorbox}
    \begin{Verbatim}
let
    val numbits = (* ... *)
    module rca_helper #(ai : bit, bi : bit, cin : bit) =
        let
            val sum = ai ^ bi ^ cin
            val cout = (ai & bi) | (ai & cin) | (bi & cin)
        in
            #(cout, sum)
        end

    fun getSecond x = (sw #2(unsw x))

    module rca <:n:> #(a, b) =
        let
            val couts = #[n; gen i =>
                                if i = 0 then
                                    rca_helper #(a[:i:], b[:i:], 'b:0)
                                else
                                    rca_helper #(a[:i:], b[:i:], #1(couts[:i - 1:]))]
        in
            unsw Array.fromList(List.map getSecond (Array.toList(sw couts)))
        end
    module n_bit_rca #(a : bit[numbits], b : bit[numbits]) = rca <:numbits:> #(a, b)
in
    n_bit_rca
end
    \end{Verbatim}
    \end{tcolorbox}
  \caption{N-bit RCA in \GEMINI}
  \label{fig:rca}
\end{figure*}

\par
On line 2, we define the desired size of the input arrays to be added. This can be defined as a constant, the result of some expression, or an input from a text file or command line. The module \texttt{rca} defined on line 13 is parameterized by the symbol \texttt{n} which is an integer referred to in the body, namely to define the size of the array in the generation expression on line 15. While we cannot return a parametrically polymorphic module, we can instantiate it to be used by another module which is what we have done on line 23. The result is an n-bit RCA module named \texttt{n_bit_rca} that takes two bit arrays of equal size and produces the result of addition using the ripple carry method. We can see the Verilog that is produced as a result of compiling this \GEMINI\space program with \texttt{numbits = 2} in Figure \ref{fig:rca-output}\footnote{The Verilog module name is taken from the name of the \GEMINI\space file from which it is compiled.}. 

\begin{figure*}
    \begin{tcolorbox}
    \begin{Verbatim}
module adder(input [1:0] a, input [1:0] b, output reg [1:0] out);
    reg r9, r13, r10, r12, r11, r2, r6, r8, r7, r3, r5, r4;
    reg [1:0] r1;

    always @(*) begin
        r4 <= a[1];
        r5 <= b[1];
        r3 <= r4 ^ r5;
        r7 <= a[0];
        r8 <= b[0];
        r6 <= r7 & r8;
        r2 <= r3 ^ r6; r11 <= a[0];
        r12 <= b[0];
        r10 <= r11 ^ r12; r13 <= 1'b0;
        r9 <= r10 ^ r13; r1[1] <= r2; r1[0] <= r9;
        out <= r1;
    end
endmodule
    \end{Verbatim}
    \end{tcolorbox}
  \caption{2-bit RCA in Verilog, compiled from \GEMINI}
  \label{fig:rca-output}
\end{figure*}

\section{Future Work}
\label{sec:future}
We conclude with some remarks on the various ways in which to improve upon and extend the work presented in this paper.

\subsection{Improvements}
While the \GEMINI\space compiler presented in this paper makes a few optimizations to reduce compilation time, there is additional scope for improvements to the output Verilog code. We list some of these in this Section.

\par
\textbf{\textit{Dataflow optimizations.}} The analysis of dataflow is common in compiler optimization as a way to determine the flow of information through a program. In particular, available expression analysis is a useful kind of dataflow optimization in the context of \GEMINI\space in order to reduce the number of logical gates used in the final program. The analysis of available expressions can be used in order to carry out common subexpression elimination by determining expressions that need not be recomputed and reusing their results after the initial computation instead of computing them again. This would also have the effect of reducing the number of intermediary wires declared for a given module.

\par
\textbf{\textit{Constraining resources finitely.}} Along a similar line, it may be the case that the circuit to be designed realistically has specific limitations on the number of resources, such as wires and logic gates, that are available for use. As such, the \GEMINI\space compiler can be instructed to constrain resource allocation to a certain amount. Dataflow optimizations such as common subexpression elimination are one way to reduce resource consumption, though other methods such as register allocation can be leveraged in order to deal with a finite number of available resources.

\par
\textbf{\textit{Improved legibility.}} Currently, a compiled \GEMINI\space program produces a single Verilog module. Any intermediate modules defined within a \GEMINI\space program are effectively expanded inline in the final module that is produced. For large and complex circuits, this can negatively impact modularity and readability by producing monolithic modules that are hard to test and debug. A useful improvement to the \GEMINI\space compiler would be to enable multiple modules to be produced from a single \GEMINI\space program. One possible mechanism to allow this to happen is to mark certain modules as "persistent", which would prevent the compiler from expanding them inline and instead create them as a separate module that can be referred from the top-level module in the resulting Verilog file. Another issue with legibility is that the names of wires are generated automatically, which makes it difficult to trace the origins of a particular logical operation to the source \GEMINI\space code. A valuable improvement to the \GEMINI\space compiler would be to allow programmers to optionally name the results of certain logical operations such that the produced Verilog code is more understandable.

\subsection{Extensions}

\par
\textbf{\textit{Testbench generation.}} An important aspect of designing electronic circuits is the ability to test its correctness. Popular programmable logic design software such as Quartus support the simulation of Verilog testbenches. The testbench is a Verilog file that defines a simulation of a given module against certain inputs, intended to ascertain the correct behavior. There is therefore scope for an extension to the \GEMINI\space compiler in which a programmer can write \GEMINI\space testbenches or individual test cases that are used to produce a Verilog testbench to accompany the produced module. These can then both be used in software such as Quartus to verify that the module behaves as expected.

\par
\textbf{\textit{\GEMINI\space signature files.}} The concept of interfaces is borrowed from popular modern programming languages. In SML and OCaml, these are named \textit{signatures}. They act by limiting the ways in which other programs are able to interact with each other, keep some of the implementation details private to a given program. In general, this does not have a significant effect on the quality of the produced Verilog module, though it can lead to improved software design practices when dealing with projects of a large scale.

\par
\textbf{\textit{Backends for other HDLs.}} As mentioned in Section \ref{sec:sw-evaluation}, the software evaluation phase of compilation produces an intermediate representation of our program consisting solely of hardware-typed values. The compiler was designed in this way intentionally to improve modularity; by generating an IR, the frontend and backend of the compiler are largely independent. In the compiler presented in this paper, the IR was used to generate a Verilog program. However, additional backends may be developed in order to produce modules of other popular HDLs, such as VHDL. These backends can use the IR produced by the existing frontend thereby halving the amount of work needed to compile to a new target language.

\bibliographystyle{unsrt}

\end{multicols}

\appendix
\section{Language Specification}
\subsection{Hardware Operators}
\label{sec:hwops}
\begin{threeparttable}
\begin{tabularx}{\linewidth}{llX}  
\toprule
Operator        & Syntax                & Semantic Result \\
\midrule
\texttt{\&}     & \texttt{e1 \& e2}     & Bitwise logical "and" of composing bits\tnote{1}\\
\texttt{|}      & \texttt{e1 | e2}      & Bitwise logical "or" of composing bits\tnote{1}\\
\texttt{\textasciicircum}     & \texttt{e1 \textasciicircum e2}     & Bitwise logical "xor" of composing bits\tnote{1}\\
\texttt{!}      & \texttt{!e}           & Negation of bit operand\\
\texttt{<{}<}     & \texttt{e1 <{}< e2}     & Shifts left bit array operand to the left by the amount specified (as an unsigned integer) by the right bit array\\
\texttt{>{}>}     & \texttt{e1 >{}> e2}     & Shifts left bit array operand to the right by the amount specified (as an unsigned integer) by the right bit array\\
\texttt{>{}>{}>}    & \texttt{e1 >{}>{}> e2}    & Shifts left bit array operand to the right by the amount specified (as an unsigned integer) by the right bit array, filling with the most significant bit value\\
\texttt{\&->}   & \texttt{\&->e1}       & Bitwise and-reduction of bit array operand\\
\texttt{|->}    & \texttt{|->e1}        & Bitwise or-reduction of bit array operand\\
\texttt{\textasciicircum->}    & \texttt{\textasciicircum->e1}        & Bitwise xor-reduction of bit array operand\\
\texttt{\&\&}    & \texttt{e1 \&\& e2}  & Or-reduction of both bit array operands, followed by bitwise logical "and" of resulting bits\\
\texttt{||}    & \texttt{e1 || e2}  & Or-reduction of both bit array operands, followed by bitwise logical "or" of resulting bits\\
\texttt{\textasciicircum\textasciicircum}    & \texttt{e1 \textasciicircum\textasciicircum e2}  & Or-reduction of both bit array operands, followed by bitwise logical "xor" of resulting bits\\
\bottomrule
\end{tabularx}
\begin{tablenotes}\footnotesize
\item [1]These operators recurse through subexpressions to perform the bitwise operation on the composing bits while retaining the overall structure.
\end{tablenotes}
\end{threeparttable}

\subsection{Arithmetic Operators}
\begin{tabularx}{\linewidth}{llX}  
\toprule
Operator        & Syntax                & Semantic Result \\
\midrule
\texttt{\~}      & \texttt{\textasciitilde e}           & Negation of integer operand\\
\texttt{+}      & \texttt{e1 + e2}      & Addition of integer operands\\
\texttt{-}      & \texttt{e1 - e2}      & Subtraction of right integer operand from left integer operand\\
\texttt{/}      & \texttt{e1 / e2}      & Division of left integer operand by right integer operand, rounded towards negative infinity\\
\texttt{*}      & \texttt{e1 * e2}      & Multiplication of integer operands\\
\texttt{\%}     & \texttt{e1 \% e2}     & Modulo of dividend left integer operand with divisor right integer operand\\
\texttt{+.}     & \texttt{e1 +. e2}     & Addition of real operands\\
\texttt{-.}     & \texttt{e1 -. e2}     & Subtraction of right real operand from left real operand\\
\texttt{/.}     & \texttt{e1 /. e2}     & Division of left real operand by right real operand\\
\texttt{*.}     & \texttt{e1 *. e2}     & Multiplication of real operands\\
\bottomrule
\end{tabularx}

\subsection{Conditional Operators}
\begin{tabularx}{\linewidth}{llX}  
\toprule
Operator        & Syntax                & Semantic Result \\
\midrule
\texttt{andalso}    & \texttt{e1 andalso e2}    & Logical conjunction of both integer operands\\
\texttt{orelse}     & \texttt{e1 orelse e2}     & Logical disjunction of both integer operands\\
\texttt{not}     & \texttt{not e}               & Logical complementation of integer operand\\
\bottomrule
\end{tabularx}

\subsection{Comparison Operators}
\begin{tabularx}{\linewidth}{llX}  
\toprule
Operator        & Syntax                & Semantic Result \\
\midrule
\texttt{=}    & \texttt{e1 = e2}    & Equality of operands\\
\texttt{<>}     & \texttt{e1 <> e2}     & Non-equality of operands\\
\texttt{>}     & \texttt{e1 > e2}               & Left operand has a strictly greater order than right operand\\
\texttt{<}     & \texttt{e1 < e2}               & Left operand has a strictly lesser order than right operand\\
\texttt{>=}     & \texttt{e1 >= e2}               & Left operand has a greater or equal order compared to right operand\\
\texttt{<=}     & \texttt{e1 <= e2}               & Left operand has a lesser or equal order compared to right operand\\
\bottomrule
\end{tabularx}

\subsection{List Operators}
\label{sec:lstops}
\begin{tabularx}{\linewidth}{llX}  
\toprule
Operator        & Syntax                & Semantic Result \\
\midrule
\texttt{::}    & \texttt{e1::e2}    & Concatenation of left element operand to the beginning of right list operand\\
\bottomrule
\end{tabularx}

\subsection{Library Functions}
\label{sec:lib}
\begin{tabularx}{\linewidth}{lllX}  
\toprule
Structure       & Function      & Type      & Semantic Result \\
\midrule
\multirow{2}{*}{\texttt{Core}} & \texttt{print} & \texttt{string -> unit} & Write a string to the standard output\\
\cline{2-4}
                               & \texttt{read} & \texttt{string -> string} & Read the contents of a file\\
\hline
\multirow{7}{*}{\texttt{List}} & \texttt{nth} & \texttt{(`a list * int) -> `a} & Return an element from a list given an index; raises an exception if the index is out of bounds\\
\cline{2-4}
                               & \texttt{length} & \texttt{`a list -> int} & Return the length of a list\\
\cline{2-4}
                               & \texttt{rev} & \texttt{`a list -> `a list} & Return the reversed list\\
\cline{2-4}
                               & \texttt{map} & \texttt{(`a -> `b) -> `a list -> `b list} & Apply a mapping function to each element of a list and return the result\\
\cline{2-4}
                               & \texttt{filter} & \texttt{(`a -> int) -> `a list -> `a list} & Return a list containing only elements that satisfy the predicate function\\
\cline{2-4}
                               & \texttt{foldl} & \texttt{(`a * `b -> `b) -> `b -> `a list -> `b} & Accumulate a value by iterating over a list from left to right\\
\cline{2-4}
                               & \texttt{foldr} & \texttt{(`a * `b -> `b) -> `b -> `a list -> `b} & Same as \texttt{foldl} except iteration is from right to left\\
\hline
\texttt{Int} & \texttt{toString} & \texttt{int -> string} & Return a string representation of an int\\
\hline
\multirow{4}{*}{\texttt{String}} & \texttt{size} & \texttt{string -> int} & Return the number of characters in a string\\
\cline{2-4}
                & \texttt{substring} & \texttt{(string * int * int) -> string} & Return the substring from a start to end location of a string; raises an exception if either index is out of bounds\\
\cline{2-4}
                & \texttt{concat} & \texttt{string list -> string} & Return the concatenation of all strings in a list\\
\cline{2-4}
                & \texttt{split} & \texttt{string -> string -> string list} & Return a list of strings resulting from splitting an original string over some delimiter\\
\hline
\multirow{5}{*}{\texttt{Real}}   & \texttt{floor} & \texttt{real -> int} & Return a real rounded towards negative infinity\\
\cline{2-4}
                & \texttt{ceil} & \texttt{real -> int} & Return a real rounded towards positive infinity\\
\cline{2-4}
                & \texttt{round} & \texttt{real -> int} & Return a real rounded towards the closest integer\\
\cline{2-4}
                & \texttt{fromInt} & \texttt{int -> real} & Return a real converted from an integer\\
\cline{2-4}
                & \texttt{toString} & \texttt{real -> string} & Return a string representation of a real\\
\hline
\multirow{2}{*}{\texttt{Array}}   & \texttt{toList} & \texttt{`a[n] sw -> `a sw list} & Return a list of software-wrapped hardware values from a software-wrapped hardware array\\
\cline{2-4}
                & \texttt{fromList} & \texttt{`a sw list -> `a[n] sw} & Return a software-wrapped hardware array from a list of software-wrapped hardware values\\
\hline
\texttt{BitArray}  & \texttt{twosComp} & \texttt{bit[n] ~> bit[n]} & Return a circuit performing twos-complement of a bit array\\
\hline
\texttt{HW}     & \texttt{dff} & \texttt{`a @ n ~> `a @ (n + 1)} & Return a DFF circuit from a given hardware input\\
\bottomrule
\end{tabularx}

\subsection{Operator Order of Precedence}
\label{sec:order-of-precedence}
\begin{tabularx}{\linewidth}{ll}  
\toprule
Operator(s)        & Associativity \\
\midrule
\texttt{\~}, \texttt{!}, \texttt{|->}, \texttt{\&->}, \texttt{\textasciicircum->} & N/A\\
\texttt{\$} & N/A\\
\texttt{\#f} & N/A\\
\texttt{[:i:]} & N/A\\
\texttt{/.}, \texttt{*.}, \texttt{/}, \texttt{*}, \texttt{\&}, \texttt{\%} & left\\
\texttt{-.}, \texttt{+.}, \texttt{-}, \texttt{+}, \texttt{\textasciicircum}, \texttt{|} & left\\
\texttt{\&\&} & left\\
\texttt{||}, \texttt{\textasciicircum\textasciicircum} & left\\
\texttt{::} & right\\
\texttt{>}, \texttt{<}, \texttt{>=}, \texttt{<=} & left\\
\texttt{=}, \texttt{<>} & left\\
\texttt{<{}<}, \texttt{>{}>}, \texttt{>{}>{}>} & left\\
\texttt{andalso} & left\\
\texttt{orelse} & left\\
\texttt{:=} & right\\
\bottomrule
\end{tabularx}

\section{Formalizations}
\subsection{Derived Terms}
\label{sec:derivedterms}
\begin{tabularx}{\linewidth}{ll}  
\toprule
Name        & Equivalence \\
\midrule
tuple & $(e_i)^{i \in 1..n} \equiv \{i = e_i\}^{i \in 1..n}$\\
unit & \texttt{() $\equiv$ \{\}}\\
logical and & \texttt{e1 andalso e2 $\equiv$ if e1 then e2 else 0}\\
logical or & \texttt{e1 orelse e2 $\equiv$ if e1 then 1 else e2}\\
logical not & \texttt{not e $\equiv$ if e1 then 0 else 1}\\
and-reduction & \texttt{\&->\#[e$_i$]$^{i \in 1..n}$ $\equiv$ e$_1$ \& ... \& e$_n$}\\
or-reduction & \texttt{|->\#[e$_i$]$^{i \in 1..n}$ $\equiv$ e$_1$ | ... | e$_n$}\\
xor-reduction & \texttt{\textasciicircum->\#[e$_i$]$^{i \in 1..n}$ $\equiv$ e$_1$ \textasciicircum ... \textasciicircum e$_n$}\\
and-collapse & \texttt{e1 \&\& e2} $\equiv$ \texttt{(|->e1) \& (|->e2)}\\
or-collapse & \texttt{e1 \&\& e2} $\equiv$ \texttt{(|->e1) | (|->e2)}\\
xor-collapse & \texttt{e1 \&\& e2} $\equiv$ \texttt{(|->e1) \textasciicircum (|->e2)}\\
if-then & \texttt{if e1 then e2 $\equiv$ if e1 then e2 else \{\}}\\
sequence & \texttt{(e1; e2) $\equiv$ $(\lambda x:T.e_2) e_1$ where $x \notin FreeVar(e_2)$}\\
\bottomrule
\end{tabularx}

\subsection{Software Value Grammar}
\label{sec:swvgrammar}
\begin{grammar}
    <swv> ::= <integer>
    \alt <real>
    \alt <string>
    \alt <list>
    \alt <software record>
    \alt <sw>
    \alt <ref>
    \alt <variant>
    \alt <function>
    
    <integer> ::= $i \in \mathbb{Z} \cap [-2^{31}, 2^{31} - 1]$
    
    <real> ::= $r \in \mathbb{R} \cap [2^{-1074}, (2 - 2^{-52}) \times 2^{1023}] \cap [-(2 - 2^{-52}) \times 2^{1023}, -2^{-1074}]\: \cap $ \{numbers expressible as IEEE FP\}
    
    <string> ::= $s \in \textstyle \bigcup\limits_{i = 0}^{2^{63} - 1}A^{i}$ where $A$ is the ASCII alphabet and $A^i$ denotes a sequence of $i$ characters from the alphabet $A$
    
    <list> ::= [<swv>$_i^{i \in 0..n-1}$]
    
    <software record> ::= $\{l_i = \langle swv \rangle_i^{i \in 1..n}\}$
    
    <sw> ::= $\omega$(<hwv>)
    
    <ref> ::= $\mu$[$l \mapsto $ <swv>]
    
    <variant> ::= $C_i$ $\langle swv \rangle$
    
    <function> ::= $\lambda x : T_S.e$
    
\end{grammar}

\subsection{Software Term Grammar}
\label{sec:swtgrammar}
\begin{grammar}
    <exp> ::= <literal>
    \alt <access>
    \alt <let binding>
    \alt <conditional>
    \alt <operation>
    \alt <assignment>
    \alt <pattern match>
    \alt <sequence>
    \alt <application>
    
    <literal> ::= <identifier>
    \alt <integer literal>
    \alt <real literal>
    \alt <string literal>
    \alt <list literal>
    \alt <software record literal>
    \alt <ref literal>
    \alt <sw literal>
    
    <identifier> ::= <id-start> <id-tail>
    
    <id-start> ::= \{any alphabetic character or underscore\}
    
    <id-tail> ::= \{any alphanumeric character or underscore\} <id-tail> 
    \alt $\epsilon$

    <integer literal> ::= <binary-integer>
    \alt <octal-integer>
    \alt <decimal-integer>
    \alt <hex-integer>
    
    <binary-integer> ::= "#\'b:" <binary-digits>
    
    <octal-integer> ::= "#\'o:" <octal-digits>
    
    <decimal-integer> ::= <decimal-digits>
    \alt <sign> <decimal-digits>

    <hex-integer> ::= "#\'x:" <hex-digits>
    
    <real literal> ::= <real-tail>
    \alt <sign> <real-tail>
    
    <real-tail> ::= <decimal-digits> "." \\
                    <decimal-digits-or-empty> <exponent-or-empty>
    \alt <decimal-digits-or-empty> "." \\
         <decimal-digits>\\
         <exponent-or-empty>
    \alt <decimal-digits> <exponent>
    
    <decimal-digits-or-empty> ::= <decimal-digits>
    \alt $\epsilon$
    
    <exponent> ::= "E" <decimal-digits>
    \alt "E" <sign> <decimal-digits>
    \alt "e" <decimal-digits>
    \alt "e" <sign> <decimal-digits>
    
    <exponent-or-empty> ::= <exponent>
    \alt $\epsilon$
    
    <binary-digit> ::= any of \{0, 1\}
    
    <octal-digit> ::= any of \{0, 1, 2, 3, 4, 5, 6, 7\}
    
    <decimal-digit> ::= any of \{0, 1, 2, 3, 4, 5, 6, 7, 8, 9\}
    
    <hex-digit> ::= any of \{0, 1, 2, 3, 4, 5, 6, 7, 8, 9, a, b, c, d, e, f, A, B, C, D, E, F\}
    
    <binary-digits> ::= <binary-digits> <binary-digit>
    \alt <binary-digit>
    
    <octal-digits> ::= <octal-digits> <octal-digit>
    \alt <octal-digit>
    
    <decimal-digits> ::= <decimal-digits> <decimal-digit> 
    \alt <decimal-digit>
    
    <hex-digits> ::= <hex-digits> <hex-digit>
    \alt <hex-digit>
    
    <sign> ::= "~"
    
    <string literal> ::= "\"" <string-body> "\""
    
    <string-body> ::= <string-body> <string-character>
    \alt $\epsilon$
    
    <string-character> ::= \{any printable character, including space, except for double-quotes ("\"") and backslash ("\\")\}
    \alt "\\" <escape-character>
    
    <escape-character> ::= any of \{"\\", "\'", "\"", a, b, e, f, n, r, t, 0\}
    \alt <hex-digits>
    
    <list literal> ::= "[" <list-body> "]"
    \alt "nil"
    
    <list-body> ::= <exp> <exp comma tail>
    \alt $\epsilon$
    
    <exp comma tail> ::= "," <exp> <exp comma tail>
    \alt $\epsilon$
    
    <software record literal> ::= "\{" <record-body> "\}"
    
    <record-body> ::= <identifier> "=" <exp> <record-tail>
    \alt $\epsilon$
    
    <record-tail> ::= "," <identifier> "=" <exp> <record tail>
    \alt $\epsilon$
    
    <ref literal> ::= "ref" <exp>
    
    <sw literal> ::= "sw" <hwv>
    
    <access> ::= <struct-access>
    \alt <record-access>
    \alt <ref-access>
    \alt <tuple-access>
    
    <struct-access> ::= <identifier> "." <identifier>
    
    <record-access> ::= "#" <identifier> <exp>
    
    <ref-access> ::= "\$" <exp>
    
    <tuple-access> ::= "#" <decimal-integer> <exp>
    
    <let binding> ::= "let" <decs> "in" <exp> "end"
    
    <decs> ::= <val-dec> <decs>
    \alt <ty-dec> <decs>
    \alt <sdataty-dec> <decs>
    \alt <fun-dec> <decs>
    \alt $\epsilon$

    <val-dec> ::= "val" <identifier> "=" <exp>
    \alt "val" <identifier> ":" <ty> "=" <exp>
    
    <ty-dec> ::= "type" <tyvars> <identifier> "=" <ty>
    
    <tyvars> ::= <tyvar>
    \alt "(" <tyvar> <tyvars tail> ")"
    \alt $\epsilon$
    
    <tyvars tail> ::= "," <tyvar> <tyvars tail>
    \alt $\epsilon$
    
    <ty> ::= <tyvar>
    \alt <identifier>
    \alt "{" <ty-fields> "}"
    \alt <ty> "list"
    \alt <ty> "ref"
    \alt <ty> "sw"
    \alt <ty> "->" <ty>
    \alt "(" <ty> ")"
    
    <tyvar> ::= "'" <identifier>
    
    <ty-fields> ::= <ty-fields-tail>
    \alt $\epsilon$
    
    <ty-fields-tail> ::= <identifier> ":" <ty> 
    \alt <ty-fields-tail> "," <identifier> ":" <ty> 
    
    <sdataty-dec> ::= "sdatatype" <tyvars> <identifier> "=" <identifier> "of" <ty> <dataty-tail>
    
    <dataty-tail> ::= "|:" <identifier> "of" <ty> <dataty-tail>
    \alt $\epsilon$
    
    <fun-dec> ::= "fun" <identifier> <fun-params> "=" <exp>
    \alt "fun" <identifier> <fun-params> ":" <ty> "=" <exp>
    
    <fun-params> ::= <fun-params> <fun-param>
    \alt <fun-param>
    
    <fun-param> ::= <identifier>
    \alt "(" <fun-param-list> ")"
    \alt "{" <fun-param-list> "}"
    
    <fun-param-list> ::= <fun-param-elem> <fun-param-list-tail>
    
    <fun-param-elem> ::= <identifier>
    \alt <identifier> ":" <ty>
    
    <fun-param-list-tail> ::= "," <fun-param-elem> <fun-param-list-tail>
    \alt $\epsilon$
    
    <conditional> ::= "if" <exp> "then" <exp> "else" <exp>
    \alt "if" <exp> "then" <exp>
    
    <operation> ::= <arith-op>
    \alt <compare-op>
    \alt <list-op>
    
    <arith-op> ::= <int-op>
    \alt <real-op>
    
    <int-op> ::= "~" <exp>
    \alt <exp> "+" <exp>
    \alt <exp> "-" <exp>
    \alt <exp> "*" <exp>
    \alt <exp> "/" <exp>
    \alt <exp> "
    
    <real-op> ::= <exp> "+." <exp>
    \alt <exp> "-." <exp>
    \alt <exp> "*." <exp>
    \alt <exp> "/." <exp>

    <compare-op> ::= <exp> "=" <sw-t>
    \alt <exp> "<>" <exp>
    \alt <exp> ">" <exp>
    \alt <exp> "<" <exp>
    \alt <exp> ">=" <exp>
    \alt <exp> "<=" <exp>
    
    <list-op> ::= <exp> "::" <exp>
    
    <assign> ::= <exp> ":=" <exp>
    
    <pattern-match> ::= "case" <exp> "of" <pattern> "=>" <exp> <matches-tail>
    
    <pattern> ::= <integer literal>
    \alt <string literal>
    \alt <real literal>
    \alt <identifier>
    \alt <identifier> <pattern>
    \alt "(" <opt-patterns> ")"
    \alt "{" <record-patterns> "}"
    \alt <pattern> "::" <pattern>
    
    <opt-patterns> ::= <pattern> <opt-patterns-tail>
    \alt $\epsilon$
    
    <opt-patterns-tail> ::= "," <pattern> <opt-pattern-tails>
    \alt $\epsilon$
    
    <record-patterns> ::= <identifier "=" <pattern> <record-patterns-tail>
    
    <record-patterns-tail> ::= "," <identifier> "=" <pattern> <rec-patterns-tail>
    \alt $\epsilon$
    
    <matches-tail> ::= <matches-tail> "|:" <pattern> "=>" <exp>
    
    <sequence> ::= "(" <exp> ")"
    \alt "(" <exp> ";" <exp> <sequence-tail> ")"
    
    <sequence-tail> ::= ";" <exp> <sequence-tail>
    \alt $\epsilon$
    
    <application> ::= <exp> <exp>
\end{grammar}

\subsection{Hardware Syntax Grammar}
\label{sec:hwsgrammar}
\begin{grammar}
    <exp> ::= <literal>
    \alt <access>
    \alt <let binding>
    \alt <operation>
    \alt <parameterization>

    <literal> ::= <bit literal>
    \alt <array literal>
    \alt <hardware record literal>
    
    <bit literal> ::= "\'b:" <binary-digit>
    
    <array literal> ::= "#[" <list-body> "]"
    \alt <gen-array>
    \alt <bit-array>
    
    <gen-array> ::= "#[" <exp> ";" "gen" <identifier> "=>" <exp> "]"
    
    <bit-array> ::= <exp> "'s:" <exp>
    \alt  <exp> "'u:" <exp>
    \alt  <exp> "'r:" <exp>
    
    <hardware record literal> ::= "#\{" <record-body> "\}"
    
    <access> ::= <array-access>
    
    <array-access> ::= <exp> "[:" <exp> ":]"
    
    <let binding> ::= "let" <decs> "in" <exp> "end"
    
    <decs> ::= <hdataty-dec> <decs>
    \alt <module-dec> <decs>
    \alt $\epsilon$
    
    <ty> ::= "#{" <ty-fields> "}"
    \alt <ty> "#*" <ty>
    \alt <ty> "[" <decimal-integer> "]"
    \alt <ty> "@" <decimal-integer>
    
    <hdataty-dec> ::= "hdatatype" <tyvars> <identifier> "=" <identifier> "of" <ty> <dataty-tail>
    
    <module-dec> ::= "module" <identifier> <mod-sw-param> <mod-param> "=" <exp>
    \alt "module" <identifier> <mod-sw-param> <mod-param> ":" <ty> "=" <exp>
    
    <mod-sw-param> ::= "<:" <fun-param> ">:"
    \alt $\epsilon$
    
    <mod-param> ::= <identifier>
    \alt "#(" <fun-param-list> ")"
    \alt "#{" <fun-param-list> "}"
    
    <operation> ::= <bit-op>
    \alt <unsw>
    
    <bit-op> ::= "!" <exp>
    \alt "|->" <exp>
    \alt "&->" <exp>
    \alt "^->" <exp>
    \alt <exp> "&&" <exp>
    \alt <exp> "||" <exp>
    \alt <exp> "^^" <exp>
    \alt <exp> "&" <exp>
    \alt <exp> "|" <exp>
    \alt <exp> "^" <exp>
    \alt <exp> "<<" <exp>
    \alt <exp> ">>" <exp>
    \alt <exp> ">>>" <exp>
    
    <unsw> ::= "unsw" <exp>
    
    <parameterization> ::= <exp> "<:" <exp> ">:"
\end{grammar}

\subsection{Typing Rules}
\label{sec:fulltyrules}
\tyrule{%
    \AxiomC{$\mathtt{x : T \in \Gamma}$}%
    \UnaryInfC{$\mathtt{\Gamma \vdash x : T}$}
    \DisplayProof
}{T-VAR}

\tyrule{%
    \AxiomC{$\mathtt{\Gamma, x : T_1 \vdash t_2 : T_2}$}%
    \UnaryInfC{$\mathtt{\Gamma \vdash \lambda x : T_1 . t_2 : T_1 \rightarrow T_2}$}
    \DisplayProof
}{T-ABS}

\tyrule{%
    \AxiomC{$\mathtt{\Gamma \vdash t_1 : T_1 \rightarrow T_2}$}%
    \AxiomC{$\mathtt{\Gamma \vdash t_2 : T_1}$}
    \BinaryInfC{$\mathtt{\Gamma \vdash t_1 t_2 : T_2}$}
    \DisplayProof
}{T-APP}

\tyrule{%
    \AxiomC{$\langle int \rangle$ : \texttt{int}}%
    \DisplayProof
}{T-INT}

\tyrule{%
    \AxiomC{$\langle real \rangle$ : \texttt{real}}%
    \DisplayProof
}{T-REAL}

\tyrule{%
    \AxiomC{$\langle string \rangle$ : \texttt{string}}%
    \DisplayProof
}{T-STRING}

\tyrule{%
    \AxiomC{$\langle bit \rangle$ : \texttt{bit}}%
    \DisplayProof
}{T-BIT}

\tyrule{%
    \AxiomC{\texttt{nil :} $\mathtt{T_S}$ \texttt{list}}%
    \DisplayProof
}{T-NIL}

\tyrule{%
    \AxiomC{$\mathtt{t_1 : int}$}%
    \UnaryInfC{\textasciitilde $\mathtt{t_1 : int}$}
    \DisplayProof
}{T-INT-NEG}

\tyrule{%
    \AxiomC{$\mathtt{t_1 : int}$}%
    \AxiomC{$\mathtt{t_2 : int}$}
    \BinaryInfC{$\mathtt{t_1}$ \texttt{+} $\mathtt{t_2 : int}$}
    \DisplayProof
}{T-INT-ADD}

\tyrule{%
    \AxiomC{$\mathtt{t_1 : int}$}%
    \AxiomC{$\mathtt{t_2 : int}$}
    \BinaryInfC{$\mathtt{t_1}$ \texttt{-} $\mathtt{t_2 : int}$}
    \DisplayProof
}{T-INT-SUB}

\tyrule{%
    \AxiomC{$\mathtt{t_1 : int}$}%
    \AxiomC{$\mathtt{t_2 : int}$}
    \BinaryInfC{$\mathtt{t_1}$ \texttt{*} $\mathtt{t_2 : int}$}
    \DisplayProof
}{T-INT-MUL}

\tyrule{%
    \AxiomC{$\mathtt{t_1 : int}$}%
    \AxiomC{$\mathtt{t_2 : int}$}
    \BinaryInfC{$\mathtt{t_1}$ \texttt{/} $\mathtt{t_2 : int}$}
    \DisplayProof
}{T-INT-DIV}

\tyrule{%
    \AxiomC{$\mathtt{t_1 : int}$}%
    \AxiomC{$\mathtt{t_2 : int}$}
    \BinaryInfC{$\mathtt{t_1}$ \texttt{\%} $\mathtt{t_2 : int}$}
    \DisplayProof
}{T-INT-MOD}

\tyrule{%
    \AxiomC{$\mathtt{t_1 : real}$}%
    \UnaryInfC{\textasciitilde $\mathtt{t_1 : real}$}
    \DisplayProof
}{T-REAL-NEG}

\tyrule{%
    \AxiomC{$\mathtt{t_1 : real}$}%
    \AxiomC{$\mathtt{t_2 : real}$}
    \BinaryInfC{$\mathtt{t_1}$ \texttt{+.} $\mathtt{t_2 : real}$}
    \DisplayProof
}{T-REAL-ADD}

\tyrule{%
    \AxiomC{$\mathtt{t_1 : real}$}%
    \AxiomC{$\mathtt{t_2 : real}$}
    \BinaryInfC{$\mathtt{t_1}$ \texttt{-.} $\mathtt{t_2 : real}$}
    \DisplayProof
}{T-REAL-SUB}

\tyrule{%
    \AxiomC{$\mathtt{t_1 : real}$}%
    \AxiomC{$\mathtt{t_2 : real}$}
    \BinaryInfC{$\mathtt{t_1}$ \texttt{*.} $\mathtt{t_2 : real}$}
    \DisplayProof
}{T-REAL-MUL}

\tyrule{%
    \AxiomC{$\mathtt{t_1 : real}$}%
    \AxiomC{$\mathtt{t_2 : real}$}
    \BinaryInfC{$\mathtt{t_1}$ \texttt{/.} $\mathtt{t_2 : real}$}
    \DisplayProof
}{T-REAL-DIV}

\tyrule{%
    \AxiomC{$\mathtt{t_1 : T_H}$}%
    \UnaryInfC{\texttt{!}$\mathtt{t_1 : T_H}$}
    \DisplayProof
}{T-BIT-NEG}

\tyrule{%
    \AxiomC{$\mathtt{t_1 : T_H}$}
    \AxiomC{$\mathtt{t_2 : T_H}$}
    \BinaryInfC{$\mathtt{t_1}$ \texttt{\&} $\mathtt{t_2 : T_H}$}
    \DisplayProof
}{T-AND}

\tyrule{%
    \AxiomC{$\mathtt{t_1 : T_H}$}
    \AxiomC{$\mathtt{t_2 : T_H}$}
    \BinaryInfC{$\mathtt{t_1}$ \texttt{|} $\mathtt{t_2 : T_H}$}
    \DisplayProof
}{T-OR}

\tyrule{%
    \AxiomC{$\mathtt{t_1 : T_H}$}
    \AxiomC{$\mathtt{t_2 : T_H}$}
    \BinaryInfC{$\mathtt{t_1}$ \texttt{$^\wedge$} $\mathtt{t_2 : T_H}$}
    \DisplayProof
}{T-XOR}

\tyrule{%
    \AxiomC{$\mathtt{t_1 : bit}$\texttt{[n]}}%
    \UnaryInfC{\texttt{\&->}$\mathtt{t_1 : bit}$}
    \DisplayProof
}{T-AND-RED}

\tyrule{%
    \AxiomC{$\mathtt{t_1 : bit}$\texttt{[n]}}%
    \UnaryInfC{\texttt{|->}$\mathtt{t_1 : bit}$}
    \DisplayProof
}{T-OR-RED}

\tyrule{%
    \AxiomC{$\mathtt{t_1 : bit}$\texttt{[n]}}%
    \UnaryInfC{\texttt{$^\wedge$->}$\mathtt{t_1 : bit}$}
    \DisplayProof
}{T-XOR-RED}

\tyrule{%
    \AxiomC{$\mathtt{t_1 : bit}$\texttt{[n]}}%
    \AxiomC{$\mathtt{t_2 : bit}$\texttt{[n]}}
    \BinaryInfC{$\mathtt{t_1}$ \texttt{\&}\texttt{\&} $\mathtt{t_2 : bit}$}
    \DisplayProof
}{T-LOG-AND}

\tyrule{%
    \AxiomC{$\mathtt{t_1 : bit}$\texttt{[n]}}%
    \AxiomC{$\mathtt{t_2 : bit}$\texttt{[n]}}
    \BinaryInfC{$\mathtt{t_1}$ \texttt{|}\texttt{|} $\mathtt{t_2 : bit}$}
    \DisplayProof
}{T-LOG-OR}

\tyrule{%
    \AxiomC{$\mathtt{t_1 : bit}$\texttt{[n]}}%
    \AxiomC{$\mathtt{t_2 : bit}$\texttt{[n]}}
    \BinaryInfC{$\mathtt{t_1}$ \texttt{$^\wedge$}\texttt{$^\wedge$} $\mathtt{t_2 : bit}$}
    \DisplayProof
}{T-LOG-XOR}

\tyrule{%
    \AxiomC{$\mathtt{t_1 : bit}$\texttt{[n]}}%
    \AxiomC{$\mathtt{t_2 : bit}$\texttt{[m]}}
    \BinaryInfC{$\mathtt{t_1}$ \texttt{<}\texttt{<} $\mathtt{t_2 : bit}$\texttt{[n]}}
    \DisplayProof
}{T-SLL}

\tyrule{%
    \AxiomC{$\mathtt{t_1 : bit}$\texttt{[n]}}%
    \AxiomC{$\mathtt{t_2 : bit}$\texttt{[m]}}
    \BinaryInfC{$\mathtt{t_1}$ \texttt{>}\texttt{>} $\mathtt{t_2 : bit}$\texttt{[n]}}
    \DisplayProof
}{T-SRL}

\tyrule{%
    \AxiomC{$\mathtt{t_1 : bit}$\texttt{[n]}}%
    \AxiomC{$\mathtt{t_2 : bit}$\texttt{[m]}}
    \BinaryInfC{$\mathtt{t_1}$ \texttt{>}\texttt{>}\texttt{>} $\mathtt{t_2 : bit}$\texttt{[n]}}
    \DisplayProof
}{T-SRA}

\tyrule{%
    \AxiomC{$\mathtt{t_1 : T_S}$}%
    \AxiomC{$\mathtt{t_2 : T_S}$}
    \BinaryInfC{$\mathtt{t_1}$ \texttt{=} $\mathtt{t_2 : int}$}
    \DisplayProof
}{T-EQ}

\tyrule{%
    \AxiomC{$\mathtt{t_1 : T_S}$}%
    \AxiomC{$\mathtt{t_2 : T_S}$}
    \BinaryInfC{$\mathtt{t_1}$ \texttt{<>} $\mathtt{t_2 : int}$}
    \DisplayProof
}{T-NEQ}

\tyrule{%
    \AxiomC{$\mathtt{t_1 : T_S}$}%
    \AxiomC{$\mathtt{t_2 : T_S}$}
    \BinaryInfC{$\mathtt{t_1}$ \texttt{>=} $\mathtt{t_2 : int}$}
    \DisplayProof
}{T-GEQ}

\tyrule{%
    \AxiomC{$\mathtt{t_1 : T_S}$}%
    \AxiomC{$\mathtt{t_2 : T_S}$}
    \BinaryInfC{$\mathtt{t_1}$ \texttt{>} $\mathtt{t_2 : int}$}
    \DisplayProof
}{T-GT}

\tyrule{%
    \AxiomC{$\mathtt{t_1 : T_S}$}%
    \AxiomC{$\mathtt{t_2 : T_S}$}
    \BinaryInfC{$\mathtt{t_1}$ \texttt{<=} $\mathtt{t_2 : int}$}
    \DisplayProof
}{T-LEQ}

\tyrule{%
    \AxiomC{$\mathtt{t_1 : T_S}$}%
    \AxiomC{$\mathtt{t_2 : T_S}$}
    \BinaryInfC{$\mathtt{t_1}$ \texttt{<} $\mathtt{t_2 : int}$}
    \DisplayProof
}{T-LT}

\tyrule{%
    \AxiomC{$\mathtt{t_1 : T_S}$}%
    \AxiomC{$\mathtt{t_2 : T_S}$ \texttt{list}}
    \BinaryInfC{$\mathtt{t_1 :: t_2 : T_S}$ \texttt{list}}
    \DisplayProof
}{T-CONS}

\tyrule{%
    \AxiomC{$\mathtt{t_1 : int}$}%
    \AxiomC{$\mathtt{t_2 : T}$}
    \AxiomC{$\mathtt{t_1 : T}$}
    \TrinaryInfC{\texttt{if} $\mathtt{t_1}$ \texttt{then} $\mathtt{t_2}$ \texttt{else} $\mathtt{t_3 : T}$}
    \DisplayProof
}{T-IFELSE}

\tyrule{%
    \AxiomC{$\mathtt{t_1 : T_S}$}%
    \UnaryInfC{\texttt{ref} $\mathtt{t_1 : T_S}$ \texttt{ref}}
    \DisplayProof
}{T-REF}

\tyrule{%
    \AxiomC{$\mathtt{t_1 : T_H}$}%
    \UnaryInfC{\texttt{sw} $\mathtt{t_1 : T_H}$ \texttt{sw}}
    \DisplayProof
}{T-SW}

\tyrule{%
    \AxiomC{\texttt{$\mathtt{t_1 : T_H}$ sw}}%
    \UnaryInfC{\texttt{unsw} $\mathtt{t_1 : T_H}$}
    \DisplayProof
}{T-UNSW}

\tyrule{%
    \AxiomC{$\mathtt{t_1 : T_S}$ \texttt{ref}}%
    \AxiomC{$\mathtt{t_2 : T_S}$}
    \BinaryInfC{$\mathtt{t_1}$ \texttt{:=} $\mathtt{t_2 : unit}$}
    \DisplayProof
}{T-ASSIGN}

\tyrule{%
    \AxiomC{$\mathtt{t_1 : T_H}$\texttt{[n]}}%
    \AxiomC{$\mathtt{t_2 : int}$}
    \BinaryInfC{$\mathtt{t_1}$\texttt{[}$\mathtt{t_2}$\texttt{]} $\mathtt{: T_H}$}
    \DisplayProof
}{T-ARR-ACC}

\tyrule{%
    \AxiomC{$\mathtt{t_1 : T_S}$ \texttt{ref}}%
    \UnaryInfC{\texttt{\$}$\mathtt{t_1 : T_S}$}
    \DisplayProof
}{T-DEREF}

\tyrule{%
    \AxiomC{for each $i$ $\mathtt{ \Gamma \vdash t_i : T_i}$}%
    \UnaryInfC{$\mathtt{\Gamma \vdash \{l_i = t_{i}^{i \in 1..n}\} : \{l_i : T_{i}^{i \in 1..n}\}}$}
    \DisplayProof
}{T-RCD}

\tyrule{%
    \AxiomC{$\mathtt{\Gamma \vdash t_1 : \{l_i : T_{i}^{i \in 1..n}\}}$}%
    \UnaryInfC{$\mathtt{\Gamma \vdash}$ \texttt{\#}$\mathtt{l_j}$ $\mathtt{t_1 : T_{j}}$}
    \DisplayProof
}{T-PROJ}

\tyrule{%
    \AxiomC{$\mathtt{\Gamma, \tau \vdash t_1 : T_1}$}%
    \noLine
    \UnaryInfC{$\mathtt{\Gamma, \tau, x_1 : T_1 \vdash}$\texttt{let} $\mathtt{(x_i}$ \texttt{=} $\mathtt{t_i)^{i \in 2..n}}$ \texttt{in} $\mathtt{t_0}$ \texttt{end}$\mathtt{: T_0}$}
    \UnaryInfC{\texttt{let} $\mathtt{(x_i}$ \texttt{=} $\mathtt{t_i)^{i \in 1..n}}$ \texttt{in} $\mathtt{t_0}$ \texttt{end} $\mathtt{: T_0}$}
    \DisplayProof
}{T-LET}

\tyrule{%
    \AxiomC{$\mathtt{D : \langle C_i:T_i \rangle^{i \in 1..n} \in \tau}$}
    \AxiomC{$\mathtt{\Gamma \vdash t : T_j}$}
    \BinaryInfC{$\mathtt{\Gamma, \tau \vdash C_j}$ \texttt{t} $\mathtt{: D}$}
    \DisplayProof
}{T-DATATY}

\tyrule{%
    \AxiomC{$\mathtt{\Gamma \vdash t_0 : \langle C_i : T_i \rangle ^{i \in 1..n}}$}
    \noLine
    \UnaryInfC{for each $i$ $\mathtt{ \Gamma, x_i : T_i \vdash t_i : T}$}
    \UnaryInfC{$\mathtt{\Gamma \vdash}$ \texttt{case} $\mathtt{t_0}$ \texttt{of} $\mathtt{C_i x_i}$ \texttt{=>} $\mathtt{t_i ^{\hspace{1mm} i \in 1..n} : T}$}
    \DisplayProof
}{T-CASE}

\subsection{Evaluation Rules}
\label{sec:fullevrules}
\tyrule{%
    \AxiomC{$\mathtt{t_1} \longrightarrow \mathtt{t_1'}$}
    \UnaryInfC{$\mathtt{t_1 t_2} \longrightarrow \mathtt{t_1' t_2}$}
    \DisplayProof
}{E-APP1}

\tyrule{%
    \AxiomC{$\mathtt{t_2} \longrightarrow \mathtt{t_2'}$}
    \UnaryInfC{$\mathtt{v_1 t_2} \longrightarrow \mathtt{v_1 t_2'}$}
    \DisplayProof
}{E-APP2}

\tyrule{%
    \AxiomC{$\mathtt{(\lambda x.t_1)v_1} \longrightarrow \mathtt{[x \mapsto v_1]t_1}$}
    \DisplayProof
}{E-APPABS}

\tyrule{%
    \AxiomC{$\mathtt{t_1 \longrightarrow t_1'}$}
    \UnaryInfC{\texttt{if} $\mathtt{t_1}$ \texttt{then} $\mathtt{t_2}$ \texttt{else} $\mathtt{t_3}$}
    \noLine
    \UnaryInfC{$\longrightarrow$ \texttt{if} $\mathtt{t_1'}$ \texttt{then} $\mathtt{t_2}$ \texttt{else} $\mathtt{t_3}$}
    \DisplayProof
}{E-IFELSE}

\tyrule{%
    \AxiomC{$\mathtt{v_1 : int}$}
    \AxiomC{$\mathtt{v_1 \neq 0}$}
    \BinaryInfC{\texttt{if} $\mathtt{v_1}$ \texttt{then} $\mathtt{t_2}$ \texttt{else} $\mathtt{t_3} \longrightarrow \mathtt{t_2}$}
    \DisplayProof
}{E-IFELSE-T}

\tyrule{%
    \AxiomC{\texttt{if 0 then} $\mathtt{t_2}$ \texttt{else} $\mathtt{t_3} \longrightarrow \mathtt{t_3}$}
    \DisplayProof
}{E-IFELSE-F}

\tyrule{%
    \AxiomC{$\mathtt{t_1 \longrightarrow t_1'}$}
    \UnaryInfC{\textasciitilde $\mathtt{t_1} \longrightarrow$ \textasciitilde $\mathtt{t_1'}$}
    \DisplayProof
}{E-NEG}

\tyrule{%
    \AxiomC{$\mathtt{t_1 \longrightarrow t_1'}$}
    \UnaryInfC{$\mathtt{t_1}$ \texttt{+} $\mathtt{t_2} \longrightarrow \mathtt{t_1'}$ \texttt{+} $\mathtt{t_2}$}
    \DisplayProof
}{E-INT-ADD1}

\tyrule{%
    \AxiomC{$\mathtt{t_2 \longrightarrow t_2'}$}
    \UnaryInfC{$\mathtt{v_1}$ \texttt{+} $\mathtt{t_2} \longrightarrow \mathtt{v_1}$ \texttt{+} $\mathtt{t_2'}$}
    \DisplayProof
}{E-INT-ADD2}

\tyrule{%
    \AxiomC{$\mathtt{t_1 \longrightarrow t_1'}$}
    \UnaryInfC{$\mathtt{t_1}$ \texttt{-} $\mathtt{t_2} \longrightarrow \mathtt{t_1'}$ \texttt{-} $\mathtt{t_2}$}
    \DisplayProof
}{E-INT-SUB1}

\tyrule{%
    \AxiomC{$\mathtt{t_2 \longrightarrow t_2'}$}
    \UnaryInfC{$\mathtt{v_1}$ \texttt{-} $\mathtt{t_2} \longrightarrow \mathtt{v_1}$ \texttt{-} $\mathtt{t_2'}$}
    \DisplayProof
}{E-INT-SUB2}

\tyrule{%
    \AxiomC{$\mathtt{t_1 \longrightarrow t_1'}$}
    \UnaryInfC{$\mathtt{t_1}$ \texttt{*} $\mathtt{t_2} \longrightarrow \mathtt{t_1'}$ \texttt{*} $\mathtt{t_2}$}
    \DisplayProof
}{E-INT-MUL1}

\tyrule{%
    \AxiomC{$\mathtt{t_2 \longrightarrow t_2'}$}
    \UnaryInfC{$\mathtt{v_1}$ \texttt{*} $\mathtt{t_2} \longrightarrow \mathtt{v_1}$ \texttt{*} $\mathtt{t_2'}$}
    \DisplayProof
}{E-INT-MUL2}

\tyrule{%
    \AxiomC{$\mathtt{t_1 \longrightarrow t_1'}$}
    \UnaryInfC{$\mathtt{t_1}$ \texttt{/} $\mathtt{t_2} \longrightarrow \mathtt{t_1'}$ \texttt{/} $\mathtt{t_2}$}
    \DisplayProof
}{E-INT-DIV1}

\tyrule{%
    \AxiomC{$\mathtt{t_2 \longrightarrow t_2'}$}
    \UnaryInfC{$\mathtt{v_1}$ \texttt{/} $\mathtt{t_2} \longrightarrow \mathtt{v_1}$ \texttt{/} $\mathtt{t_2'}$}
    \DisplayProof
}{E-INT-DIV2}

\tyrule{%
    \AxiomC{$\mathtt{t_1 \longrightarrow t_1'}$}
    \UnaryInfC{$\mathtt{t_1}$ \texttt{\%} $\mathtt{t_2} \longrightarrow \mathtt{t_1'}$ \texttt{\%} $\mathtt{t_2}$}
    \DisplayProof
}{E-INT-MOD1}

\tyrule{%
    \AxiomC{$\mathtt{t_2 \longrightarrow t_2'}$}
    \UnaryInfC{$\mathtt{v_1}$ \texttt{\%} $\mathtt{t_2} \longrightarrow \mathtt{v_1}$ \texttt{\%} $\mathtt{t_2'}$}
    \DisplayProof
}{E-INT-MOD2}

\tyrule{%
    \AxiomC{$\mathtt{t_1 \longrightarrow t_1'}$}
    \UnaryInfC{$\mathtt{t_1}$ \texttt{+.} $\mathtt{t_2} \longrightarrow \mathtt{t_1'}$ \texttt{+.} $\mathtt{t_2}$}
    \DisplayProof
}{E-REAL-ADD1}

\tyrule{%
    \AxiomC{$\mathtt{t_2 \longrightarrow t_2'}$}
    \UnaryInfC{$\mathtt{v_1}$ \texttt{+.} $\mathtt{t_2} \longrightarrow \mathtt{v_1}$ \texttt{+.} $\mathtt{t_2'}$}
    \DisplayProof
}{E-REAL-ADD2}

\tyrule{%
    \AxiomC{$\mathtt{t_1 \longrightarrow t_1'}$}
    \UnaryInfC{$\mathtt{t_1}$ \texttt{-.} $\mathtt{t_2} \longrightarrow \mathtt{t_1'}$ \texttt{-.} $\mathtt{t_2}$}
    \DisplayProof
}{E-REAL-SUB1}

\tyrule{%
    \AxiomC{$\mathtt{t_2 \longrightarrow t_2'}$}
    \UnaryInfC{$\mathtt{v_1}$ \texttt{-.} $\mathtt{t_2} \longrightarrow \mathtt{v_1}$ \texttt{-.} $\mathtt{t_2'}$}
    \DisplayProof
}{E-REAL-SUB2}

\tyrule{%
    \AxiomC{$\mathtt{t_1 \longrightarrow t_1'}$}
    \UnaryInfC{$\mathtt{t_1}$ \texttt{*.} $\mathtt{t_2} \longrightarrow \mathtt{t_1'}$ \texttt{*.} $\mathtt{t_2}$}
    \DisplayProof
}{E-REAL-MUL1}

\tyrule{%
    \AxiomC{$\mathtt{t_2 \longrightarrow t_2'}$}
    \UnaryInfC{$\mathtt{v_1}$ \texttt{*.} $\mathtt{t_2} \longrightarrow \mathtt{v_1}$ \texttt{*.} $\mathtt{t_2'}$}
    \DisplayProof
}{E-REAL-MUL2}

\tyrule{%
    \AxiomC{$\mathtt{t_1 \longrightarrow t_1'}$}
    \UnaryInfC{$\mathtt{t_1}$ \texttt{/.} $\mathtt{t_2} \longrightarrow \mathtt{t_1'}$ \texttt{/.} $\mathtt{t_2}$}
    \DisplayProof
}{E-REAL-DIV1}

\tyrule{%
    \AxiomC{$\mathtt{t_2 \longrightarrow t_2'}$}
    \UnaryInfC{$\mathtt{v_1}$ \texttt{/.} $\mathtt{t_2} \longrightarrow \mathtt{v_1}$ \texttt{/.} $\mathtt{t_2'}$}
    \DisplayProof
}{E-REAL-DIV2}

\tyrule{%
    \AxiomC{$\mathtt{t_j} \longrightarrow \mathtt{t_j'}$}
    \UnaryInfC{\texttt{\#[$\mathtt{v_i}^{i \in 1..j-1}$, $\mathtt{t_j}$, $\mathtt{t_k}^{k \in j+1..n}$]}}
    \noLine
    \UnaryInfC{\texttt{$\longrightarrow$ \#[$\mathtt{v_i}^{i \in 1..j-1}$, $\mathtt{t_j'}$, $\mathtt{t_k}^{k \in j+1..n}$]}}
    \DisplayProof
}{E-ARR}

\tyrule{%
    \AxiomC{$\mathtt{t_1} \longrightarrow \mathtt{t_1'}$}
    \UnaryInfC{\texttt{!$\mathtt{t_1} \longrightarrow$ !$\mathtt{t_1'}$}}
    \DisplayProof
}{E-BIT-NEG1}

\tyrule{%
    \AxiomC{\texttt{$\mathtt{v : T_H}$[n]}}
    \UnaryInfC{\texttt{!$\mathtt{v} \longrightarrow$ \#[!$\mathtt{v}$[i]]$^{i \in 0..n-1}$}}
    \DisplayProof
}{E-BIT-NEG2}

\tyrule{%
    \AxiomC{\texttt{$\mathtt{v : }$ \#\{$\mathtt{l_i : T_{Hi}}^{i \in 1..n}$\}}}
    \UnaryInfC{\texttt{!$\mathtt{v} \longrightarrow$ \#\{$\mathtt{l_i = }$!\#$\mathtt{l_i}$ $\mathtt{v}$\}$^{i \in 1..n}$}}
    \DisplayProof
}{E-BIT-NEG3}

\tyrule{%
    \AxiomC{$\mathtt{v : bit}$}
    \UnaryInfC{\texttt{!$\mathtt{v} \longrightarrow$
    \begin{tikzpicture}
        \node (x) at (0, 0) {$\mathtt{v}$};
        \node[not gate US, draw] at ($(x) + (0.8, 0)$) (notx) {};
        \draw (x) -- (notx.input);
        \draw (notx.output) -- node[above]{$\mathtt{\bar v_1}$} ($(notx) + (1.5, 0)$);
    \end{tikzpicture}}}
    \DisplayProof
}{E-BIT-NEG}

\tyrule{%
    \AxiomC{$\mathtt{t_1} \longrightarrow \mathtt{t_1'}$}
    \UnaryInfC{\texttt{$\mathtt{t_1}$ \& $\mathtt{t_2} \longrightarrow \mathtt{t_1'}$ \& $\mathtt{t_2}$}}
    \DisplayProof
}{E-AND1}

\tyrule{%
    \AxiomC{$\mathtt{t_2} \longrightarrow \mathtt{t_2'}$}
    \UnaryInfC{\texttt{$\mathtt{v_1}$ \& $\mathtt{t_2} \longrightarrow \mathtt{v_1}$ \& $\mathtt{t_2'}$}}
    \DisplayProof
}{E-AND2}

\tyrule{%
    \AxiomC{\texttt{$\mathtt{v_1 : bit}$[n]}}
    \AxiomC{\texttt{$\mathtt{v_2 : bit}$[n]}}
    \BinaryInfC{\texttt{$\mathtt{v_1}$ \& $\mathtt{v_2} \longrightarrow$ \#[$\mathtt{v_{1, i}}$ \& $\mathtt{v_{2, i}}$]$^{i \in 0..n-1}$}}
    \DisplayProof
}{E-AND3}

\tyrule{%
    \AxiomC{\texttt{$\mathtt{v_1 : }$ \#\{$\mathtt{l_i : T_{Hi}}^{i \in 1..n}$\}}}
    \noLine
    \UnaryInfC{\texttt{$\mathtt{v_2 : }$ \#\{$\mathtt{l_i : T_{Hi}}^{i \in 1..n}$\}}}
    \UnaryInfC{\texttt{$\mathtt{v_1}$ \& $\mathtt{v_2} \longrightarrow$ \#\{$\mathtt{l_i}$ = \#$\mathtt{l_i}$ $\mathtt{v_1}$ \& \#$\mathtt{l_i}$ $\mathtt{v_2}$\}$^{i \in 1..n}$}}
    \DisplayProof
}{E-AND4}

\tyrule{%
    \AxiomC{$\mathtt{t_1} \longrightarrow \mathtt{t_1'}$}
    \UnaryInfC{\texttt{$\mathtt{t_1}$ | $\mathtt{t_2} \longrightarrow \mathtt{t_1'}$ | $\mathtt{t_2}$}}
    \DisplayProof
}{E-OR1}

\tyrule{%
    \AxiomC{$\mathtt{t_2} \longrightarrow \mathtt{t_2'}$}
    \UnaryInfC{\texttt{$\mathtt{v_1}$ | $\mathtt{t_2} \longrightarrow \mathtt{v_1}$ | $\mathtt{t_2'}$}}
    \DisplayProof
}{E-OR2}

\tyrule{%
    \AxiomC{\texttt{$\mathtt{v_1 : bit}$[n]}}
    \AxiomC{\texttt{$\mathtt{v_2 : bit}$[n]}}
    \BinaryInfC{\texttt{$\mathtt{v_1}$ | $\mathtt{v_2} \longrightarrow$ \#[$\mathtt{v_{1, i}}$ | $\mathtt{v_{2, i}}$]$^{i \in 0..n-1}$}}
    \DisplayProof
}{E-OR3}

\tyrule{%
    \AxiomC{\texttt{$\mathtt{v_1 : }$ \#\{$\mathtt{l_i : T_{Hi}}^{i \in 1..n}$\}}}
    \noLine
    \UnaryInfC{\texttt{$\mathtt{v_2 : }$ \#\{$\mathtt{l_i : T_{Hi}}^{i \in 1..n}$\}}}
    \UnaryInfC{\texttt{$\mathtt{v_1}$ | $\mathtt{v_2} \longrightarrow$ \#\{$\mathtt{l_i}$ = \#$\mathtt{l_i}$ $\mathtt{v_1}$ | \#$\mathtt{l_i}$ $\mathtt{v_2}$\}$^{i \in 1..n}$}}
    \DisplayProof
}{E-OR4}

\tyrule{%
    \AxiomC{$\mathtt{t_1} \longrightarrow \mathtt{t_1'}$}
    \UnaryInfC{\texttt{$\mathtt{t_1}$ $^\wedge$ $\mathtt{t_2} \longrightarrow \mathtt{t_1'}$ $^\wedge$ $\mathtt{t_2}$}}
    \DisplayProof
}{E-XOR1}

\tyrule{%
    \AxiomC{$\mathtt{t_2} \longrightarrow \mathtt{t_2'}$}
    \UnaryInfC{\texttt{$\mathtt{v_1}$ $^\wedge$ $\mathtt{t_2} \longrightarrow \mathtt{v_1}$ $^\wedge$ $\mathtt{t_2'}$}}
    \DisplayProof
}{E-XOR2}

\tyrule{%
    \AxiomC{\texttt{$\mathtt{v_1 : bit}$[n]}}
    \AxiomC{\texttt{$\mathtt{v_2 : bit}$[n]}}
    \BinaryInfC{\texttt{$\mathtt{v_1}$ $^\wedge$ $\mathtt{v_2} \longrightarrow$ \#[$\mathtt{v_{1, i}}$ $^\wedge$ $\mathtt{v_{2, i}}$]$^{i \in 0..n-1}$}}
    \DisplayProof
}{E-XOR3}

\tyrule{%
    \AxiomC{\texttt{$\mathtt{v_1 : }$ \#\{$\mathtt{l_i : T_{Hi}}^{i \in 1..n}$\}}}
    \noLine
    \UnaryInfC{\texttt{$\mathtt{v_2 : }$ \#\{$\mathtt{l_i : T_{Hi}}^{i \in 1..n}$\}}}
    \UnaryInfC{\texttt{$\mathtt{v_1}$ $^\wedge$ $\mathtt{v_2} \longrightarrow$ \#\{$\mathtt{l_i}$ = \#$\mathtt{l_i}$ $\mathtt{v_1}$ $^\wedge$ \#$\mathtt{l_i}$ $\mathtt{v_2}$\}$^{i \in 1..n}$}}
    \DisplayProof
}{E-XOR4}

\tyrule{%
    \AxiomC{$\mathtt{t_1} \longrightarrow \mathtt{t_1'}$}
    \UnaryInfC{\texttt{\&->$\mathtt{t_1} \longrightarrow$ \&->$\mathtt{t_1'}$}}
    \DisplayProof
}{E-AND-RED1}

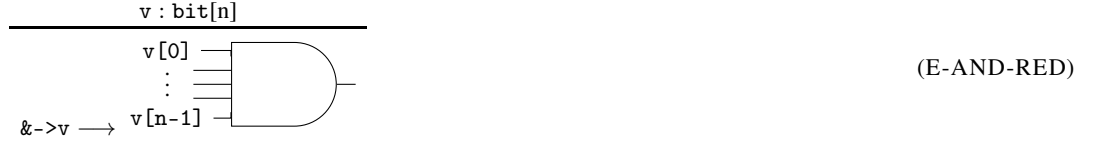
\tyrule{%
    \AxiomC{$\mathtt{v : bit}$[n]}
    \UnaryInfC{\texttt{\&->}$\mathtt{v} \longrightarrow$
    \begin{tikzpicture}
        \node (v0) at (0, 0.9) {\texttt{$\mathtt{v}$[0]}};
        \node (vdots) at (0, 0.55) {$\vdots$};
        \node (vn) at (0, 0) {\texttt{$\mathtt{v}$[n-1]}};
        \node[and gate US, draw, logic gate inputs=nnnnn, scale=1.1] at ($(vdots) + (1.5, -0.1)$) (xand) {};
        \draw (v0) -| (xand.input 1);
        \draw (vn) -| (xand.input 5);
        \draw (xand.output) -- ($(xand) + (1, 0)$);
        \draw ($(xand.input 2) + (-0.5, 0)$) -- (xand.input 2); \draw ($(xand.input 3) + (-0.5, 0)$) -- (xand.input 3);
        \draw ($(xand.input 4) + (-0.5, 0)$) -- (xand.input 4);
    \end{tikzpicture}}
    \DisplayProof
}{E-AND-RED}

\tyrule{%
    \AxiomC{$\mathtt{t_1} \longrightarrow \mathtt{t_1'}$}
    \UnaryInfC{\texttt{|->$\mathtt{t_1} \longrightarrow$ |->$\mathtt{t_1'}$}}
    \DisplayProof
}{E-OR-RED1}

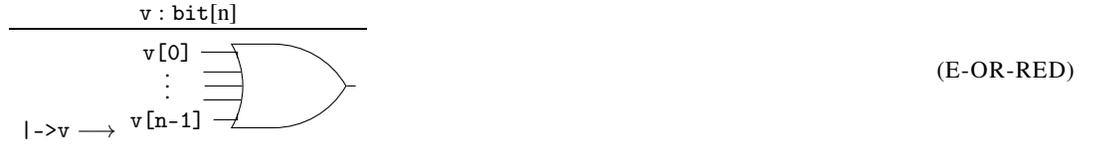
\tyrule{%
    \AxiomC{$\mathtt{v : bit}$[n]}
    \UnaryInfC{\texttt{|->}$\mathtt{v} \longrightarrow$
    \begin{tikzpicture}
        \node (v0) at (0, 0.9) {\texttt{$\mathtt{v}$[0]}};
        \node (vdots) at (0, 0.55) {$\vdots$};
        \node (vn) at (0, 0) {\texttt{$\mathtt{v}$[n-1]}};
        \node[or gate US, draw, logic gate inputs=nnnnn, scale=1.1] at ($(vdots) + (1.5, -0.1)$) (xand) {};
        \draw (v0) -| (xand.input 1);
        \draw (vn) -| (xand.input 5);
        \draw (xand.output) -- ($(xand) + (1, 0)$);
        \draw ($(xand.input 2) + (-0.5, 0)$) -- (xand.input 2); \draw ($(xand.input 3) + (-0.5, 0)$) -- (xand.input 3);
        \draw ($(xand.input 4) + (-0.5, 0)$) -- (xand.input 4);
    \end{tikzpicture}}
    \DisplayProof
}{E-OR-RED}

\tyrule{%
    \AxiomC{$\mathtt{t_1} \longrightarrow \mathtt{t_1'}$}
    \UnaryInfC{\texttt{$^\wedge$->$\mathtt{t_1} \longrightarrow$ $^\wedge$->$\mathtt{t_1'}$}}
    \DisplayProof
}{E-XOR-RED1}

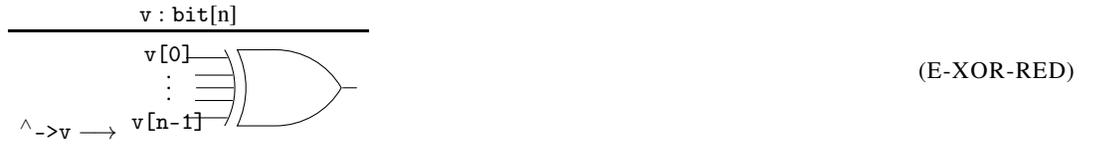
\tyrule{%
    \AxiomC{$\mathtt{v : bit}$[n]}
    \UnaryInfC{\texttt{$^\wedge$->}$\mathtt{v} \longrightarrow$
    \begin{tikzpicture}
        \node (v0) at (0, 0.9) {\texttt{$\mathtt{v}$[0]}};
        \node (vdots) at (0, 0.55) {$\vdots$};
        \node (vn) at (0, 0) {\texttt{$\mathtt{v}$[n-1]}};
        \node[xor gate US, draw, scale=2] at ($(vdots) + (1.5, -0.1)$) (xand) {};
        \draw ($(xand.input 1) + (-0.5, 0)$) -- (xand.input 1);
        \draw ($(xand.input 2) + (-0.5, 0)$) -- (xand.input 2);
        \draw ($(xand.west) + (-0.5, 0)$) -- (xand.west);
        \draw ($(xand.north west) + (-0.5, -0.1)$) -- ($(xand.north west) + (0.05, -0.1)$);
        \draw ($(xand.south west) + (-0.5, 0.1)$) -- ($(xand.south west) + (0.05, 0.1)$);
        \draw (xand.output) -- ($(xand) + (1, 0)$);
    \end{tikzpicture}}
    \DisplayProof
}{E-XOR-RED}

\tyrule{%
    \AxiomC{$\mathtt{t_1 \longrightarrow t_1'}$}
    \UnaryInfC{$\mathtt{t_1}$ \texttt{<}\texttt{<} $\mathtt{t_2} \longrightarrow \mathtt{t_1'}$ \texttt{<}\texttt{<} $\mathtt{t_2}$}
    \DisplayProof
}{E-SLL1}

\tyrule{%
    \AxiomC{$\mathtt{t_2 \longrightarrow t_2'}$}
    \UnaryInfC{$\mathtt{v_1}$ \texttt{<}\texttt{<} $\mathtt{t_2} \longrightarrow \mathtt{v_1}$ \texttt{<}\texttt{<} $\mathtt{t_2'}$}
    \DisplayProof
}{E-SLL2}

\tyrule{%
    \AxiomC{$\mathtt{t_1 \longrightarrow t_1'}$}
    \UnaryInfC{$\mathtt{t_1}$ \texttt{>}\texttt{>} $\mathtt{t_2} \longrightarrow \mathtt{t_1'}$ \texttt{>}\texttt{>} $\mathtt{t_2}$}
    \DisplayProof
}{E-SRL1}

\tyrule{%
    \AxiomC{$\mathtt{t_2 \longrightarrow t_2'}$}
    \UnaryInfC{$\mathtt{v_1}$ \texttt{>}\texttt{>} $\mathtt{t_2} \longrightarrow \mathtt{v_1}$ \texttt{>}\texttt{>} $\mathtt{t_2'}$}
    \DisplayProof
}{E-SRL2}

\tyrule{%
    \AxiomC{$\mathtt{t_1 \longrightarrow t_1'}$}
    \UnaryInfC{$\mathtt{t_1}$ \texttt{>}\texttt{>}\texttt{>} $\mathtt{t_2} \longrightarrow \mathtt{t_1'}$ \texttt{>}\texttt{>}\texttt{>} $\mathtt{t_2}$}
    \DisplayProof
}{E-SRA1}

\tyrule{%
    \AxiomC{$\mathtt{t_2 \longrightarrow t_2'}$}
    \UnaryInfC{$\mathtt{v_1}$ \texttt{>}\texttt{>}\texttt{>} $\mathtt{t_2} \longrightarrow \mathtt{v_1}$ \texttt{>}\texttt{>}\texttt{>} $\mathtt{t_2'}$}
    \DisplayProof
}{E-SRA2}

\tyrule{%
    \AxiomC{$\mathtt{t_1 \longrightarrow t_1'}$}
    \UnaryInfC{$\mathtt{t_1}$ \texttt{=} $\mathtt{t_2} \longrightarrow \mathtt{t_1'}$ \texttt{=} $\mathtt{t_2}$}
    \DisplayProof
}{E-EQ1}

\tyrule{%
    \AxiomC{$\mathtt{t_2 \longrightarrow t_2'}$}
    \UnaryInfC{$\mathtt{v_1}$ \texttt{=} $\mathtt{t_2} \longrightarrow \mathtt{v_1}$ \texttt{=} $\mathtt{t_2'}$}
    \DisplayProof
}{E-EQ2}

\tyrule{%
    \AxiomC{$\mathtt{v_1}$ \texttt{=} $\mathtt{v_2} \longrightarrow$ subject to semantics of type}
    \DisplayProof
}{E-EQ}

\tyrule{%
    \AxiomC{$\mathtt{t_1 \longrightarrow t_1'}$}
    \UnaryInfC{$\mathtt{t_1}$ \texttt{<>} $\mathtt{t_2} \longrightarrow \mathtt{t_1'}$ \texttt{<>} $\mathtt{t_2}$}
    \DisplayProof
}{E-NEQ1}

\tyrule{%
    \AxiomC{$\mathtt{t_2 \longrightarrow t_2'}$}
    \UnaryInfC{$\mathtt{v_1}$ \texttt{<>} $\mathtt{t_2} \longrightarrow \mathtt{v_1}$ \texttt{<>} $\mathtt{t_2'}$}
    \DisplayProof
}{E-NEQ2}

\tyrule{%
    \AxiomC{$\mathtt{v_1}$ \texttt{<>} $\mathtt{v_2} \longrightarrow$ subject to semantics of type}
    \DisplayProof
}{E-NEQ}

\tyrule{%
    \AxiomC{$\mathtt{t_1 \longrightarrow t_1'}$}
    \UnaryInfC{$\mathtt{t_1}$ \texttt{<} $\mathtt{t_2} \longrightarrow \mathtt{t_1'}$ \texttt{<} $\mathtt{t_2}$}
    \DisplayProof
}{E-LT1}

\tyrule{%
    \AxiomC{$\mathtt{t_2 \longrightarrow t_2'}$}
    \UnaryInfC{$\mathtt{v_1}$ \texttt{<} $\mathtt{t_2} \longrightarrow \mathtt{v_1}$ \texttt{<} $\mathtt{t_2'}$}
    \DisplayProof
}{E-LT2}

\tyrule{%
    \AxiomC{$\mathtt{v_1}$ \texttt{<} $\mathtt{v_2} \longrightarrow$ subject to semantics of type}
    \DisplayProof
}{E-LT}

\tyrule{%
    \AxiomC{$\mathtt{t_1 \longrightarrow t_1'}$}
    \UnaryInfC{$\mathtt{t_1}$ \texttt{<=} $\mathtt{t_2} \longrightarrow \mathtt{t_1'}$ \texttt{<=} $\mathtt{t_2}$}
    \DisplayProof
}{E-LEQ1}

\tyrule{%
    \AxiomC{$\mathtt{t_2 \longrightarrow t_2'}$}
    \UnaryInfC{$\mathtt{v_1}$ \texttt{<=} $\mathtt{t_2} \longrightarrow \mathtt{v_1}$ \texttt{<=} $\mathtt{t_2'}$}
    \DisplayProof
}{E-LEQ2}

\tyrule{%
    \AxiomC{$\mathtt{v_1}$ \texttt{<=} $\mathtt{v_2} \longrightarrow$ subject to semantics of type}
    \DisplayProof
}{E-LEQ}

\tyrule{%
    \AxiomC{$\mathtt{t_1 \longrightarrow t_1'}$}
    \UnaryInfC{$\mathtt{t_1}$ \texttt{>} $\mathtt{t_2} \longrightarrow \mathtt{t_1'}$ \texttt{>} $\mathtt{t_2}$}
    \DisplayProof
}{E-GT1}

\tyrule{%
    \AxiomC{$\mathtt{t_2 \longrightarrow t_2'}$}
    \UnaryInfC{$\mathtt{v_1}$ \texttt{>} $\mathtt{t_2} \longrightarrow \mathtt{v_1}$ \texttt{>} $\mathtt{t_2'}$}
    \DisplayProof
}{E-GT2}

\tyrule{%
    \AxiomC{$\mathtt{v_1}$ \texttt{>} $\mathtt{v_2} \longrightarrow$ subject to semantics of type}
    \DisplayProof
}{E-GT}

\tyrule{%
    \AxiomC{$\mathtt{t_1 \longrightarrow t_1'}$}
    \UnaryInfC{$\mathtt{t_1}$ \texttt{>=} $\mathtt{t_2} \longrightarrow \mathtt{t_1'}$ \texttt{>=} $\mathtt{t_2}$}
    \DisplayProof
}{E-GEQ1}

\tyrule{%
    \AxiomC{$\mathtt{t_2 \longrightarrow t_2'}$}
    \UnaryInfC{$\mathtt{v_1}$ \texttt{>=} $\mathtt{t_2} \longrightarrow \mathtt{v_1}$ \texttt{>=} $\mathtt{t_2'}$}
    \DisplayProof
}{E-GEQ2}

\tyrule{%
    \AxiomC{$\mathtt{v_1}$ \texttt{>=} $\mathtt{v_2} \longrightarrow$ subject to semantics of type}
    \DisplayProof
}{E-GEQ}

\tyrule{%
    \AxiomC{$\mathtt{t_1 \longrightarrow t_1'}$}
    \UnaryInfC{$\mathtt{t_1}$ \texttt{::} $\mathtt{t_2} \longrightarrow \mathtt{t_1'}$ \texttt{::} $\mathtt{t_2}$}
    \DisplayProof
}{E-CONS1}

\tyrule{%
    \AxiomC{$\mathtt{t_2 \longrightarrow t_2'}$}
    \UnaryInfC{$\mathtt{v_1}$ \texttt{::} $\mathtt{t_2} \longrightarrow \mathtt{v_1}$ \texttt{::} $\mathtt{t_2'}$}
    \DisplayProof
}{E-CONS2}

\tyrule{%
    \AxiomC{\texttt{\#$\mathtt{l_j}$ $\mathtt{\{l_i=v_i^{i \in 1..n}\}} \longrightarrow \mathtt{v_j}$}}
    \DisplayProof
}{E-PROJ-RCD}

\tyrule{%
    \AxiomC{$\mathtt{t_1 \longrightarrow t_1'}$}
    \UnaryInfC{\texttt{\#$\mathtt{l}$ $\mathtt{t_1} \longrightarrow$ \#$\mathtt{l}$ $\mathtt{t_1'}$}}
    \DisplayProof
}{E-PROJ}

\tyrule{%
    \AxiomC{$\mathtt{t_j \longrightarrow t_j'}$}
    \UnaryInfC{$\mathtt{\{l_i=v_i^{i \in 1..j-1}, l_j=t_j, l_k=t_k^{k \in j+1..n}\}}$}
    \noLine
    \UnaryInfC{$\longrightarrow \mathtt{\{l_i=v_i^{i \in 1..j-1}, l_j=t_j', l_k=t_k^{k \in j+1..n}\}}$}
    \DisplayProof
}{E-RCD}

\tyrule{%
    \AxiomC{$\mathtt{\{l_i=v_i^{i \in 1..j-1}, l_j=v_j, l_k=t_k^{k \in j+1..n}\}}$}
    \noLine
    \UnaryInfC{$\longrightarrow \mathtt{\{l_i=v_i^{i \in 1..j}, l_k=t_k^{k \in j+1..n}\}}$}
    \DisplayProof
}{E-RCDV}

\tyrule{%
    \AxiomC{$\mathtt{t_1}$ $|$ $\mu \longrightarrow \mathtt{t_1'}$ $|$ $\mu'$}
    \UnaryInfC{\texttt{\$}$\mathtt{t_1}$ $|$ $\mu \longrightarrow$ \texttt{\$}$\mathtt{t_1'}$ $|$ $\mu'$}
    \DisplayProof
}{E-DEREF}

\tyrule{%
    \AxiomC{$\mu(l)=\mathtt{v}$}
    \UnaryInfC{\texttt{\$}$l$ $|$ $\mu \longrightarrow$ $\mathtt{v}$ $|$ $\mu$}
    \DisplayProof
}{E-DEREFLOC}

\tyrule{%
    \AxiomC{$\mathtt{t_1}$ $|$ $\mu \longrightarrow \mathtt{t_1'}$ $|$ $\mu'$}
    \UnaryInfC{$\mathtt{t_1}$ \texttt{:=} $\mathtt{t_2}$ $|$ $\mu \longrightarrow \mathtt{t_1'}$ \texttt{:=} $\mathtt{t_2}$ $|$ $\mu'$}
    \DisplayProof
}{E-ASSIGN1}

\tyrule{%
    \AxiomC{$\mathtt{t_2}$ $|$ $\mu \longrightarrow \mathtt{t_2'}$ $|$ $\mu'$}
    \UnaryInfC{$l$ \texttt{:=} $\mathtt{t_2}$ $|$ $\mu \longrightarrow l$ \texttt{:=} $\mathtt{t_2'}$ $|$ $\mu'$}
    \DisplayProof
}{E-ASSIGN2}

\tyrule{%
    \AxiomC{$l$ \texttt{:=} $\mathtt{v_1}$ $|$ $\mu \longrightarrow \mathtt{unit}$ $|$ \texttt{[$l \mapsto \mathtt{v_1}$] $\mu$}}
    \DisplayProof
}{E-ASSIGN}

\tyrule{%
    \AxiomC{$\mathtt{t_1}$ $|$ $\mu \longrightarrow \mathtt{t_1'}$ $|$ $\mu'$}
    \UnaryInfC{\texttt{ref} $\mathtt{t_1}$ $|$ $\mu \longrightarrow$ \texttt{ref} $\mathtt{t_1'}$ $|$ $\mu'$}
    \DisplayProof
}{E-REF}

\tyrule{%
    \AxiomC{$l \notin dom(\mu)$}
    \UnaryInfC{\texttt{ref} $\mathtt{v_1}$ $|$ $\mu \longrightarrow l$ $|$ $(\mu, l \mapsto \mathtt{v_1})$}
    \DisplayProof
}{E-REFV}

\tyrule{
    \AxiomC{\texttt{let} $\mathtt{x_1}$ \texttt{=} $\mathtt{v_1}$ $\mathtt{(x_i}$ \texttt{=} $\mathtt{t_i)^{i \in 2..n}}$}
    \noLine
    \UnaryInfC{\texttt{in} $\mathtt{t_0}$ \texttt{end} $|$ $\mathtt{(\Gamma, \tau)}$}
    \noLine
    \UnaryInfC{$\longrightarrow$ \texttt{let} $\mathtt{x_2}$ \texttt{=} $\mathtt{t_2}$ $\mathtt{(x_i}$ \texttt{=} $\mathtt{t_i)^{i \in 3..n}}$}
    \noLine
    \UnaryInfC{\texttt{in} $\mathtt{t_0}$ \texttt{end} $|$ $\mathtt{(\Gamma, \tau, x_1 \mapsto v_1)}$}
    \DisplayProof
}{E-LETV1}

\tyrule{
    \AxiomC{\texttt{let x =} $\mathtt{v}$ \texttt{in} $\mathtt{t} \longrightarrow$ \texttt{[}$\mathtt{x \mapsto v}$\texttt{]}$\mathtt{t}$}
    \DisplayProof
}{E-LETV2}

\tyrule{
    \AxiomC{$\mathtt{t_1} \longrightarrow \mathtt{t_1'}$}
    \UnaryInfC{\texttt{let} $\mathtt{x_1}$ \texttt{=} $\mathtt{t_1}$ $\mathtt{(x_i}$ \texttt{=} $\mathtt{t_i)^{i \in 2..n}}$ \texttt{in} $\mathtt{t_0}$ \texttt{end}}
    \noLine
    \UnaryInfC{$\longrightarrow$ \texttt{let} $\mathtt{x_1}$ \texttt{=} $\mathtt{t_1'}$ $\mathtt{(x_i}$ \texttt{=} $\mathtt{t_i)^{i \in 2..n}}$ \texttt{in} $\mathtt{t_0}$ \texttt{end}}
    \DisplayProof
}{E-LET}

\tyrule{%
    \AxiomC{$\mathtt{t_i} \longrightarrow \mathtt{t_i'}$}
    \UnaryInfC{$\mathtt{C_i}$ $\mathtt{t_i} \longrightarrow \mathtt{C_i}$ $\mathtt{t_i'}$}
    \DisplayProof
}{E-DATATY}

\tyrule{
    \AxiomC{$\mathtt{t_0} \longrightarrow \mathtt{t_0'}$}
    \UnaryInfC{\texttt{case} $\mathtt{t_0}$ \texttt{of} $\mathtt{C_i}$ $\mathtt{x_i}$ \texttt{=>} $\mathtt{t_i^{i \in 1..n}}$}
    \noLine
    \UnaryInfC{$\longrightarrow$ \texttt{case} $\mathtt{t_0'}$ \texttt{of} $\mathtt{C_i}$ $\mathtt{x_i}$ \texttt{=>} $\mathtt{t_i^{i \in 1..n}}$}
    \DisplayProof
}{E-CASE}

\tyrule{
    \AxiomC{\texttt{case} $\mathtt{C_j}$ $\mathtt{v_j}$ \texttt{of} $\mathtt{C_i}$ $\mathtt{x_i}$ \texttt{=>} $\mathtt{t_i^{i \in 1..n}}$}
    \noLine
    \UnaryInfC{$\longrightarrow$ \texttt{[}$\mathtt{x_j \mapsto v_j}$\texttt{]} $\mathtt{t_j}$}
    \DisplayProof
}{E-CASE-TY}

\tyrule{
    \AxiomC{$\mathtt{v_1} \equiv \mathtt{v_{1,i}^{i \in 0..n-1}}$}
    \UnaryInfC{$\mathtt{v_1}$\texttt{[}$\mathtt{v_2}$\texttt{]} $\longrightarrow \mathtt{v_{1,v_2}}$}
    \DisplayProof
}{E-ARR-ACC}

\tyrule{
    \AxiomC{$\mathtt{t_2} \longrightarrow \mathtt{t_2'}$}
    \UnaryInfC{$\mathtt{v_1}$\texttt{[}$\mathtt{t_2}$\texttt{]} $\longrightarrow \mathtt{v_1}$\texttt{[}$\mathtt{t_2'}$\texttt{]}}
    \DisplayProof
}{E-ARR-ACC1}

\tyrule{
    \AxiomC{$\mathtt{t_1}$ $|$ $\omega \longrightarrow \mathtt{t_1'}$ $|$ $\omega'$}
    \UnaryInfC{\texttt{sw} $\mathtt{t_1}$ $|$ $\omega \longrightarrow$ \texttt{sw} $\mathtt{t_1'}$ $|$ $\omega'$}
    \DisplayProof
}{E-SW}

\tyrule{
    \AxiomC{$w \notin dom(\omega)$}
    \UnaryInfC{\texttt{sw} $\mathtt{v_1}$ $|$ $\omega \longrightarrow w$ $|$ $(\omega, w \mapsto \mathtt{v_1})$}
    \DisplayProof
}{E-SWV}

\tyrule{%
    \AxiomC{$\mathtt{t_1}$ $|$ $\sigma \longrightarrow \mathtt{t_1'}$ $|$ $\sigma'$}
    \UnaryInfC{\texttt{unsw }$\mathtt{t_1}$ $|$ $\sigma \longrightarrow$ \texttt{unsw} $\mathtt{t_1'}$ $|$ $\sigma'$}
    \DisplayProof
}{E-UNSW}

\tyrule{%
    \AxiomC{$\omega(w)=\mathtt{v}$}
    \UnaryInfC{\texttt{unsw }$w$ $|$ $\omega \longrightarrow$ $\mathtt{v}$ $|$ $\omega$}
    \DisplayProof
}{E-UNSWWRAP}

\section{Proofs}
\subsection{Inversion of Typing Relation}
\label{sec:inversion-of-typing}
\begin{enumerate}
    %
    \item If $\mathtt{\Gamma \vdash x : R}$, then $\mathtt{x : R \in \Gamma}$.
    \item If $\mathtt{\Gamma \vdash \lambda x : T_1 . t_2 : R}$, then $\mathtt{R = T_1 \rightarrow R_2}$ for some $\mathtt{R_2}$ with $\mathtt{\Gamma, x : T_1 \vdash t_2 : R_2}$.
    \item If $\mathtt{\Gamma \vdash t_1}$ $\mathtt{t_2 : R}$ then there is some type $\mathtt{T_{11}}$ such that $\mathtt{\Gamma \vdash t_1 : T_{11} \rightarrow R}$ and that $\mathtt{\Gamma \vdash t_2 : T_{11}}$.
    %
    \item If $\langle integer \rangle \mathtt{: R}$, then \texttt{R = int}.
    \item If $\langle real \rangle \mathtt{: R}$, then \texttt{R = real}.
    \item If $\langle string \rangle \mathtt{: R}$, then \texttt{R = string}.
    \item If \texttt{'b:0} $\mathtt{: R}$, then \texttt{R = bit}.
    \item If \texttt{'b:1} $\mathtt{: R}$, then \texttt{R = bit}.
    \item If \texttt{nil} $\mathtt{: R}$, then \texttt{R = $\mathtt{T_S}$ list}.
    \item If \texttt{()} $\mathtt{: R}$, then \texttt{R = unit}.
    %
    \item If \textasciitilde $\mathtt{t_1 : int}$, then \texttt{$\mathtt{t_1 : int}$}.
    \item If $\mathtt{t_1}$ \texttt{+} $\mathtt{t_2 : R}$, then \texttt{R = int}, $\mathtt{t_1 : int}$ and $\mathtt{t_2 : int}$.
    \item If $\mathtt{t_1}$ \texttt{-} $\mathtt{t_2 : R}$, then \texttt{R = int}, $\mathtt{t_1 : int}$ and $\mathtt{t_2 : int}$.
    \item If $\mathtt{t_1}$ \texttt{*} $\mathtt{t_2 : R}$, then \texttt{R = int}, $\mathtt{t_1 : int}$ and $\mathtt{t_2 : int}$.
    \item If $\mathtt{t_1}$ \texttt{/} $\mathtt{t_2 : R}$, then \texttt{R = int}, $\mathtt{t_1 : int}$ and $\mathtt{t_2 : int}$.
    \item If $\mathtt{t_1}$ \texttt{\%} $\mathtt{t_2 : R}$, then \texttt{R = int}, $\mathtt{t_1 : int}$ and $\mathtt{t_2 : int}$.
    %
    \item If \textasciitilde $\mathtt{t_1 : real}$, then \texttt{$\mathtt{t_1 : real}$}.
    \item If $\mathtt{t_1}$ \texttt{+.} $\mathtt{t_2 : R}$, then \texttt{R = real}, $\mathtt{t_1 : real}$ and $\mathtt{t_2 : real}$.
    \item If $\mathtt{t_1}$ \texttt{-.} $\mathtt{t_2 : R}$, then \texttt{R = real}, $\mathtt{t_1 : real}$ and $\mathtt{t_2 : real}$.
    \item If $\mathtt{t_1}$ \texttt{*.} $\mathtt{t_2 : R}$, then \texttt{R = real}, $\mathtt{t_1 : real}$ and $\mathtt{t_2 : real}$.
    \item If $\mathtt{t_1}$ \texttt{/.} $\mathtt{t_2 : R}$, then \texttt{R = real}, $\mathtt{t_1 : real}$ and $\mathtt{t_2 : real}$.
    %
    \item If \texttt{!}$\mathtt{t_1 : R}$, then \texttt{R = $\mathtt{T_H}$} and \texttt{$\mathtt{t_1 : T_H}$}.
    \item If $\mathtt{t_1}$ \texttt{\&} $\mathtt{t_2 : R}$, then \texttt{R = $\mathtt{T_H}$}, \texttt{$\mathtt{t_1 : T_H}$} and \texttt{$\mathtt{t_2 : T_H}$}.
    \item If $\mathtt{t_1}$ \texttt{|} $\mathtt{t_2 : R}$, then \texttt{R = $\mathtt{T_H}$}, \texttt{$\mathtt{t_1 : T_H}$} and \texttt{$\mathtt{t_2 : T_H}$}.
    \item If $\mathtt{t_1}$ \texttt{$^\wedge$} $\mathtt{t_2 : R}$, then \texttt{R = $\mathtt{T_H}$}, \texttt{$\mathtt{t_1 : T_H}$} and \texttt{$\mathtt{t_2 : T_H}$}.
    %
    \item If \texttt{\&->$\mathtt{t_1 : R}$}, then \texttt{R = bit} and $\mathtt{t_1 : bit}$\texttt{[n]}.
    \item If \texttt{|->$\mathtt{t_1 : R}$}, then \texttt{R = bit} and $\mathtt{t_1 : bit}$\texttt{[n]}.
    \item If \texttt{$^\wedge$->$\mathtt{t_1 : R}$}, then \texttt{R = bit} and $\mathtt{t_1 : bit}$\texttt{[n]}.
    \item If $\mathtt{t_1}$ \texttt{\&\&} $\mathtt{t_2 : R}$, then \texttt{R = bit}, \texttt{$\mathtt{t_1 : bit}$[n]} and \texttt{$\mathtt{t_2 : bit}$[n]}.
    \item If $\mathtt{t_1}$ \texttt{||} $\mathtt{t_2 : R}$, then \texttt{R = bit}, \texttt{$\mathtt{t_1 : bit}$[n]} and \texttt{$\mathtt{t_2 : bit}$[n]}.
    \item If $\mathtt{t_1}$ \texttt{$^{\wedge\wedge}$} $\mathtt{t_2 : R}$, then \texttt{R = bit}, \texttt{$\mathtt{t_1 : bit}$[n]} and \texttt{$\mathtt{t_2 : bit}$[n]}.
    \item If $\mathtt{t_1}$ \texttt{<}\texttt{<} $\mathtt{t_2 : R}$, then \texttt{R = bit[n]}, \texttt{$\mathtt{t_1 : bit}$[n]} and \texttt{$\mathtt{t_2 : bit}$[m]}.
    \item If $\mathtt{t_1}$ \texttt{>}\texttt{>} $\mathtt{t_2 : R}$, then \texttt{R = bit[n]}, \texttt{$\mathtt{t_1 : bit}$[n]} and \texttt{$\mathtt{t_2 : bit}$[m]}.
    \item If $\mathtt{t_1}$ \texttt{>}\texttt{>}\texttt{>} $\mathtt{t_2 : R}$, then \texttt{R = bit[n]}, \texttt{$\mathtt{t_1 : bit}$[n]} and \texttt{$\mathtt{t_2 : bit}$[m]}.
    %
    \item If \texttt{$\mathtt{t_1}$ = $\mathtt{t_2 : R}$}, then \texttt{R = int}, $\mathtt{t_1 : T_S}$ and $\mathtt{t_2 : T_S}$.
    \item If \texttt{$\mathtt{t_1}$ <> $\mathtt{t_2 : R}$}, then \texttt{R = int}, $\mathtt{t_1 : T_S}$ and $\mathtt{t_2 : T_S}$.
    \item If \texttt{$\mathtt{t_1}$ < $\mathtt{t_2 : R}$}, then \texttt{R = int}, $\mathtt{t_1 : T_S}$ and $\mathtt{t_2 : T_S}$.
    \item If \texttt{$\mathtt{t_1}$ <= $\mathtt{t_2 : R}$}, then \texttt{R = int}, $\mathtt{t_1 : T_S}$ and $\mathtt{t_2 : T_S}$.
    \item If \texttt{$\mathtt{t_1}$ > $\mathtt{t_2 : R}$}, then \texttt{R = int}, $\mathtt{t_1 : T_S}$ and $\mathtt{t_2 : T_S}$.
    \item If \texttt{$\mathtt{t_1}$ >= $\mathtt{t_2 : R}$}, then \texttt{R = int}, $\mathtt{t_1 : T_S}$ and $\mathtt{t_2 : T_S}$.
    %
    \item If \texttt{$\mathtt{t_1}$ andalso $\mathtt{t_2 : R}$}, then \texttt{R = int}, $\mathtt{t_1 : int}$ and $\mathtt{t_2 : int}$.
    \item If \texttt{$\mathtt{t_1}$ orelse $\mathtt{t_2 : R}$}, then \texttt{R = int}, $\mathtt{t_1 : int}$ and $\mathtt{t_2 : int}$.
    \item If \texttt{not $\mathtt{t_1 : R}$}, then \texttt{R = int} and $\mathtt{t_1 : int}$.
    %
    \item If \texttt{$\mathtt{t_1}$::$\mathtt{t_2 : R}$}, then \texttt{R = $\mathtt{T_S}$ list}, $\mathtt{t_1 : T_S}$ and \texttt{$\mathtt{t_2 : T_S}$ list}.
    %
    \item If \texttt{if $\mathtt{t_1}$ then $\mathtt{t_2}$ else $\mathtt{t_3 : R}$}, then $\mathtt{t_1 : int}$, $\mathtt{t_2 : R}$, and $\mathtt{t_3 : R}$.
    \item If \texttt{if $\mathtt{t_1}$ then $\mathtt{t_2 : R}$}, then \texttt{R = unit}, $\mathtt{t_1 : int}$ and $\mathtt{t_2 : unit}$.
    %
    \item If \texttt{ref $\mathtt{t_1 : R}$}, then \texttt{R = $\mathtt{T_S}$ ref} and $\mathtt{t_1 : T_S}$.
    \item If \texttt{$\mathtt{t_1}$ := $\mathtt{t_2 : R}$}, then \texttt{R = unit}, \texttt{$\mathtt{t_1 : T_S}$ ref} and $\mathtt{t_2 : T_S}$.
    \item If \texttt{\$$\mathtt{t_1 : R}$}, then \texttt{R = $\mathtt{T_S}$} and \texttt{$\mathtt{t_1 : T_S}$ ref}.
    %
    \item If \texttt{sw $\mathtt{t_1 : R}$}, then \texttt{R = $\mathtt{T_H}$ sw} and $\mathtt{t_1 : T_H}$.
    %
    \item If \texttt{$\mathtt{t_1}$[$\mathtt{t_2}$]$\mathtt{: R}$}, then \texttt{R = $\mathtt{T_H}$}, \texttt{$\mathtt{t_1 : T_H}$[n]} and $\mathtt{t_2 : int}$.
    %
    \item If $\mathtt{\{l_i = t_i^{i \in 1..n}\} : R}$, then \texttt{R = $\mathtt{\{l_i : T_i^{i \in 1..n}\}}$} and for each $i$, $\mathtt{t_i : T_i}$.
    \item If \texttt{\#$\mathtt{l_j}$ $\mathtt{t_1 : R}$}, then \texttt{R = $\mathtt{T_j}$} and $\mathtt{t_1 : \{l_i : T_i^{i \in 1..n}\}}$.
    %
    \item If \texttt{let x = $\mathtt{t_1}$ in $\mathtt{t_2}$ end $\mathtt{: R}$}, then \texttt{$\mathtt{\Gamma, \tau \vdash}$ R = $\mathtt{T_2}$}, $\mathtt{\Gamma, \tau \vdash t_1 : T_1}$ and $\mathtt{\Gamma, \tau, x : T_1 \vdash t_2 : T_2}$.
    %
    \item If \texttt{$\mathtt{C_j}$ $\mathtt{t : R}$}, then \texttt{R = $\mathtt{\langle C_i : T_i \rangle^{i \in 1..n}}$} and $\mathtt{t : T_j}$.
    %
    \item If \texttt{case $\mathtt{t_0}$ of $\mathtt{C_i}$ $\mathtt{x_i}$ => $\mathtt{t_i^{i \in 1..n} : R}$}, then \texttt{R = T}, $\mathtt{t_0 : \langle C_i : T_i \rangle^{i \in 1..n}}$, and for each $i$, $\mathtt{\Gamma, x_i : T_i \vdash t_i : T}$.
\end{enumerate}

\subsection{Proof of Progress}
\label{sec:progress}
\begin{multicols}{2}
\begin{proof}
    By induction on a derivation of $\mathtt{t : T}$.
    \begin{description}
        \item[\textit{\textmd{Case}} \normalfont T-INT, T-REAL, T-STRING, T-BIT, T-NIL:]\mbox{}\\
        \mbox{}\\
        Immediate since \texttt{t} is a value.
        
        \item[\textit{\textmd{Case}} \normalfont T-VAR:]\mbox{}\\
        \mbox{}\\
        Cannot occur since $\mathtt{t}$ must be closed as per the hypothesis.
        
        \item[\textit{\textmd{Case}} \normalfont T-ABS:]\mbox{}\\
        \mbox{}\\
        Immediate since \texttt{t} is a value.
        
        \item[\textit{\textmd{Case}} \normalfont T-APP:]\mbox{}\\
        \begin{quote}
            \texttt{t = $\mathtt{t_1}$ $\mathtt{t_2}$}\\
            $\mathtt{\vdash t_1 : T_{11} \rightarrow T_{12}}$\\
            $\mathtt{\vdash t_2 : T_{11}}$
        \end{quote}
        By the induction hypothesis, either $\mathtt{t_1}$ is a value or else there is some $\mathtt{t_1'}$ for which $\mathtt{t_1} \longrightarrow \mathtt{t_1'}$, and likewise for $\mathtt{t_2}$. If $\mathtt{t_1} \longrightarrow \mathtt{t_1'}$, then by E-APP1, $\mathtt{t} \longrightarrow \mathtt{t_1'}$ $\mathtt{t_2}$. On the other hand, if $\mathtt{t_1}$ is a value and $\mathtt{t_2} \longrightarrow \mathtt{t_2'}$, then by E-APP2, $\mathtt{t} \longrightarrow \mathtt{t_1}$ $\mathtt{t_2'}$. Finally, if both $\mathtt{t_1}$ and $\mathtt{t_2}$ are values, then case 5 of the canonical forms lemma tells us that $\mathtt{t_1}$ has the form $\mathtt{\lambda x : T_{11} . t_{12}}$, and so by E-APPABS, \texttt{$\mathtt{t} \longrightarrow$ [$\mathtt{x \mapsto t_2}$]$\mathtt{t_12}$}.
        
        \item[\textit{\textmd{Case}} \normalfont T-INT-NEG:]\mbox{}\\
        \vspace{-3ex}
        \begin{quote}
            \texttt{t = } \textasciitilde$\mathtt{t_1}$\\
            $\mathtt{t_1 : int}$
        \end{quote}
        By the induction hypothesis, either $\mathtt{t_1}$ is a value or else there is some $\mathtt{t_1'}$ such that $\mathtt{t_1} \longrightarrow \mathtt{t_1'}$. If $\mathtt{t_1}$ is a value, then case 1 of the canonical forms lemma assures us that it is an integer value as described in Section \ref{sec:swvgrammar} with existing and valid semantic meaning for negation, yielding a value. On the other hand, if $\mathtt{t_1} \longrightarrow \mathtt{t_1'}$, then by E-NEG, $\mathtt{t} \longrightarrow$ \textasciitilde $\mathtt{t_1'}$.
        
        \item[\textit{\textmd{Case}} \normalfont T-INT-ADD:]\mbox{}\\
        \vspace{-3ex}
        \begin{quote}
            \texttt{t = $\mathtt{t_1}$ + $\mathtt{t_2}$}\\
            $\mathtt{t_1 : int}$\\
            $\mathtt{t_2 : int}$
        \end{quote}
        By the induction hypothesis, either $\mathtt{t_1}$ is a value or else there is some $\mathtt{t_1'}$ such that $\mathtt{t_1} \longrightarrow \mathtt{t_1'}$. If $\mathtt{t_1} \longrightarrow \mathtt{t_1'}$, then by E-INT-ADD1, \texttt{$\mathtt{t} \longrightarrow \mathtt{t_1'}$ + $\mathtt{t_2}$}. On the other hand, if $\mathtt{t_1}$ is a value, then case 1 of the canonical forms lemma assures us that it is an integer value as described in Section \ref{sec:swvgrammar}. Further in this case, by the induction hypothesis, either $\mathtt{t_2}$ is a value or else there is some $\mathtt{t_2'}$ such that $\mathtt{t_2} \longrightarrow \mathtt{t_2'}$. If $\mathtt{t_2} \longrightarrow \mathtt{t_2'}$, then by E-INT-ADD2, \texttt{$\mathtt{t} \longrightarrow \mathtt{v_1}$ + $\mathtt{t_2'}$}. On the other hand, if $\mathtt{t_2}$ is a value, then case 1 of the canonical forms lemma assures us that it is an integer value as described in Section \ref{sec:swvgrammar}. Thus, both $\mathtt{t_1}$ and $\mathtt{t_2}$ are integer values in this case and addition has a well-defined semantic meaning thereby yielding a value.
        
        \item[\textit{\textmd{Case}} \normalfont T-INT-(SUB/MUL/DIV/MOD):]\mbox{}\\
        \mbox{}\\
        Similar to T-INT-ADD.
        
        \item[\textit{\textmd{Case}} \normalfont T-REAL-NEG:]\mbox{}\\
        \vspace{-3ex}
        \begin{quote}
            \texttt{t = } \textasciitilde$\mathtt{t_1}$\\
            $\mathtt{t_1 : real}$
        \end{quote}
        By the induction hypothesis, either $\mathtt{t_1}$ is a value or else there is some $\mathtt{t_1'}$ such that $\mathtt{t_1} \longrightarrow \mathtt{t_1'}$. If $\mathtt{t_1}$ is a value, then case 2 of the canonical forms lemma assures us that it is a value in the domain of real numbers as described in Section \ref{sec:swvgrammar} with existing and valid semantic meaning for negation, yielding a value. On the other hand, if $\mathtt{t_1} \longrightarrow \mathtt{t_1'}$, then by E-NEG, $\mathtt{t} \longrightarrow$ \textasciitilde $\mathtt{t_1'}$.
        
        \item[\textit{\textmd{Case}} \normalfont T-REAL-ADD:]\mbox{}\\
        \vspace{-3ex}
        \begin{quote}
            \texttt{t = $\mathtt{t_1}$ +. $\mathtt{t_2}$}\\
            $\mathtt{t_1 : real}$\\
            $\mathtt{t_2 : real}$
        \end{quote}
        By the induction hypothesis, either $\mathtt{t_1}$ is a value or else there is some $\mathtt{t_1'}$ such that $\mathtt{t_1} \longrightarrow \mathtt{t_1'}$. If $\mathtt{t_1} \longrightarrow \mathtt{t_1'}$, then by E-REAL-ADD1, \texttt{$\mathtt{t} \longrightarrow \mathtt{t_1'}$ +. $\mathtt{t_2}$}. On the other hand, if $\mathtt{t_1}$ is a value, then case 2 of the canonical forms lemma assures us that it is a value in the domain of real numbers as described in Section \ref{sec:swvgrammar}. Further in this case, by the induction hypothesis, either $\mathtt{t_2}$ is a value or else there is some $\mathtt{t_2'}$ such that $\mathtt{t_2} \longrightarrow \mathtt{t_2'}$. If $\mathtt{t_2} \longrightarrow \mathtt{t_2'}$, then by E-REAL-ADD2, \texttt{$\mathtt{t} \longrightarrow \mathtt{v_1}$ +. $\mathtt{t_2'}$}. On the other hand, if $\mathtt{t_2}$ is a value, then case 2 of the canonical forms lemma assures us that it is a value in the domain of real numbers as described in Section \ref{sec:swvgrammar}. Thus, both $\mathtt{t_1}$ and $\mathtt{t_2}$ are real values in this case and addition has a well-defined semantic meaning thereby yielding a value.
        
        \item[\textit{\textmd{Case}} \normalfont T-REAL-(SUB/MUL/DIV):]\mbox{}\\
        \mbox{}\\
        Similar to T-REAL-ADD.
        
        \item[\textit{\textmd{Case}} \normalfont T-BIT-NEG:]\mbox{}\\
        \vspace{-3ex}
        \begin{quote}
            \texttt{t = !}$\mathtt{t_1}$\\
            $\mathtt{t_1 : T_H}$
        \end{quote}
        By the induction hypothesis, either $\mathtt{t_1}$ is a value or else there is some $\mathtt{t_1'}$ such that $\mathtt{t_1} \longrightarrow \mathtt{t_1'}$. If $\mathtt{t_1} \longrightarrow \mathtt{t_1'}$ then by E-BIT-NEG1, \texttt{$\mathtt{t} \longrightarrow$ !$\mathtt{t_1'}$}. On the other hand, if $\mathtt{t_1}$ is a value then it is a hardware value and we can perform structural induction on the possible types. If it is a bit then by E-BIT-NEG we produce a logical not-gate. If it is a hardware array then by E-BIT-NEG2 we move the negation inside and apply to each element. Similarly, if it is a hardware record then by E-BIT-NEG3 we move the negation inside and apply to each field element.
        
        \item[\textit{\textmd{Case}} \normalfont T-AND:]\mbox{}\\
        \vspace{-3ex}
        \begin{quote}
            \texttt{$\mathtt{t_1}$ \& $\mathtt{t_2}$}\\
            $\mathtt{t_1 : T_H}$\\
            $\mathtt{t_2 : T_H}$
        \end{quote}
        By the induction hypothesis, either $\mathtt{t_1}$ is a value or else there is some $\mathtt{t_1'}$ such that $\mathtt{t_1} \longrightarrow \mathtt{t_1'}$. If $\mathtt{t_1} \longrightarrow \mathtt{t_1'}$, then by E-AND1, \texttt{$\mathtt{t} \longrightarrow \mathtt{t_1'}$ \& $\mathtt{t_2}$}. On the other hand, if $\mathtt{t_1}$ is a value, then it is a hardware value. Further in this case, by the induction hypothesis, either $\mathtt{t_2}$ is a value or else there is some $\mathtt{t_2'}$ such that $\mathtt{t_2} \longrightarrow \mathtt{t_2'}$. If $\mathtt{t_2} \longrightarrow \mathtt{t_2'}$, then by E-AND2, \texttt{$\mathtt{t} \longrightarrow \mathtt{v_1}$ \& $\mathtt{t_2'}$}. On the other hand, if $\mathtt{t_2}$ is a value, then it is a hardware value. If, both $\mathtt{t_1}$ and $\mathtt{t_2}$ are hardware values then we can perform structural induction on the possible types. If $\mathtt{T_H = bit}$ then by the derived-term definition, \texttt{$\mathtt{t} \longrightarrow$ \&->\#[$\mathtt{t_1}$, $\mathtt{t_2}$]}. If it is a hardware array then by E-AND3 we move the operation inside and apply to each pair of elements. Similarly, if it is a hardware record then by E-AND4 we move the operation inside and apply to each pair of field elements.
        
        \item[\textit{\textmd{Case}} \normalfont T-(OR/XOR):]\mbox{}\\
        \mbox{}\\
        Similar to T-AND.
        
        \item[\textit{\textmd{Case}} \normalfont T-AND-RED:]\mbox{}\\
        \vspace{-3ex}
        \begin{quote}
            \texttt{t = \&->}$\mathtt{t_1}$\\
            $\mathtt{t_1 : bit}$\texttt{[n]}
        \end{quote}
        By the induction hypothesis, either $\mathtt{t_1}$ is a value or else there is some $\mathtt{t_1'}$ such that $\mathtt{t_1} \longrightarrow \mathtt{t_1'}$. If $\mathtt{t_1} \longrightarrow \mathtt{t_1'}$ then by E-AND-RED1, \texttt{$\mathtt{t} \longrightarrow$ \&->$\mathtt{t_1'}$}. On the other hand if $\mathtt{t_1}$ is a value, then it is a bit array and is evaluated by E-AND-RED.
        
        \item[\textit{\textmd{Case}} \normalfont T-(OR/XOR)-RED:]\mbox{}\\
        \mbox{}\\
        Similar to T-AND-RED.
        
        \item[\textit{\textmd{Case}} \normalfont T-SLL:]\mbox{}\\
        \vspace{-3ex}
        \begin{quote}
            \texttt{$\mathtt{t_1}$ <}\texttt{< $\mathtt{t_2}$}\\
            $\mathtt{t_1 : bit}$\texttt{[n]}\\
            $\mathtt{t_2 : bit}$\texttt{[m]}
        \end{quote}
        By the induction hypothesis, either $\mathtt{t_1}$ is a value or else there is some $\mathtt{t_1'}$ such that $\mathtt{t_1} \longrightarrow \mathtt{t_1'}$. If $\mathtt{t_1} \longrightarrow \mathtt{t_1'}$ then by E-SLL1, \texttt{$\mathtt{t} \longrightarrow \mathtt{t_1'}$ <}\texttt{< $\mathtt{t_2}$}. On the other hand, if $\mathtt{t_1}$ is a value, then either $\mathtt{t_2}$ is a value or else there is some $\mathtt{t_2'}$ such that $\mathtt{t_2} \longrightarrow \mathtt{t_2'}$. If $\mathtt{t_2} \longrightarrow \mathtt{t_2'}$ then by E-SLL2, \texttt{$\mathtt{t} \longrightarrow \mathtt{v_1}$ <}\texttt{< $\mathtt{t_2'}$}. On the other hand, if $\mathtt{t_2}$ is a value, then both $\mathtt{t_1}$ and $\mathtt{t_2}$ are values and by the definition in Section \ref{sec:hwvgrammar}, \texttt{t} is a value.
        
        \item[\textit{\textmd{Case}} \normalfont T-(SRL/SRA):]\mbox{}\\
        \mbox{}\\
        Similar to T-SLL.
        
        \item[\textit{\textmd{Case}} \normalfont T-EQ:]\mbox{}\\
        \vspace{-3ex}
        \begin{quote}
            \texttt{t = $\mathtt{t_1}$ = $\mathtt{t_2}$}\\
            $\mathtt{t_1 : T_S}$\\
            $\mathtt{t_2 : T_S}$
        \end{quote}
        By the induction hypothesis, either $\mathtt{t_1}$ is a value or else there is some $\mathtt{t_1'}$ such that $\mathtt{t_1} \longrightarrow \mathtt{t_1'}$. If $\mathtt{t_1} \longrightarrow \mathtt{t_1'}$, then by E-EQ1, \texttt{$\mathtt{t} \longrightarrow \mathtt{t_1'}$ = $\mathtt{t_2}$}. On the other hand, if $\mathtt{t_1}$ is a value, then it may be of any software type. Further in this case, by the induction hypothesis, either $\mathtt{t_2}$ is a value or else there is some $\mathtt{t_2'}$ such that $\mathtt{t_2} \longrightarrow \mathtt{t_2'}$. If $\mathtt{t_2} \longrightarrow \mathtt{t_2'}$, then by E-EQ2, \texttt{$\mathtt{t} \longrightarrow \mathtt{v_1}$ = $\mathtt{t_2'}$}. On the other hand, if $\mathtt{t_2}$ is a value, then it may be of any software type. Further, both $\mathtt{t_1}$ and $\mathtt{t_2}$ are values in this case and equality has semantic meaning defined as per the type definition, thereby yielding a value.
        
        \item[\textit{\textmd{Case}} \normalfont T-(NEQ/LT/LEQ/GT/GEQ):]\mbox{}\\
        \mbox{}\\
        Similar to T-EQ.
        
        \item[\textit{\textmd{Case}} \normalfont T-NOT:]\mbox{}\\
        \vspace{-3ex}
        \begin{quote}
            \texttt{t = not $\mathtt{t_1}$}\\
            $\mathtt{t_1 : int}$\\
        \end{quote}
        By the induction hypothesis, either $\mathtt{t_1}$ is a value or else there is some $\mathtt{t_1'}$ such that $\mathtt{t_1} \longrightarrow \mathtt{t_1'}$. In either case, \texttt{t} can be evaluated by E-NOT, namely $\mathtt{t} \longrightarrow$ \texttt{if $\mathtt{t_1}$ then 0 else 1}.
        
        \item[\textit{\textmd{Case}} \normalfont T-CONS:]\mbox{}\\
        \vspace{-3ex}
        \begin{quote}
            \texttt{t = $\mathtt{t_1}$::$\mathtt{t_2}$}\\
            $\mathtt{t_1 : T_S}$\\
            $\mathtt{t_2 : T_S}$ \texttt{list}
        \end{quote}
        By the induction hypothesis, either $\mathtt{t_1}$ is a value or else there is some $\mathtt{t_1'}$ such that $\mathtt{t_1} \longrightarrow \mathtt{t_1'}$. If $\mathtt{t_1} \longrightarrow \mathtt{t_1'}$, then by E-CONS1, \texttt{$\mathtt{t} \longrightarrow \mathtt{t_1'}$ :: $\mathtt{t_2}$}. On the other hand, if $\mathtt{t_1}$ is a value, then either $\mathtt{t_2}$ is a value or else there is some $\mathtt{t_2'}$ such that $\mathtt{t_2} \longrightarrow \mathtt{t_2'}$. If $\mathtt{t_2} \longrightarrow \mathtt{t_2'}$, then by E-CONS2, \texttt{$\mathtt{t} \longrightarrow \mathtt{v_1}$ :: $\mathtt{t_2'}$}. On the other hand, if $\mathtt{t_2}$ is a value, then both $\mathtt{t_1}$ and $\mathtt{t_2}$ are values and list concatenation evaluates under well-defined semantics yielding a value.
        
        \item[\textit{\textmd{Case}} \normalfont T-IFELSE:]\mbox{}\\
        \vspace{-3ex}
        \begin{quote}
            \texttt{t = if $\mathtt{t_1}$ then $\mathtt{t_2}$ else $\mathtt{t_3}$}\\
            $\mathtt{t_1 : int}$\\
            $\mathtt{t_2 : T}$\\
            $\mathtt{t_3 : T}$
        \end{quote}
        By the induction hypothesis, either $\mathtt{t_1}$ is a value or else there is some $\mathtt{t_1'}$ such that $\mathtt{t_1} \longrightarrow \mathtt{t_1'}$. If $\mathtt{t_1} \longrightarrow \mathtt{t_1'}$, then by E-IFELSE, \texttt{$\mathtt{t} \longrightarrow$ if $\mathtt{t_1'}$ then $\mathtt{t_2}$ else $\mathtt{t_3}$}. On the other hand, if $\mathtt{t_1}$ is a value, then case 1 of the canonical forms lemma assures us that it is an integer value as described in Section \ref{sec:swvgrammar}. Further, it is either zero or non-zero, and can be evaluated by either E-IFELSE-T or E-IFELSE-F.
    
        \item[\textit{\textmd{Case}} \normalfont T-REF:]\mbox{}\\
        \vspace{-3ex}
        \begin{quote}
            \texttt{t = ref $\mathtt{t_1}$}\\
            $\mathtt{t_1 : T_S}$
        \end{quote}
        By the induction hypothesis, either $\mathtt{t_1}$ is a value or else there is some $\mathtt{t_1'}$ such that $\mathtt{t_1} \longrightarrow \mathtt{t_1'}$. If $\mathtt{t_1} \longrightarrow \mathtt{t_1'}$, then by E-REF, \texttt{$\mathtt{t} \longrightarrow$ ref $\mathtt{t_1'}$}. On the other hand, if $\mathtt{t_1}$ is a value, then by E-REFV, \texttt{$\mathtt{t} \longrightarrow l | (\mu, l \mapsto \mathtt{v_1})$}.
        
        \item[\textit{\textmd{Case}} \normalfont T-SW:]\mbox{}\\
        \vspace{-3ex}
        \begin{quote}
            \texttt{t = sw $\mathtt{t_1}$}\\
            $\mathtt{t_1 : T_H}$
        \end{quote}
        By the induction hypothesis, either $\mathtt{t_1}$ is a value or else there is some $\mathtt{t_1'}$ such that $\mathtt{t_1} \longrightarrow \mathtt{t_1'}$. If $\mathtt{t_1} \longrightarrow \mathtt{t_1'}$ then by E-SW, \texttt{$\mathtt{t} \longrightarrow$ sw $\mathtt{t_1'}$}. On the other hand, if $\mathtt{t_1}$ is a value then by E-SWV, \texttt{$\mathtt{t} \longrightarrow w | (\sigma, w \mapsto \mathtt{v_1})$}.
        
        \item[\textit{\textmd{Case}} \normalfont T-UNSW:]\mbox{}\\
        \vspace{-3ex}
        \begin{quote}
            \texttt{t = unsw $\mathtt{t_1}$}\\
            \texttt{$\mathtt{t_1 : T_H}$ sw}
        \end{quote}
        By the induction hypothesis, either $\mathtt{t_1}$ is a value or else there is some $\mathtt{t_1'}$ such that $\mathtt{t_1} \longrightarrow \mathtt{t_1'}$. If $\mathtt{t_1} \longrightarrow \mathtt{t_1'}$ then by E-UNSW, \texttt{$\mathtt{t} \longrightarrow$ hw $\mathtt{t_1'}$}. On the other hand, if $\mathtt{t_1}$ is a value then by E-UNSWWRAP, \texttt{t $\longrightarrow$ $w$ $|$ $\sigma$}.
        
        \item[\textit{\textmd{Case}} \normalfont T-ASSIGN:]\mbox{}\\
        \vspace{-3ex}
        \begin{quote}
            \texttt{t = $\mathtt{t_1}$ := $\mathtt{t_2}$ $|$ $\mu$}\\
            \texttt{$\mathtt{t_1 : T_S}$ ref}\\
            $\mathtt{t_2 : T_S}$
        \end{quote}
        By the induction hypothesis, either $\mathtt{t_1}$ is a value or else there is some $\mathtt{t_1'}$ such that $\mathtt{t_1}$ $|$ $\mu$ $\longrightarrow \mathtt{t_1'}$ $|$ $\mu'$. If $\mathtt{t_1}$ $|$ $\mu$ $\longrightarrow \mathtt{t_1'}$ $|$ $\mu'$, then by E-ASSIGN1, \texttt{$\mathtt{t} \longrightarrow \mathtt{t_1'}$ := $\mathtt{t_2}$ $|$ $\mu'$}. On the other hand, if $\mathtt{t_1}$ is a value, then by the induction hypothesis, either $\mathtt{t_2}$ is a value or else there is some $\mathtt{t_2'}$ such that $\mathtt{t_2}$ $|$ $\mu$ $\longrightarrow \mathtt{t_2'}$ $|$ $\mu'$. If $\mathtt{t_2}$ $|$ $\mu$ $\longrightarrow \mathtt{t_2'}$ $|$ $\mu'$, then by E-ASSIGN2, \texttt{$\mathtt{t} \longrightarrow \mathtt{v_1}$ := $\mathtt{t_2'}$ $|$ $\mu'$}. On the other hand, if $\mathtt{t_2}$ is a value, then both $\mathtt{t_1}$ and $\mathtt{t_2}$ are values and so by E-ASSIGN, \texttt{$\mathtt{t} \longrightarrow$ unit $|$ [$l \mapsto \mathtt{v_1}$]$\mu$}.
        
        \item[\textit{\textmd{Case}} \normalfont T-DEREF:]\mbox{}\\
        \vspace{-3ex}
        \begin{quote}
            \texttt{t = \$$\mathtt{t_1}$ $|$ $\mu$}\\
            \texttt{$\mathtt{t_1 : T_S}$ ref}
        \end{quote}
        By the induction hypothesis, either $\mathtt{t_1}$ is a value or else there is some $\mathtt{t_1'}$ and $\mu'$ such that $\mathtt{t_1}$ $|$ $\mu \longrightarrow \mathtt{t_1'}$ $|$ $\mu'$. If $\mathtt{t_1}$ $|$ $\mu \longrightarrow \mathtt{t_1'}$ $|$ $\mu'$, then by E-DEREF, \texttt{t $\longrightarrow$ \$$\mathtt{t_1'}$ $|$ $\mu'$}. On the other hand, if $\mathtt{t_1}$ is a value then case 6 of the canonical forms lemma assures us that it is a location $l$ in store $\mu$, and so assuming that \texttt{$\mu(l)$ = v} then by E-DEREFLOC, \texttt{t $\longrightarrow$ v $|$ $\mu$}.
        
        \item[\textit{\textmd{Case}} \normalfont T-ARR-ACC:]\mbox{}\\
        \vspace{-3ex}
        \begin{quote}
            \texttt{t = $\mathtt{t_1}$[$\mathtt{t_2}$]}\\
            \texttt{$\mathtt{t_1 : T_H}$[n]}\\
            \texttt{$\mathtt{t_2 : int}$}
        \end{quote}
        By the induction hypothesis, either $\mathtt{t_1}$ is a value or else there is some $\mathtt{t_1'}$ such that $\mathtt{t_1} \longrightarrow \mathtt{t_1'}$. Since $\mathtt{t_1}$ is of hardware type it is a value by definition. By the induction hypothesis, either $\mathtt{t_2}$ is a value or else there is some $\mathtt{t_2'}$ such that $\mathtt{t_2} \longrightarrow \mathtt{t_2'}$. If $\mathtt{t_2} \longrightarrow \mathtt{t_2'}$, then by E-ARR-ACC1, \texttt{t $\longrightarrow \mathtt{t_1}$[$\mathtt{t_2'}$]}. On the other hand, if $\mathtt{t_2}$ is a value, then both $\mathtt{t_1}$ and $\mathtt{t_2}$ are values and by E-ARR-ACC, $\mathtt{t} \longrightarrow \mathtt{v_{1,v_2}}$.
        
        \item[\textit{\textmd{Case}} \normalfont T-RCD:]\mbox{}\\
        \vspace{-3ex}
        \begin{quote}
            \texttt{t = \{$\mathtt{l_i}$ = $\mathtt{t_i^{i \in 1..n}}$\}}\\
            for each $i$, \texttt{$\mathtt{t_i : T_i}$}
        \end{quote}
        By the induction hypothesis, for all $j$ either $\mathtt{t_j}$ is a value or else there is some $\mathtt{t_j'}$ such that $\mathtt{t_j} \longrightarrow \mathtt{t_j'}$. If $\mathtt{t_j} \longrightarrow \mathtt{t_j'}$, then by E-RCD, $\mathtt{t} \longrightarrow \mathtt{\{l_i=v_i^{i \in 1..j-1}, l_j=t_j', l_k=t_k^{k \in j+1..n}\}}$. On the other hand, if $\mathtt{t_j}$ is a value, then by E-RCDV, $\mathtt{t} \longrightarrow \mathtt{\{l_i=v_i^{i \in 1..j}, l_k=t_k^{k \in j+1..n}\}}$.
        
        \item[\textit{\textmd{Case}} \normalfont T-PROJ:]\mbox{}\\
        \vspace{-3ex}
        \begin{quote}
            \texttt{t = \#$\mathtt{l_j}$ $\mathtt{t_1}$}\\
            $\mathtt{t_1 : \{l_i : T_i^{i \in 1..n}\}}$
        \end{quote}
        By the induction hypothesis, either $\mathtt{t_1}$ is a value or else there is some $\mathtt{t_1'}$ such that $\mathtt{t_1} \longrightarrow \mathtt{t_1'}$. If $\mathtt{t_1} \longrightarrow \mathtt{t_1'}$ then by E-PROJ, \texttt{t $\longrightarrow$ \#$\mathtt{l_j}$ $\mathtt{t_1'}$}. On the other hand, if $\mathtt{t_1}$ is a value then case 7 of the canonical forms lemma assures us that it is a record value of form $\mathtt{\{l_i = v_i^{i \in 1..n}\}}$ and so by E-PROJ-RCD, $\mathtt{t} \longrightarrow \mathtt{v_j}$.
        
        \item[\textit{\textmd{Case}} \normalfont T-LET:]\mbox{}\\
        \vspace{-3ex}
        \begin{quote}
            \texttt{t = let $\mathtt{x_1 = t_1 (x_i}$ = $\mathtt{t_i)^{i \in 2..n}}$ in $\mathtt{t_0}$ end}
            \texttt{$\mathtt{t_1 : T_1}$}
        \end{quote}
        By the induction hypothesis, $\mathtt{t_1}$ is either a value or else there is some $\mathtt{t_1'}$ such that $\mathtt{t_1} \longrightarrow \mathtt{t_1'}$. If $\mathtt{t_1} \longrightarrow \mathtt{t_1'}$, then by E-LET, \texttt{t $\longrightarrow$ let $\mathtt{x_1 = t_1' (x_i}$ = $\mathtt{t_i)^{i \in 2..n}}$ in $\mathtt{t_0}$ end}. On the other hand, if $\mathtt{t_1}$ is a value then by the definition of let-bindings $n \geq 1$. If $n = 1$, then by E-LETV2, \texttt{t $\longrightarrow$ [$\mathtt{x_1 \mapsto v_1}$]$\mathtt{t_0}$}. On the other hand, if $n > 1$, then by E-LETV1, \texttt{t $\longrightarrow$ let $\mathtt{x_2 = t_2 (x_i}$ = $\mathtt{t_i)^{i \in 3..n}}$ in $\mathtt{t_0}$ end}.
        
        \item[\textit{\textmd{Case}} \normalfont T-DATATY:]\mbox{}\\
        \vspace{-3ex}
        \begin{quote}
            \texttt{t = $\mathtt{C_j}$ $\mathtt{t_1}$}\\
            \texttt{$\mathtt{t_1 : T_j}$}
        \end{quote}
        By the induction hypothesis, $\mathtt{t_1}$ is either a value or else there is some $\mathtt{t_1'}$ such that $\mathtt{t_1} \longrightarrow \mathtt{t_1'}$. If $\mathtt{t_1} \longrightarrow \mathtt{t_1'}$, then by E-DATATY, \texttt{t $\longrightarrow \mathtt{C_j}$ $\mathtt{t_1'}$}. On the other hand, if $\mathtt{t_1}$ is a value then case 8 of the canonical forms lemma assures us that it is a datatype value of form $\langle \mathtt{C_j = t_1} \rangle$.
        
        \item[\textit{\textmd{Case}} \normalfont T-CASE:]\mbox{}\\
        \vspace{-3ex}
        \begin{quote}
            \texttt{t = case $\mathtt{t_0}$ of $\mathtt{C_i}$ $\mathtt{x_i}$ => $\mathtt{t_i^{i \in 1..n}}$}\\
            $\mathtt{t_0 : \langle C_i : T_i \rangle^{i \in 1..n}}$
        \end{quote}
        By the induction hypothesis, $\mathtt{t_0}$ is either a value or else there is some $\mathtt{t_0'}$ such that $\mathtt{t_0} \longrightarrow \mathtt{t_0'}$. If $\mathtt{t_0} \longrightarrow \mathtt{t_0'}$, then by E-CASE, \texttt{t $\longrightarrow$ case $\mathtt{t_0'}$ of $\mathtt{C_i}$ $\mathtt{x_i}$ => $\mathtt{t_i^{i \in 1..n}}$}. On the other hand, if $\mathtt{t_0}$ is a value then case 8 of the canonical forms lemma assures us that it is a datatype value of form $\langle \mathtt{C_j = v_j} \rangle$ and so by E-CASE-TY, \texttt{t $\longrightarrow$ [$\mathtt{x_j \mapsto v_j}$]$\mathtt{t_j}$}.
    \end{description}
\end{proof}
\end{multicols}

\subsection{Proof of Preservation}
\label{sec:preservation}
\begin{multicols}{2}
\begin{proof}
    By induction on a derivation of $\mathtt{t : T}$. At each step of the induction, we assume that the desired property holds for all subderivations (i.e. that if $\mathtt{s : S}$ and $\mathtt{s \longrightarrow s'}$, then $\mathtt{s' : S}$, whenever $\mathtt{s : S}$ is proved by a subderivation of the present one) and proceed by case analysis on the final rule in the derivation.
    \begin{description}
        \item[\textit{\textmd{Case}} \normalfont T-VAR:]\mbox{}\\
            \vspace{-3ex}
            \begin{quote}
                \texttt{t = x}\\
                \texttt{x : T $\in \mathtt{\Gamma}$}
            \end{quote}
            If the last rule in the derivation is T-VAR, then we know from the form of this rule that $\mathtt{t}$ must be a variable of type $\mathtt{T}$. Thus, $\mathtt{t}$ is a value, so it cannot be the case that $\mathtt{t \longrightarrow t'}$ for any $\mathtt{t'}$, and the requirements of the theorem are vacuously satisfied.
            
        \item[\textit{\textmd{Case}} \normalfont T-ABS:]\mbox{}\\
            \vspace{-3ex}
            \begin{quote}
                \texttt{t = $\mathtt{\lambda x: T_1.t_2}$}\\
            \end{quote}
            If the last rule in the derivation is T-ABS, then we know from the form of this rule that $\mathtt{t}$ must be an abstraction. Thus, $\mathtt{t}$ is a value, so it cannot be the case that $\mathtt{t \longrightarrow t'}$ for any $\mathtt{t'}$, and the requirements of the theorem are vacuously satisfied.
            
        \item[\textit{\textmd{Case}} \normalfont T-APP:]\mbox{}\\
            \vspace{-3ex}
            \begin{quote}
                \texttt{t = $\mathtt{t_1}$ $\mathtt{t_2}$}\\
                \texttt{$\mathtt{\Gamma \vdash t_1 : T_{11} \rightarrow T_{12}}$}\\
                \texttt{$\mathtt{\Gamma \vdash t_2 : T_{11}}$}\\
                \texttt{T = $\mathtt{T_{12}}$}
            \end{quote}
            Looking at the evaluation rules with applications on the left-hand side, we find that there are three rules by which $\mathtt{t \longrightarrow t'}$ can be derived: E-APP1, E-APP2, and E-APPABS. We consider each case separately.
            \vspace{-3ex}
            \begin{quote}
                \item[\textit{\textmd{Subcase}} \normalfont E-APP1:]\mbox{}\\
                \vspace{-3ex}
                \begin{quote}
                    $\mathtt{t_1 \longrightarrow t_1'}$\\
                    \texttt{$\mathtt{t'}$ = $\mathtt{t_1'}$ $\mathtt{t_2}$}
                \end{quote}
                From the assumptions of the T-APP case, we have a subderivation of the original typing derivation whose conclusion is $\mathtt{\Gamma \vdash t_1 : T_{11} \rightarrow T_{12}}$. We can apply the induction hypothesis to this subderivation obtaining $\mathtt{\Gamma \vdash t_1' : T_{11} \rightarrow T_{12}}$. Combining this with the fact that $\mathtt{\Gamma \vdash t_2 : T_{11}}$, we can apply rule T-APP to conclude that $\mathtt{\Gamma \vdash t' : T}$.
                
                \item[\textit{\textmd{Subcase}} \normalfont E-APP2:]\mbox{}\\
                \mbox{}\\
                Similar to E-APP1.
                
                \item[\textit{\textmd{Subcase}} \normalfont E-APPABS:]\mbox{}\\
                \vspace{-3ex}
                \begin{quote}
                    \texttt{$\mathtt{t_1}$ = $\mathtt{\lambda x:T_{11}.t_{12}}$}
                    \texttt{$\mathtt{t_2}$ = $\mathtt{v_2}$}
                    \texttt{$\mathtt{t'}$ = [$\mathtt{x \mapsto v_2}$]$\mathtt{t_{12}}$}
                \end{quote}
                Using the inversion lemma, we can deconstruct the typing derivation for $\mathtt{\lambda x:T_{11}.t_{12}}$ yielding $\mathtt{\Gamma, x : T_{11} \vdash t_{12} : T_{12}}$. From this we obtain $\mathtt{\Gamma \vdash t' : T_{12}}$.
            \end{quote}
            
        \item[\textit{\textmd{Case}} \normalfont T-INT:]\mbox{}\\
        \vspace{-3ex}
        \begin{quote}
            \texttt{t = $\langle integer \rangle$}\\
            \texttt{T = int}
        \end{quote}
        If the last rule in the derivation is T-INT, then we know from the form of this rule that $\mathtt{t}$ must be an integer value as described in Section \ref{sec:swvgrammar} and that $\mathtt{T}$ is \texttt{int}. Thus, $\mathtt{t}$ is a value, so it cannot be the case that $\mathtt{t \longrightarrow t'}$ for any $\mathtt{t'}$, and the requirements of the theorem are vacuously satisfied.
        
        \item[\textit{\textmd{Case}} \normalfont T-REAL, T-STRING, T-BIT, T-NIL:]\mbox{}\\
        \mbox{}\\
        Similar to T-INT.
        
        \item[\textit{\textmd{Case}} \normalfont T-INT-NEG:]\mbox{}\\
        \vspace{-3ex}
        \begin{quote}
            \texttt{t = \textasciitilde$\mathtt{t_1}$}\\
            \texttt{T = int}
        \end{quote}
        If the last rule in the derivation is T-INT-NEG, then we know from the form of this rule that $\mathtt{t}$ must have the form \texttt{\textasciitilde$\mathtt{t_1}$} and that $\mathtt{T}$ is \texttt{int}. We must further have a subderivation with conclusion $\mathtt{t_1 : int}$. Looking at the evaluation rules with integer negation on the left-hand side, there is one rule by which $\mathtt{t \longrightarrow t'}$ can be derived: E-NEG.
        \vspace{-3ex}
        \begin{quote}
            \item[\textit{\textmd{Subcase}} \normalfont E-NEG:]\mbox{}\\
            \vspace{-3ex}
            \begin{quote}
                $\mathtt{t_1 \longrightarrow t_1'}$\\
                \texttt{$\mathtt{t'}$ = \textasciitilde$\mathtt{t_1'}$}
            \end{quote}
            From the assumptions of the T-INT-NEG case, we have a subderivation of the original typing derivation with conclusion $\mathtt{t_1 : int}$. We may apply the induction hypothesis to this subderivation, obtaining $\mathtt{t_1' : int}$. Then, by T-INT-NEG, \texttt{\textasciitilde$\mathtt{t_1' : int}$} and so $\mathtt{t' : int}$.
        \end{quote}
        
        \item[\textit{\textmd{Case}} \normalfont T-INT-ADD:]\mbox{}\\
        \vspace{-3ex}
        \begin{quote}
            \texttt{t = $\mathtt{t_1}$ + $\mathtt{t_2}$}\\
            \texttt{T = int}
        \end{quote}
        If the last rule in the derivation is T-INT-ADD, then we know from the form of this rule that $\mathtt{t}$ must have the form \texttt{$\mathtt{t_1}$ + $\mathtt{t_2}$} and that $\mathtt{T}$ is \texttt{int}. We must further have subderivations with conclusions $\mathtt{t_1 : int}$ and $\mathtt{t_2 : int}$. Looking at the evaluation rules with integer addition on the left-hand side, there are two rules by which $\mathtt{t \longrightarrow t'}$ can be derived: E-INT-ADD1 and E-INT-ADD2.
        \vspace{-3ex}
        \begin{quote}
            \item[\textit{\textmd{Subcase}} \normalfont E-INT-ADD1:]\mbox{}\\
            \vspace{-3ex}
            \begin{quote}
                $\mathtt{t_1 \longrightarrow t_1'}$\\
                \texttt{$\mathtt{t'}$ = $\mathtt{t_1'}$ + $\mathtt{t_2}$}
            \end{quote}
            From the assumptions of the T-INT-ADD case, we have a subderivation of the original typing derivation with conclusion $\mathtt{t_1 : int}$. We may apply the induction hypothesis to this subderivation, obtaining $\mathtt{t_1' : int}$. We also have a subderivation with conclusion $\mathtt{t_2 : int}$. Thus, by T-INT-ADD, \texttt{$\mathtt{t_1'}$ + $\mathtt{t_2 : int}$} and so $\mathtt{t' : int}$.
            
            \item[\textit{\textmd{Subcase}} \normalfont E-INT-ADD2:]\mbox{}\\
            \vspace{-3ex}
            \begin{quote}
                $\mathtt{t_2 \longrightarrow t_2'}$\\
                \texttt{$\mathtt{t'}$ = $\mathtt{v_1}$ + $\mathtt{t_2'}$}
            \end{quote}
            From the assumptions of the T-INT-ADD case, we have a subderivation of the original typing derivation with conclusion $\mathtt{t_2 : int}$. We may apply the induction hypothesis to this subderivation, obtaining $\mathtt{t_2' : int}$. We also have a subderivation with conclusion $\mathtt{t_1 : int}$ and since $\mathtt{v_1}$ is the value form denotation of $\mathtt{t_1}$, it also holds that $\mathtt{v_1 : int}$. Thus, by T-INT-ADD, \texttt{$\mathtt{v_1}$ + $\mathtt{t_2' : int}$} and so $\mathtt{t' : int}$.
        \end{quote}
        
        \item[\textit{\textmd{Case}} \normalfont T-INT-(SUB/MUL/DIV/MOD):]\mbox{}\\
        \mbox{}\\
        Similar to T-INT-ADD.
        
        \item[\textit{\textmd{Case}} \normalfont T-REAL-(ADD/SUB/MUL/DIV):]\mbox{}\\
        \mbox{}\\
        Similar to T-INT-ADD.
        
        \item[\textit{\textmd{Case}} \normalfont T-BIT-NEG:]\mbox{}\\
        \vspace{-3ex}
        \begin{quote}
            \texttt{t = !$\mathtt{t_1}$}\\
            \texttt{T = $\mathtt{T_H}$}
        \end{quote}
        If the last rule in the derivation is T-BIT-NEG, then we know from the form of this rule that $\mathtt{t}$ must have the form \texttt{!$\mathtt{t_1}$} and that $\mathtt{T}$ is $\mathtt{T_H}$. Looking at the evaluation rules there is only one rule by which $\mathtt{t} \longrightarrow \mathtt{t'}$ can be derived: E-BIT-NEG1. We may apply the induction hypothesis to our subderivation, obtaining $\mathtt{t_1' : T_H}$. Then, by T-BIT-NEG, \texttt{!$\mathtt{t_1' : T_H}$} and so $\mathtt{t' : T_H}$.
        
        \item[\textit{\textmd{Case}} \normalfont T-AND:]\mbox{}\\
        \vspace{-3ex}
        \begin{quote}
            \texttt{t = $\mathtt{t_1}$ \& $\mathtt{t_2}$}\\
            \texttt{T = $\mathtt{T_H}$}
        \end{quote}
        If the last rule in the derivation is T-AND, then we know from the form of this rule that $\mathtt{t}$ must have the form \texttt{$\mathtt{t_1}$ \& $\mathtt{t_2}$} and that $\mathtt{T}$ is $\mathtt{T_H}$. We must further have subderivations with conclusions $\mathtt{t_1 : T_H}$ and $\mathtt{t_2 : T_H}$. Looking at the evaluation rules with \texttt{\&} on the left-hand side, there are two rules by which $\mathtt{t \longrightarrow t'}$ can be derived: E-AND1 and E-AND2.
        \vspace{-3ex}
        \begin{quote}
            \item[\textit{\textmd{Subcase}} \normalfont E-AND1:]\mbox{}\\
            \vspace{-3ex}
            \begin{quote}
                $\mathtt{t_1 \longrightarrow t_1'}$\\
                \texttt{$\mathtt{t'}$ = $\mathtt{t_1'}$ \& $\mathtt{t_2}$}
            \end{quote}
            From the assumptions of the T-AND case, we have a subderivation of the original typing derivation with conclusion $\mathtt{t_1 : T_H}$. We may apply the induction hypothesis to this subderivation, obtaining $\mathtt{t_1' : T_H}$. We also have a subderivation with conclusion $\mathtt{t_2 : T_H}$. Thus, by T-AND, \texttt{$\mathtt{t_1'}$ \& $\mathtt{t_2 : T_H}$} and so $\mathtt{t' : T_H}$.
            
            \item[\textit{\textmd{Subcase}} \normalfont E-AND2:]\mbox{}\\
            \vspace{-3ex}
            \begin{quote}
                $\mathtt{t_2 \longrightarrow t_2'}$\\
                \texttt{$\mathtt{t'}$ = $\mathtt{v_1}$ \& $\mathtt{t_2'}$}
            \end{quote}
            From the assumptions of the T-AND case, we have a subderivation of the original typing derivation with conclusion $\mathtt{t_2 : T_H}$. We may apply the induction hypothesis to this subderivation, obtaining $\mathtt{t_2' : T_H}$. We also have a subderivation with conclusion $\mathtt{t_1 : T_H}$ and since $\mathtt{v_1}$ is the value form denotation of $\mathtt{t_1}$, it also holds that $\mathtt{v_1 : T_H}$. Thus, by T-INT-ADD, \texttt{$\mathtt{v_1}$ \& $\mathtt{t_2' : T_H}$} and so $\mathtt{t' : T_H}$.
        \end{quote}
        
        \item[\textit{\textmd{Case}} \normalfont T-(OR/XOR):]\mbox{}\\
        \mbox{}\\
        Similar to T-AND.
        
        \item[\textit{\textmd{Case}} \normalfont T-AND-RED:]\mbox{}\\
        \vspace{-3ex}
        \begin{quote}
            \texttt{t = \&->$\mathtt{t_1}$}\\
            \texttt{T = $\mathtt{bit}$[n]}
        \end{quote}
        If the last rule in the derivation is T-AND-RED, then we know from the form of this rule that $\mathtt{t}$ must have the form \texttt{\&->$\mathtt{t_1}$} and that $\mathtt{T}$ is \texttt{bit[n]}. Looking at the evaluation rules there is only one rule by which $\mathtt{t} \longrightarrow \mathtt{t'}$ can be derived: E-AND-RED1. We may apply the induction hypothesis to our subderivation, obtaining $\mathtt{t_1' : T_H}$. Then, by T-AND-RED, \texttt{\&->$\mathtt{t_1' : bit}$[n]} and so \texttt{$\mathtt{t' : bit}$[n]}.
        
        \item[\textit{\textmd{Case}} \normalfont T-(OR/XOR)-RED:]\mbox{}\\
        \mbox{}\\
        Similar to T-AND-RED.
        
         \item[\textit{\textmd{Case}} \normalfont T-SLL:]\mbox{}\\
        \vspace{-3ex}
        \begin{quote}
            \texttt{t = $\mathtt{t_1}$ <}\texttt{< $\mathtt{t_2}$}\\
            \texttt{T = bit[n]}
        \end{quote}
        If the last rule in the derivation is T-SLL, then we know from the form of this rule that $\mathtt{t}$ must have the form \texttt{$\mathtt{t_1}$ <}\texttt{< $\mathtt{t_2}$} and that $\mathtt{T}$ is \texttt{bit[n]}. We must further have subderivations with conclusions \texttt{$\mathtt{t_1 : bit}$[n]} and \texttt{$\mathtt{t_2 : bit}$[m]}. Looking at the evaluation rules with left-shifting on the left-hand side, there are two rules by which $\mathtt{t \longrightarrow t'}$ can be derived: E-SLL1 and E-SLL2.
        \vspace{-3ex}
        \begin{quote}
            \item[\textit{\textmd{Subcase}} \normalfont E-SLL1:]\mbox{}\\
            \vspace{-3ex}
            \begin{quote}
                $\mathtt{t_1 \longrightarrow t_1'}$\\
                \texttt{$\mathtt{t'}$ = $\mathtt{t_1'}$ <}\texttt{< $\mathtt{t_2}$}
            \end{quote}
            From the assumptions of the T-SLL case, we have a subderivation of the original typing derivation with conclusion \texttt{$\mathtt{t_1 : bit}$[n]}. We may apply the induction hypothesis to this subderivation, obtaining \texttt{$\mathtt{t_1' : bit}$[n]}. We also have a subderivation with conclusion \texttt{$\mathtt{t_2 : bit}$[m]}. Thus, by T-SLL, \texttt{$\mathtt{t_1'}$ <}\texttt{< $\mathtt{t_2 : bit}$[n]} and so \texttt{$\mathtt{t' : bit}$[n]}.
            
            \item[\textit{\textmd{Subcase}} \normalfont E-SLL2:]\mbox{}\\
            \vspace{-3ex}
            \begin{quote}
                $\mathtt{t_2 \longrightarrow t_2'}$\\
                \texttt{$\mathtt{t'}$ = $\mathtt{v_1}$ <}\texttt{< $\mathtt{t_2'}$}
            \end{quote}
            From the assumptions of the T-SLL case, we have a subderivation of the original typing derivation with conclusion \texttt{$\mathtt{t_2 : bit}$[m]}. We may apply the induction hypothesis to this subderivation, obtaining \texttt{$\mathtt{t_2' : bit}$[m]}. We also have a subderivation with conclusion \texttt{$\mathtt{t_1 : bit}$[n]} and since $\mathtt{v_1}$ is the value form denotation of $\mathtt{t_1}$, it also holds that \texttt{$\mathtt{v_1 : bit}$[n]}. Thus, by T-SLL, \texttt{$\mathtt{v_1}$ <}\texttt{< $\mathtt{t_2' : bit}$[n]} and so \texttt{$\mathtt{t' : bit}$[n]}.
        \end{quote}
        
        \item[\textit{\textmd{Case}} \normalfont T-(SRL/SRA):]\mbox{}\\
        \mbox{}\\
        Similar to T-SLL.
        
        \item[\textit{\textmd{Case}} \normalfont T-EQ:]\mbox{}\\
        \vspace{-3ex}
        \begin{quote}
            \texttt{t = $\mathtt{t_1}$ = $\mathtt{t_2}$}\\
            \texttt{T = int}
        \end{quote}
        If the last rule in the derivation is T-EQ, then we know from the form of this rule that $\mathtt{t}$ must have the form \texttt{$\mathtt{t_1}$ = $\mathtt{t_2}$} and that $\mathtt{T}$ is \texttt{int}. We must further have subderivations with conclusions $\mathtt{t_1 : T_S}$ and $\mathtt{t_2 : T_S}$ for some $\mathtt{T_S}$. Looking at the evaluation rules with equality on the left-hand side, there are two rules by which $\mathtt{t \longrightarrow t'}$ can be derived: E-EQ1 and E-EQ2.
        \vspace{-3ex}
        \begin{quote}
            \item[\textit{\textmd{Subcase}} \normalfont E-EQ1:]\mbox{}\\
            \vspace{-3ex}
            \begin{quote}
                $\mathtt{t_1 \longrightarrow t_1'}$\\
                \texttt{$\mathtt{t'}$ = $\mathtt{t_1'}$ = $\mathtt{t_2}$}
            \end{quote}
            From the assumptions of the T-EQ case, we have a subderivation of the original typing derivation with conclusion $\mathtt{t_1 : T_S}$. We may apply the induction hypothesis to this subderivation, obtaining $\mathtt{t_1' : T_S}$. We also have a subderivation with conclusion $\mathtt{t_2 : T_S}$. Thus, by T-EQ, \texttt{$\mathtt{t_1'}$ = $\mathtt{t_2 : int}$} and so $\mathtt{t' : int}$.
            
            \item[\textit{\textmd{Subcase}} \normalfont E-INT-ADD2:]\mbox{}\\
            \vspace{-3ex}
            \begin{quote}
                $\mathtt{t_2 \longrightarrow t_2'}$\\
                \texttt{$\mathtt{t'}$ = $\mathtt{v_1}$ = $\mathtt{t_2'}$}
            \end{quote}
            From the assumptions of the T-EQ case, we have a subderivation of the original typing derivation with conclusion $\mathtt{t_2 : T_S}$. We may apply the induction hypothesis to this subderivation, obtaining $\mathtt{t_2' : T_S}$. We also have a subderivation with conclusion $\mathtt{t_1 : T_S}$ and since $\mathtt{v_1}$ is the value form denotation of $\mathtt{t_1}$, it also holds that $\mathtt{v_1 : T_S}$. Thus, by T-EQ, \texttt{$\mathtt{v_1}$ = $\mathtt{t_2' : int}$} and so $\mathtt{t' : int}$.
        \end{quote}
        
        \item[\textit{\textmd{Case}} \normalfont T-(NEQ/LT/LEQ/GT/GEQ):]\mbox{}\\
        \mbox{}\\
        Similar to T-EQ.
        
        \item[\textit{\textmd{Case}} \normalfont T-CONS:]\mbox{}\\
        \vspace{-3ex}
        \begin{quote}
            \texttt{t = $\mathtt{t_1}$::$\mathtt{t_2}$}\\
            \texttt{T = $\mathtt{T_S}$ list}
        \end{quote}
        If the last rule in the derivation is T-CONS, then we know from the form of this rule that $\mathtt{t}$ must have the form \texttt{$\mathtt{t_1}$::$\mathtt{t_2}$} and that $\mathtt{T}$ is \texttt{$\mathtt{T_S}$ list}. We must further have subderivations with conclusions $\mathtt{t_1 : T_S}$ and \texttt{$\mathtt{t_2 : T_S}$ list} for some $\mathtt{T_S}$. Looking at the evaluation rules with the $\mathtt{cons}$ operator on the left-hand side, there are two rules by which $\mathtt{t \longrightarrow t'}$ can be derived: E-CONS1 and E-CONS2.
        \vspace{-3ex}
        \begin{quote}
            \item[\textit{\textmd{Subcase}} \normalfont E-CONS1:]\mbox{}\\
            \vspace{-3ex}
            \begin{quote}
                $\mathtt{t_1 \longrightarrow t_1'}$\\
                \texttt{$\mathtt{t'}$ = $\mathtt{t_1'}$::$\mathtt{t_2}$}
            \end{quote}
            From the assumptions of the T-CONS case, we have a subderivation of the original typing derivation with conclusion $\mathtt{t_1 : T_S}$. We may apply the induction hypothesis to this subderivation, obtaining $\mathtt{t_1' : T_S}$. We also have a subderivation with conclusion \texttt{$\mathtt{t_2 : T_S}$ list}. Thus, by T-CONS, \texttt{$\mathtt{t_1'}$::$\mathtt{t_2 : T_S}$ list} and so \texttt{$\mathtt{t' : T_S}$ list}.
            
            \item[\textit{\textmd{Subcase}} \normalfont E-CONS2:]\mbox{}\\
            \vspace{-3ex}
            \begin{quote}
                $\mathtt{t_2 \longrightarrow t_2'}$\\
                \texttt{$\mathtt{t'}$ = $\mathtt{v_1}$:: $\mathtt{t_2'}$}
            \end{quote}
            From the assumptions of the T-CONS case, we have a subderivation of the original typing derivation with conclusion \texttt{$\mathtt{t_2 : T_S}$ list}. We may apply the induction hypothesis to this subderivation, obtaining \texttt{$\mathtt{t_2' : T_S}$ list}. We also have a subderivation with conclusion $\mathtt{t_1 : T_S}$ and since $\mathtt{v_1}$ is the value form denotation of $\mathtt{t_1}$, it also holds that $\mathtt{v_1 : T_S}$. Thus, by T-CONS, \texttt{$\mathtt{v_1}$::$\mathtt{t_2' : T_S}$ list} and so \texttt{$\mathtt{t' : T_S}$ list}.
        \end{quote}
        
        \item[\textit{\textmd{Case}} \normalfont T-IFELSE:]\mbox{}\\
        \vspace{-3ex}
        \begin{quote}
            \texttt{t = if $\mathtt{t_1}$ then $\mathtt{t_2}$ else $\mathtt{t_3}$}\\
            \texttt{T = $\mathtt{T_0}$}
        \end{quote}
        If the last rule in the derivation is T-IFELSE, then we know from the form of this rule that $\mathtt{t}$ must have the form \texttt{if $\mathtt{t_1}$ then $\mathtt{t_2}$ else $\mathtt{t_3}$} and that $\mathtt{T}$ is $\mathtt{T_0}$ for some $\mathtt{T_0}$. We must further have subderivations with conclusions $\mathtt{t_1 : int}$, $\mathtt{t_2 : T_0}$ and $\mathtt{t_3 : T_0}$. Looking at the evaluation rules with a conditional on the left-hand side, there are three rules by which $\mathtt{t \longrightarrow t'}$ can be derived: E-IFELSE, E-IFELSE-T, and E-IFELSE-F.
        \vspace{-3ex}
        \begin{quote}
            \item[\textit{\textmd{Subcase}} \normalfont E-IFELSE:]\mbox{}\\
            \vspace{-3ex}
            \begin{quote}
                $\mathtt{t_1 \longrightarrow t_1'}$\\
                \texttt{$\mathtt{t'}$ = if $\mathtt{t_1'}$ then $\mathtt{t_2}$ else $\mathtt{t_3}$}
            \end{quote}
            From the assumptions of the T-IFELSE case, we have a subderivation of the original typing derivation with conclusion $\mathtt{t_1 : int}$. We may apply the induction hypothesis to this subderivation, obtaining $\mathtt{t_1' : int}$. We also have subderivations with conclusions $\mathtt{t_2 : T_0}$ and $\mathtt{t_3 : T_0}$. Thus, by T-IFELSE, \texttt{if $\mathtt{t_1'}$ then $\mathtt{t_2}$ else $\mathtt{t_3: T_0}$} and so $\mathtt{t' : T_0}$.
            
            \item[\textit{\textmd{Subcase}} \normalfont E-IFELSE-T:]\mbox{}\\
            \vspace{-3ex}
            \begin{quote}
                $\mathtt{t_1 \neq 0}$\\
                \texttt{$\mathtt{t'}$ = $\mathtt{t_2}$}
            \end{quote}
            If $\mathtt{t} \longrightarrow \mathtt{t'}$ is derived using E-IFELSE-T, then from the form of this rule we see that $\mathtt{t_1 \neq 0}$ and the resulting term $\mathtt{t'}$ is the second subexpression $\mathtt{t_2}$. This means we are finished, since we know (by the assumptions of the T-IFELSE case) that $\mathtt{t_2 : T_0}$, which is what we need.
            
            \item[\textit{\textmd{Subcase}} \normalfont E-IFELSE-F:]\mbox{}\\
            \mbox{}\\
            Similar to E-IFELSE-T.
        \end{quote}
        
        \item[\textit{\textmd{Case}} \normalfont T-REF:]\mbox{}\\
        \vspace{-3ex}
        \begin{quote}
            \texttt{t = ref $\mathtt{t_1}$ $|$ $\mu$}\\
            \texttt{T = $\mathtt{T_S}$ ref}
        \end{quote}
        If the last rule in the derivation is T-REF, then we know from the form of this rule that $\mathtt{t}$ must have the form \texttt{ref $\mathtt{t_1}$} and that $\mathtt{T}$ is \texttt{$\mathtt{T_S}$ ref} for some $\mathtt{T_S}$. We must further have a subderivation with conclusion $\mathtt{t_1 : T_S}$. Looking at the evaluation rules with $\mathtt{ref}$ on the left-hand side, there are two rules by which $\mathtt{t \longrightarrow t'}$ can be derived: E-REF, and E-REFV.
        \vspace{-3ex}
        \begin{quote}
            \item[\textit{\textmd{Subcase}} \normalfont E-REF:]\mbox{}\\
            \vspace{-3ex}
            \begin{quote}
                $\mathtt{t_1}$ $|$ $\mu \longrightarrow \mathtt{t_1'}$ $|$ $\mu'$\\
                \texttt{$\mathtt{t'}$ = ref $\mathtt{t_1'}$ $|$ $\mu'$}
            \end{quote}
            From the assumptions of the T-REF case, we have a subderivation of the original typing derivation with conclusion $\mathtt{t_1 : T_S}$. We may apply the induction hypothesis to this subderivation, obtaining $\mathtt{t_1' : T_S}$. Thus, by T-REF, \texttt{ref $\mathtt{t_1'}$ $|$ $\mu' : \mathtt{T_S}$ ref} and so \texttt{$\mathtt{t' : T_S}$ ref}.
            
            \item[\textit{\textmd{Subcase}} \normalfont E-REFV:]\mbox{}\\
            \vspace{-3ex}
            \begin{quote}
                \texttt{$\mathtt{t_1}$ = $\mathtt{v_1}$}\\
                \texttt{$\mathtt{t'}$ = $l$ $|$ $(\mu, l \mapsto \mathtt{v_1}$}
            \end{quote}
            If $\mathtt{t} \longrightarrow \mathtt{t'}$ is derived using E-REFV, then from the form of this rule we see that $\mathtt{t_1}$ = $\mathtt{v_1}$ and the resulting term $\mathtt{t'}$ is a location in the store $\mu$ augmented with a mapping between the location and $\mathtt{v_1}$. We know that locations in store $\mu$ are of type \texttt{$\mathtt{T}$ ref} where $\mathtt{T}$ is the type of the value to which it is mapped. In this case, that means that the location $l$ is of type \texttt{$\mathtt{T_S}$ ref} which is what we need.
        \end{quote}
        
        \item[\textit{\textmd{Case}} \normalfont T-ASSIGN:]\mbox{}\\
        \vspace{-3ex}
        \begin{quote}
            \texttt{t = $\mathtt{t_1}$ := $\mathtt{t_2}$}\\
            \texttt{T = unit}
        \end{quote}
        If the last rule in the derivation is T-ASSIGN, then we know from the form of this rule that $\mathtt{t}$ must have the form \texttt{$\mathtt{t_1}$ := $\mathtt{t_2}$} and that $\mathtt{T}$ is \texttt{unit}. We must further have a subderivation with conclusion \texttt{$\mathtt{t_1 : T_S}$ ref} and $\mathtt{t_2 : T_S}$ for some $\mathtt{T_S}$. Looking at the evaluation rules with the assignment operator on the left-hand side, there are three rules by which $\mathtt{t \longrightarrow t'}$ can be derived: E-ASSIGN1, E-ASSIGN2, and E-ASSIGN.
        \vspace{-3ex}
        \begin{quote}
            \item[\textit{\textmd{Subcase}} \normalfont E-ASSIGN1:]\mbox{}\\
            \vspace{-3ex}
            \begin{quote}
                $\mathtt{t_1}$ $|$ $\mu \longrightarrow \mathtt{t_1'}$ $|$ $\mu'$\\
                \texttt{$\mathtt{t'}$ = $\mathtt{t_1'}$ := $\mathtt{t_2}$ $|$ $\mu'$}
            \end{quote}
            From the assumptions of the T-ASSIGN case, we have subderivations of the original typing derivation with conclusions \texttt{$\mathtt{t_1 : T_S}$ ref} and $\mathtt{t_2 : T_S}$. We may apply the induction hypothesis to the first subderivation, obtaining \texttt{$\mathtt{t_1' : T_S}$ ref}. Thus, by T-ASSIGN, \texttt{$\mathtt{t_1'}$ := $\mathtt{t_2 : unit}$} and so \texttt{$\mathtt{t' : unit}$}.
            
            \item[\textit{\textmd{Subcase}} \normalfont E-ASSIGN2:]\mbox{}\\
            \vspace{-3ex}
            \begin{quote}
                $\mathtt{t_2}$ $|$ $\mu \longrightarrow \mathtt{t_2'}$ $|$ $\mu'$\\
                \texttt{$\mathtt{t'}$ = $\mathtt{v_1}$ := $\mathtt{t_2'}$ $|$ $\mu'$}
            \end{quote}
            From the assumptions of the T-ASSIGN case, we have subderivations of the original typing derivation with conclusions \texttt{$\mathtt{v_1 : T_S}$ ref} and $\mathtt{t_2 : T_S}$. We may apply the induction hypothesis to the second subderivation, obtaining \texttt{$\mathtt{t_2' : T_S}$}. Thus, by T-ASSIGN, \texttt{$\mathtt{v_1'}$ := $\mathtt{t_2' : unit}$} and so \texttt{$\mathtt{t' : unit}$}.
            
            \item[\textit{\textmd{Subcase}} \normalfont E-ASSIGN:]\mbox{}\\
            \vspace{-3ex}
            \begin{quote}
                \texttt{$\mathtt{t_1}$ = $l$}\\
                \texttt{$\mathtt{t_2}$ = $\mathtt{v_2}$}\\
                \texttt{$\mathtt{t'} = unit$ $|$ [$l \mapsto \mathtt{v_1}$]$\mu$}
            \end{quote}
            Immediate since $\mathtt{t' = unit}$.
        \end{quote}
        
        \item[\textit{\textmd{Case}} \normalfont T-DEREF:]\mbox{}\\
        \vspace{-3ex}
        \begin{quote}
            \texttt{t = \$$\mathtt{t_1}$}\\
            \texttt{T = $\mathtt{T_S}$}
        \end{quote}
        If the last rule in the derivation is T-DEREF, then we know from the form of this rule that $\mathtt{t}$ must have the form \texttt{\$$\mathtt{t_1}$} and that $\mathtt{T}$ is \texttt{$\mathtt{T_S}$}. We must further have a subderivation with conclusion \texttt{$\mathtt{t_1 : T_S}$ ref} for some $\mathtt{T_S}$. Looking at the evaluation rules with the dereferencing operator on the left-hand side, there are two rules by which $\mathtt{t \longrightarrow t'}$ can be derived: E-DEREF and E-DEREFLOC.
        \vspace{-3ex}
        \begin{quote}
            \item[\textit{\textmd{Subcase}} \normalfont E-DEREF:]\mbox{}\\
            \vspace{-3ex}
            \begin{quote}
                $\mathtt{t_1}$ $|$ $\mu \longrightarrow \mathtt{t_1'}$ $|$ $\mu'$\\
                \texttt{$\mathtt{t'}$ = \$$\mathtt{t_1'}$ $|$ $\mu'$}
            \end{quote}
            From the assumptions of the T-DEREF case, we have a subderivation of the original typing derivation with conclusion \texttt{$\mathtt{t_1 : T_S}$ ref}. We may apply the induction hypothesis, obtaining \texttt{$\mathtt{t_1' : T_S}$ ref}. Thus, by T-DEREF, \texttt{\$$\mathtt{t_1' : T_S}$} and so \texttt{$\mathtt{t' : T_S}$}.
            
            \item[\textit{\textmd{Subcase}} \normalfont E-DEREFLOC:]\mbox{}\\
            \vspace{-3ex}
            \begin{quote}
                \texttt{$\mathtt{t_1 = }l$ $|$ $\mu$}\\
                \texttt{$\mathtt{t'}$ = \$$l$ $|$ $\mu$}
            \end{quote}
            From the assumptions of the T-DEREF case, we have a subderivation of the original typing derivation with conclusion \texttt{$\mathtt{t_1 : T_S}$ ref}. In this subcase we have that $\mathtt{t_1}$ is the location value $l$. We know that $l$ has type \texttt{$\mathtt{T_S}$ ref}. Thus, by T-ASSIGN, \texttt{\$$l \mathtt{ : T_S}$} and so \texttt{$\mathtt{t' : T_S}$}.
        \end{quote}
        
        \item[\textit{\textmd{Case}} \normalfont T-SW:]\mbox{}\\
        \vspace{-3ex}
        \begin{quote}
            \texttt{t = sw $\mathtt{t_1}$}\\
            \texttt{T = $\mathtt{T_H}$ sw}
        \end{quote}
        If the last rule in the derivation is T-SW, then we know from the form of this rule that $\mathtt{t}$ must have the form \texttt{sw $\mathtt{t_1}$} and that $\mathtt{T}$ is \texttt{$\mathtt{T_H}$ sw} for some $\mathtt{T_H}$. We must further have a subderivation with conclusion $\mathtt{t_1 : T_H}$. Looking at the evaluation rules with $\mathtt{sw}$ on the left-hand side, there is one rule by which $\mathtt{t \longrightarrow t'}$ can be derived: E-SW. By the induction hypothesis, $\mathtt{t_1' : T_H}$ and so by T-SW, \texttt{sw $\mathtt{t_1' : T_H}$ sw}. Therefore, \texttt{$\mathtt{t' : T_H}$ sw} as needed.
        
        \item[\textit{\textmd{Case}} \normalfont T-UNSW:]\mbox{}\\
        \vspace{-3ex}
        \begin{quote}
            \texttt{t = unsw $\mathtt{t_1}$}\\
            \texttt{T = $\mathtt{T_H}$}
        \end{quote}
        If the last rule in the derivation is T-UNSW, then we know from the form of this rule that $\mathtt{t}$ must have the form \texttt{unsw $\mathtt{t_1}$} and that $\mathtt{T}$ is \texttt{$\mathtt{T_H}$} for some $\mathtt{T_H}$. We must further have a subderivation with conclusion \texttt{$\mathtt{t_1 : T_H}$ sw}. Looking at the evaluation rules with $\mathtt{hw}$ on the left-hand side, there is one rule by which $\mathtt{t \longrightarrow t'}$ can be derived: E-UNSW. By the induction hypothesis, \texttt{$\mathtt{t_1' : T_H}$ sw} and so by T-SW, \texttt{unsw $\mathtt{t_1' : T_H}$}. Therefore, \texttt{$\mathtt{t' : T_H}$} as needed.
        
        \item[\textit{\textmd{Case}} \normalfont T-ARR-ACC:]\mbox{}\\
        \vspace{-3ex}
        \begin{quote}
            \texttt{t = $\mathtt{t_1}$[$\mathtt{t_2}$]}\\
            \texttt{T = $\mathtt{T_H}$}
        \end{quote}
        If the last rule in the derivation is T-ARR-ACC, then we know from the form of this rule that $\mathtt{t}$ must have the form \texttt{$\mathtt{t_1}$[$\mathtt{t_2}$]} and that $\mathtt{T}$ is \texttt{$\mathtt{T_H}$}. We must further have a subderivation with conclusion \texttt{$\mathtt{t_1 : T_H}$[n]} and $\mathtt{t_2 : int}$. Looking at the evaluation rules with array access on the left-hand side, there are two rules by which $\mathtt{t \longrightarrow t'}$ can be derived: E-ARR-ACC and E-ARR-ACC1.
        \vspace{-3ex}
        \begin{quote}
            \item[\textit{\textmd{Subcase}} \normalfont E-ARR-ACC:]\mbox{}\\
            \vspace{-3ex}
            \begin{quote}
                \texttt{$\mathtt{t_1}$ = $\mathtt{v_1}$}\\
                \texttt{$\mathtt{t_2}$ = $\mathtt{v_2}$}\\
                \texttt{$\mathtt{t'}$ = $\mathtt{v_{1, v_2}}$}
            \end{quote}
            From the assumptions of the T-ARR-ACC case, we have subderivations of the original typing derivation with conclusions \texttt{$\mathtt{t_1 : T_H}$[n]} and $\mathtt{t_2 : int}$. Thus, the $\mathtt{v_2}^{th}$ element of $\mathtt{v_1}$ is of type $\mathtt{T_H}$, and so $\mathtt{t' : T_H}$.
            
            \item[\textit{\textmd{Subcase}} \normalfont E-ARR-ACC1:]\mbox{}\\
            \vspace{-3ex}
            \begin{quote}
                \texttt{$\mathtt{t_1} \longrightarrow \mathtt{t_1'}$}\\
                \texttt{$\mathtt{t'}$ = $\mathtt{v_1}$[$\mathtt{t_2'}$]}
            \end{quote}
            From the assumptions of the T-ARR-ACC case, we have subderivations of the original typing derivation with conclusions \texttt{$\mathtt{t_1 : T_H}$[n]} and $\mathtt{t_2 : int}$. By the induction hypothesis, $\mathtt{t_2' : int}$. Therefore, by T-ARR-ACC, \texttt{$\mathtt{v_1}$[$\mathtt{t_2'}$] $\mathtt{: T_H}$} and so $\mathtt{t' : T_H}$.
        \end{quote}

        \item[\textit{\textmd{Case}} \normalfont T-RCD:]\mbox{}\\
        \vspace{-3ex}
        \begin{quote}
            \texttt{t = \{$\mathtt{l_i = t_i^{i \in 1..n}}$\}}\\
            \texttt{T = \{$\mathtt{l_i : T_i^{i \in 1..n}}$\}}
        \end{quote}
        If the last rule in the derivation is T-RCD, then we know from the form of this rule that $\mathtt{t}$ must have the form \texttt{\{$\mathtt{l_i = t_i^{i \in 1..n}}$\}} and that $\mathtt{T}$ is \texttt{\{$\mathtt{l_i : T_i^{i \in 1..n}}$\}}. We must further have a subderivation for each $i$ with conclusion \texttt{$\mathtt{t_i : T_i}$}. Looking at the evaluation rules with record creation on the left-hand side, there are two rules by which $\mathtt{t \longrightarrow t'}$ can be derived: E-RCD and E-RCDV.
        \vspace{-3ex}
        \begin{quote}
            \item[\textit{\textmd{Subcase}} \normalfont E-RCD:]\mbox{}\\
            \vspace{-3ex}
            \begin{quote}
                \texttt{$\mathtt{t_j} \longrightarrow \mathtt{t_j'}$}\\
                \texttt{$\mathtt{t'}$ = \{$\mathtt{l_i = v_i^{i \in 1..j-1}, l_j = t_j', l_k = t_k^{k \in j+1..n}}$\}}
            \end{quote}
            From the assumptions of the T-RCD case, we have for each $i$ a subderivation of the original typing derivation with conclusion \texttt{$\mathtt{t_i : T_i}$}. Thus, $\mathtt{t_j : T_j}$. We may apply the induction hypothesis, obtaining $\mathtt{t_j' : T_j}$. Thus, by T-RCD, we obtain \texttt{\{$\mathtt{l_i = v_i^{i \in 1..j-1}, l_j = t_j', l_k = t_k^{k \in j+1..n} : \{l_i : T_i^{i \in 1..n}\}}$} and so \texttt{$\mathtt{t' : \{l_i : T_i^{i \in 1..n}\}}$}.
            
            \item[\textit{\textmd{Subcase}} \normalfont E-RCDV:]\mbox{}\\
            \vspace{-3ex}
            \begin{quote}
                \texttt{$\mathtt{t_j}$ = $\mathtt{v_j}$}\\
                \texttt{$\mathtt{t'}$ = \{$\mathtt{l_i = v_i^{i \in 1..j}, l_k = t_k^{k \in j+1..n}}$\}}
            \end{quote}
            If $\mathtt{t} \longrightarrow \mathtt{t'}$ is derived using E-RCDV, then from the form of this rule we see that $\mathtt{t_j = v_j}$ and the resulting term $\mathtt{t'}$ is an evaluation on the remaining fields of the record. This means we are finished, since we know (by the assumptions of the T-RCD case) that \texttt{$\mathtt{t' :}$\{$\mathtt{l_i = v_i^{i \in 1..j-1}, l_j = t_j', l_k = t_k^{k \in j+1..n}}$\}}, which is what we need.
        \end{quote}

        \item[\textit{\textmd{Case}} \normalfont T-PROJ:]\mbox{}\\
        \vspace{-3ex}
        \begin{quote}
            \texttt{t = \#$\mathtt{l_j}$ $\mathtt{t_1}$}\\
            \texttt{T = $\mathtt{T_j}$}
        \end{quote}
        If the last rule in the derivation is T-PROJ, then we know from the form of this rule that $\mathtt{t}$ must have the form \texttt{\#$\mathtt{l_j}$ $\mathtt{t_1}$} and that $\mathtt{T}$ is \texttt{$\mathtt{T_j}$}. We must further have a subderivation with conclusion \texttt{$\mathtt{t_1 : }$\{$\mathtt{l_i : T_i^{i \in 1..n}}$\}}. Looking at the evaluation rules with projection on the left-hand side, there are two rules by which $\mathtt{t \longrightarrow t'}$ can be derived: E-PROJ-RCD and E-PROJ.
        \vspace{-3ex}
        \begin{quote}
            \item[\textit{\textmd{Subcase}} \normalfont E-PROJ-RCD:]\mbox{}\\
            \vspace{-3ex}
            \begin{quote}
                \texttt{$\mathtt{t_1}$ = \#$\mathtt{l_j}$ \{$\mathtt{l_i = v_i^{i \in 1..n}}$\}}\\
                \texttt{$\mathtt{t'}$ = $\mathtt{v_j}$}
            \end{quote}
            From the assumptions of the T-PROJ case, we have a subderivation of the original typing derivation with conclusion \texttt{$\mathtt{t_1 : }$\{$\mathtt{l_i : T_i^{i \in 1..n}}$\}}. Thus extracting the value corresponding to field $\mathtt{l_j}$ must yield $\mathtt{v_j}$ which is of type $\mathtt{T_j}$, and so $\mathtt{t' : T_j}$.
            
            \item[\textit{\textmd{Subcase}} \normalfont E-PROJ:]\mbox{}\\
            \vspace{-3ex}
            \begin{quote}
                \texttt{$\mathtt{t_1} \longrightarrow \mathtt{t_1'}$}\\
                \texttt{$\mathtt{t'}$ = \#$\mathtt{l_j}$ $\mathtt{t_1'}$}
            \end{quote}
            From the assumptions of the T-PROJ case, we have a subderivation of the original typing derivation with conclusion \texttt{$\mathtt{t_1 : }$\{$\mathtt{l_i : T_i^{i \in 1..n}}$\}}. We may apply the induction hypothesis, obtaining \texttt{$\mathtt{t_1' : }$\{$\mathtt{l_i : T_i^{i \in 1..n}}$\}}. Thus, by T-PROJ, \texttt{\#$\mathtt{l_j}$ $\mathtt{t_1' : T_j}$} and so \texttt{$\mathtt{t' : T_j}$}.
        \end{quote}

        \item[\textit{\textmd{Case}} \normalfont T-LET:]\mbox{}\\
        \vspace{-3ex}
        \begin{quote}
            \texttt{t = let} $\mathtt{(x_i}$ \texttt{=} $\mathtt{t_i)^{i \in 1..n}}$ \texttt{in} $\mathtt{t_0}$ \texttt{end}\\
            \texttt{T = $\mathtt{T_0}$}
        \end{quote}
        If the last rule in the derivation is T-LET, then we know from the form of this rule that $\mathtt{t}$ must have the form \texttt{let} $\mathtt{(x_i}$ \texttt{=} $\mathtt{t_i)^{i \in 1..n}}$ \texttt{in} $\mathtt{t_0}$ \texttt{end} and that $\mathtt{T}$ is \texttt{$\mathtt{T_0}$}. We must further have subderivation with conclusions \texttt{$\mathtt{t_1 : T_1}$} and \texttt{let} $\mathtt{(x_i}$ \texttt{=} $\mathtt{t_i)^{i \in 2..n}}$ \texttt{in} $\mathtt{t_0}$ \texttt{end} $\mathtt{: T_0}$. Looking at the evaluation rules with let-bindings on the left-hand side, there are three rules by which $\mathtt{t \longrightarrow t'}$ can be derived: E-LET, E-LETV1, and E-LETV2.
        \vspace{-3ex}
        \begin{quote}
            \item[\textit{\textmd{Subcase}} \normalfont E-LET:]\mbox{}\\
            \vspace{-3ex}
            \begin{quote}
                \texttt{$\mathtt{t_1 \longrightarrow t_1'}$}\\
                \texttt{$\mathtt{t'}$ = let} $\mathtt{x_1 = t_1'}$ $\mathtt{(x_i}$ \texttt{=} $\mathtt{t_i)^{i \in 1..n}}$ \texttt{in} $\mathtt{t_0}$ \texttt{end}
            \end{quote}
            From the assumptions of the T-LET case, we have a subderivation of the original typing derivation with conclusion \texttt{$\mathtt{t_1 : T_1}$}. We may apply the induction hypothesis, obtaining $\mathtt{t_1' : T_1}$. Thus, by T-LET, \texttt{let} $\mathtt{x_1 = t_1'}$ $\mathtt{(x_i}$ \texttt{=} $\mathtt{t_i)^{i \in 1..n}}$ \texttt{in} $\mathtt{t_0}$ \texttt{end} $\mathtt{: T_0}$ and so $\mathtt{t' : T_0}$.
            
            \item[\textit{\textmd{Subcase}} \normalfont E-LETV1:]\mbox{}\\
            \vspace{-3ex}
            \begin{quote}
                \texttt{$\mathtt{t_1}$ = $\mathtt{v_1}$}\\
                \texttt{$\mathtt{t'}$ = let} $\mathtt{x_2 = t_2}$ $\mathtt{(x_i}$ \texttt{=} $\mathtt{t_i)^{i \in 3..n}}$ \texttt{in} $\mathtt{t_0}$ \texttt{end}
            \end{quote}
            If $\mathtt{t_1}$ is $\mathtt{v_1}$ we simply augment the variable and type binding stores with the mapping $\mathtt{x_1 \mapsto v_1}$. The term being evaluated in the body of the let-binding does not alter and as such $\mathtt{t' : T_0}$.
            
            \item[\textit{\textmd{Subcase}} \normalfont E-LETV2:]\mbox{}\\
            \mbox{}\\
            Similar to E-LETV1.
        \end{quote}

        \item[\textit{\textmd{Case}} \normalfont T-DATATY:]\mbox{}\\
        \vspace{-3ex}
        \begin{quote}
            \texttt{t = $\mathtt{C_j}$ $\mathtt{t_0}$}\\
            \texttt{T = $\langle \mathtt{C_i : T_i} \rangle ^{i \in 1..n}$}
        \end{quote}
        If the last rule in the derivation is T-DATATY, then we know from the form of this rule that $\mathtt{t}$ must have the form \texttt{$\mathtt{C_j}$ $\mathtt{t_0}$} and that $\mathtt{T}$ is $\langle \mathtt{C_i : T_i} \rangle ^{i \in 1..n}$. We must further have a subderivation with conclusion \texttt{$\mathtt{t_0 : T_j}$}. Looking at the evaluation rules with datatype construction on the left-hand side, there is one rule by which $\mathtt{t_0 \longrightarrow t_0'}$ which is E-DATATY. By the induction hypothesis, $\mathtt{t_0' : T_j}$. Then, by T-DATATY, it follows that $\mathtt{C_j}$ $\mathtt{t_0' :} \langle \mathtt{C_i : T_i} \rangle ^{i \in 1..n}$. Therefore, $\mathtt{t' : } \langle \mathtt{C_i : T_i} \rangle ^{i \in 1..n}$.

        \item[\textit{\textmd{Case}} \normalfont T-CASE:]\mbox{}\\
        \vspace{-3ex}
        \begin{quote}
            \texttt{t = case $\mathtt{t_0}$ of $\mathtt{C_i}$ $\mathtt{x_i}$ => $\mathtt{t_i}^{i \in 1..n}$}\\
            \texttt{T = $T'$}
        \end{quote}
        If the last rule in the derivation is T-CASE, then we know from the form of this rule that $\mathtt{t}$ must have the form \texttt{case $\mathtt{t_0}$ of $\mathtt{C_i}$ $\mathtt{x_i}$ => $\mathtt{t_i}^{i \in 1..n}$} and that $\mathtt{T}$ is $\mathtt{T'}$ for some $\mathtt{T'}$. We must further have subderivations with conclusions $\mathtt{t_0 : } \langle \mathtt{C_i : T_i} \rangle ^{i \in 1..n}$ and, for each $i$, $\mathtt{t_i : T'}$. Looking at the evaluation rules with $\mathtt{case}$ expressions on the left-hand side, there are two rules by which $\mathtt{t_0 \longrightarrow t_0'}$ can be derived: E-CASE and E-CASE-TY.
        \vspace{-3ex}
        \begin{quote}
            \item[\textit{\textmd{Subcase}} \normalfont E-CASE:]\mbox{}\\
            \vspace{-3ex}
            \begin{quote}
                \texttt{$\mathtt{t_0 \longrightarrow t_0'}$}\\
                \texttt{$\mathtt{t'}$ = case $\mathtt{t_0'}$ of $\mathtt{C_i}$ $\mathtt{x_i}$ => $\mathtt{t_i}^{i \in 1..n}$}
            \end{quote}
            From the assumptions of the T-CASE case, we have a subderivation of the original typing derivation with conclusion \texttt{$\mathtt{t_0 :} \langle \mathtt{C_i : T_i} \rangle ^{i \in 1..n}$}. We may apply the induction hypothesis, obtaining $\mathtt{t_0' :} \langle \mathtt{C_i : T_i} \rangle ^{i \in 1..n}$. Thus, by T-CASE, \texttt{case $\mathtt{t_0'}$ of $\mathtt{C_i}$ $\mathtt{x_i}$ => $\mathtt{t_i}^{i \in 1..n}$ $\mathtt{: T_0}$}, and so $\mathtt{t' : T'}$.
            
            \item[\textit{\textmd{Subcase}} \normalfont E-CASE-TY:]\mbox{}\\
            \vspace{-3ex}
            \begin{quote}
                \texttt{$\mathtt{t_0}$ = $\mathtt{C_j}$ $\mathtt{v_j}$}\\
                \texttt{$\mathtt{t'}$ = [$\mathtt{x_j \mapsto v_j}$] $\mathtt{t_j}$}
            \end{quote}
            As we assumed earlier, $\mathtt{t_i : T'}$ for all $\mathtt{t_i}$ and so $\mathtt{t_j : T'}$. As a result, \texttt{[$\mathtt{x_j \mapsto v_j}$] $\mathtt{t_j : T'}$}.
            
        \end{quote}
    \end{description}
\end{proof}
\end{multicols}

\section{Implementation}
\subsection{\texttt{Types.ty} datatype}
\label{sec:types-ty}
\begin{tcolorbox}
\begin{Verbatim}
datatype ty = H_TY of h_ty
            | S_TY of s_ty
            | M_TY of m_ty
            | META of tyvar
            | TOP
            | BOTTOM

and h_ty    = BIT
            | ARRAY of {ty: h_ty, size: int ref}
            | TEMPORAL of {ty: h_ty, time: int ref}
            | H_RECORD of (tyvar * h_ty) list
            | H_DATATYPE of (tyvar * h_ty option) list * unit ref
            | H_POLY of tyvar list * h_ty
            | H_META of tyvar
            | H_TOP
            | H_BOTTOM

and s_ty    = INT | REAL | STRING
            | ARROW of (s_ty * s_ty)
            | LIST of s_ty
            | SW of h_ty
            | S_RECORD of (tyvar * s_ty) list
            | REF of s_ty
            | S_DATATYPE of (tyvar * s_ty option) list * unit ref
            | S_MU of tyvar list * s_ty
            | S_POLY of tyvar list * s_ty
            | S_META of tyvar
            | S_TOP
            | S_BOTTOM

and m_ty    = MODULE of h_ty * h_ty
            | PARAMETERIZED_MODULE of s_ty * h_ty * h_ty
            | M_POLY of tyvar list * m_ty
            | M_BOTTOM
\end{Verbatim}
\end{tcolorbox}

\subsection{\texttt{Absyn.ty} datatype}
\label{sec:absyn-ty}
\begin{tcolorbox}
\begin{Verbatim}
datatype ty = NameTy of symbol * pos
            | ParameterizedTy of symbol * (ty list) * pos
            | TyVar of symbol * pos
            | SWRecordTy of field list * pos
            | HWRecordTy of field list * pos
            | ArrayTy of ty * int * pos
            | ListTy of ty * pos
            | TemporalTy of ty * int * pos
            | RefTy of ty * pos
            | SWTy of ty * pos
            | FunTy of ty * ty * pos
            | PlaceholderTy of unit ref
            | ExplicitTy of Types.ty
\end{Verbatim}
\end{tcolorbox}

\subsection{\texttt{Value.value} datatype}
\label{sec:value-value}
\begin{tcolorbox}
\begin{Verbatim}
datatype value
    = IntVal of int
    | StringVal of string
    | RealVal of real
    | ListVal of value list
    | RefVal of value ref
    | SWVal of value
    | RecordVal of (symbol * value) list
    | FunVal of (value -> value) ref
    | DatatypeVal of (symbol * unit ref * value)
    | NamedVal of symbol * Types.ty
    | BitVal of GeminiBit.bit
    | ArrayVal of value vector
    | HWRecordVal of (symbol * value) list
    | BinOpVal of {left: value, oper: binop, right: value}
    | UnOpVal of {value: value, oper: unop}
    | ArrayAccVal of {arr: value, index: int}
    | DFFVal of int
    | PreParamModuleVal of (value -> value -> value) * value
    | ModuleVal of (value -> value) * value
\end{Verbatim}
\end{tcolorbox}

\end{document}